\documentclass[backref=page]{article}
\usepackage{ijcai22}

\usepackage{times}
\usepackage{soul}
\usepackage{url}
\usepackage[hidelinks]{hyperref}
\usepackage[utf8]{inputenc}
\usepackage[small]{caption}
\usepackage{graphicx}
\usepackage{amsmath}
\usepackage{amsthm}
\usepackage{booktabs}
\title{Monotone-Value Neural Networks: Exploiting Preference Monotonicity in Combinatorial Assignment\footnotemark[1]}

\author{
Jakob Weissteiner$^{1,3}$\footnotemark[2]
\and
Jakob Heiss$^{2,3}$\footnotemark[2]\and
Julien Siems$^1$\footnotemark[2]\And
Sven Seuken$^{1,3}$
\affiliations
$^1$University of Zurich\\
$^2$ETH Zurich\\
$^3$ETH AI Center\\
\emails
weissteiner@ifi.uzh.ch,
jakob.heiss@math.ethz.ch,
juliensiems@gmail.com,
seuken@ifi.uzh.ch
}

\newcommand{\boolHardCodedAppendixRef}{0}
\usepackage{our_commands_arXiv}

\begin{document}
\maketitle

\begin{abstract}
    Many important resource allocation problems involve the \emph{combinatorial assignment} of items, e.g., auctions or course allocation.
    Because the bundle space grows exponentially in the number of items, preference elicitation is a key challenge in these domains. Recently, researchers have proposed ML-based mechanisms that outperform traditional mechanisms while reducing preference elicitation costs for agents. However, one major shortcoming of the ML algorithms that were used is their disregard of important prior knowledge about agents' preferences. To address this, we introduce \emph{monotone-value neural networks (MVNNs)}, which are designed to capture combinatorial valuations, while enforcing \emph{monotonicity} and \emph{normality}. On a technical level, we prove that our MVNNs are universal in the class of monotone and normalized value functions, and we provide a mixed-integer linear program (MILP) formulation to make solving MVNN-based winner determination problems (WDPs) practically feasible. We evaluate our MVNNs experimentally in spectrum auction domains. Our results show that MVNNs improve the prediction performance, they yield state-of-the-art allocative efficiency in the auction, and they also reduce the run-time of the WDPs. Our code is available on GitHub: \url{https://github.com/marketdesignresearch/MVNN}.
\end{abstract}
\renewcommand{\thefootnote}{\fnsymbol{footnote}}
\footnotetext[1]{This paper is the slightly updated version of \cited{weissteiner2022monotone} published at IJCAI'22 including the appendix.}
\footnotetext[2]{These authors contributed equally to this paper.}
\renewcommand{\thefootnote}{\arabic{footnote}}

\section{Introduction}\label{sec:Introduction}
Many important economic problems involve the \emph{combinatorial assignment} of multiple indivisible items to multiple agents. In domains \textit{with money}, prominent examples include \emph{combinatorial auctions (CAs)} and \emph{combinatorial exchanges (CEs)}. In CAs, heterogeneous items are allocated amongst a set of bidders, e.g., for the sale of spectrum licenses \cite{cramton2013spectrumauctions}. In CEs, a set of items is allocated between multiple agents who can be sellers \emph{and} buyers at the same time, e.g., for the reallocation of catch shares \cite{bichler2019designing}. In domains \emph{without money}, a popular example is \emph{combinatorial course allocation}, where course seats are allocated to students at large business schools \cite{budish2011combinatorial}.

What all of these domains have in common is that the agents can report their values on \textit{bundles} of items rather than only on individual items. This allows them to express more complex preferences, i.e., an agent's valuation of a bundle is not simply the sum of each individual item's value, but it can be more (\textit{complementarity}) or less (\textit{substitutability}). However, since the bundle space grows exponentially in the number of items, agents cannot report values for all bundles in settings with more than a modest number of items. Thus, parsimonious \emph{preference elicitation} is key for the design of practical combinatorial assignment mechanisms. 

In this paper, we present a new machine learning approach that exploits prior (structural) knowledge about agents' preferences and can be integrated well into iterative market mechanisms. Our contribution applies to any combinatorial assignment problem. However, we present our algorithms in the context of a CA specifically, to simplify the notation and because there exist well-studied preference generators for CAs that we can use for our experimental evaluation.

For CAs with general valuations, \cited{nisan2006communication} have shown that exponential communication in the number of items is needed in the worst case to find an optimal allocation of items to bidders, i.e., to ensure full efficiency. Thus, for general valuations, practical CAs cannot provide efficiency guarantees in large domains. In practice, \textit{iterative combinatorial auctions (ICAs)} are employed, where the auctioneer interacts with bidders over multiple rounds, eliciting a \textit{limited} amount of information, aiming to find a highly efficient allocation.
ICAs are widely used, e.g., for the sale of licenses to build offshore wind farms \cite{ausubel2011auction}. The provision of spectrum licenses via the \emph{combinatorial clock auction (CCA) \cite{ausubel2006clock}} has generated more than \$20 billion in total revenue \cite{ausubel2017practical}. Therefore, increasing the efficiency of such real-world ICAs by only 1\% point translates into monetary gains of hundreds of millions of dollars.

\subsection{ML-based Auction Design}\label{subsec:Machine Learning-based CAs}
In recent years, researchers have successfully integrated machine learning (ML) algorithms into the design of CAs to improve their performance. 
\cited{dutting2019optimal} and \cited{rahme2020permutation} used neural networks (NNs) to learn entire auction mechanisms from data, following the automated mechanism design paradigm. \cited{brero2019fast} studied a Bayesian ICA using probabilistic price updates to achieve faster convergence to an efficient allocation.

Most related to this paper is the work by \citeauthor{brero2018combinatorial} \shortcite{brero2018combinatorial,brero2021workingpaper}, who developed a value-query-based ICA that achieves even higher efficiency than the widely used CCA.
In follow-up work, \cited{weissteiner2020deep} extended their work by integrating neural networks (NN) in their mechanisms and could further increase the efficiency. Finally, \cited{weissteiner2022fourier} used Fourier transforms (FTs) to leverage different notions of sparsity of value functions in preference elicitation. However, despite these advances, it remains a challenging problem to find the efficient allocation while keeping the elicitation cost for bidders low. Even state-of-the-art approaches suffer from significant efficiency losses, highlighting the need for better preference elicitation algorithms.

We show in this paper that these limitations can be partially explained due to the usage of poor, non-informative ML-algorithms, which either do not include important prior domain knowledge or make too restrictive assumptions about the bidders' value functions. \citeauthor{brero2018combinatorial} \shortcite{brero2018combinatorial,brero2021workingpaper} used support vector regressions (SVRs) with quadratic kernels which can only learn up to two way interactions between items and do not account for an important monotonicity property of bidders' value functions. While the fully-connected feed-forward NNs used by \cited{weissteiner2020deep} are more expressive, they also do not account for this monotonicity property. In particular when operating with only a small number of data points (which is the standard in market mechanisms, because preference elicitation is costly), this can cause significant efficiency losses.

Over the last decade, major successes in ML were made by specialized NN architectures (e.g., Convolutional Neural Networks) that incorporate domain-specific prior knowledge to improve generalization \cite{bronstein2017geometric}. With this paper, we follow the same paradigm by incorporating prior knowledge about monotone preferences into an NN architecture to improve generalization, which is key for a well-functioning preference elicitation algorithm.

Several other approaches for incorporating monotonicity into NNs have previously been proposed. However, for these architectures, it is not known how the NN-based winner determination problem (WDP) could be solved quickly or they have other limitations.
\cited{sill1998monotonic} proposes only a shallow architecture which violates the normalization property. 
\cited{you2017deep} propose a complicated non-standard architecture, where no computationally feasible MILP formulation of the corresponding WDP is known.
\cited{wehenkel2019unconstrained} implement monotonicity by representing the target function as an integral of an NN and 
thus the WDP would result in a computationally infeasible MILP.
\cited{liu2020certified} train NNs with successively higher regularization until a MILP based verification procedure guarantees monotonicity. The repeated retraining and verification lead to high computational cost.

In contrast, our approach is particularly well suited for combinatorial assignment, because (i) our NN-based WDP can be formulated as a succinct MILP and thus solved quickly in practice and (ii) we propose a generic fully-connected feed-forward architecture with arbitrary number of hidden layers which can be trained efficiently.

\subsection{Our Contribution}\label{subsec:Contribution}
We make the  following contributions:
\begin{enumerate}[leftmargin=*,topsep=0pt,partopsep=0pt, parsep=0pt]
    \item We introduce \emph{monotone-value neural networks (MVNNs)}, a new class of NNs. For our MVNNs, we use as activation function \emph{bounded ReLU (bReLU)} and enforce constraints on the parameters such that they are normalized and fulfill a monotonicity property (\Cref{sec:Utility Networks}). These MVNNs are specifically suited to model monotone (combinatorial) value functions in economic settings.
    \item On a technical level, we provide the following two theorems (\Cref{subsec:Theoretical Analysis and MIP-Formulation}): First, we prove that MVNNs are universal in the class of monotone and normalized combinatorial value functions, i.e., one can represent \emph{any} value function with arbitrarily complex substitutabilities and complementarities exactly as an MVNN. Second, we show how to formulate the MVNN-based WDP as a MILP, which is key to calculate optimal allocations in practice.
    \item We experimentally evaluate the learning performance of MVNNs vs. NNs in four different spectrum CA domains and show that MVNNs are significantly better at modelling bidders' combinatorial value functions (\Cref{sec:Prediction Performance of MVNNs}).    
    \item Finally, we experimentally investigate the performance of MVNNs vs. plain NNs when integrated into an existing ML-based ICA (MLCA) and compare them also to the FT-based method by \cited{weissteiner2022fourier}. We show that MVNNs lead to substantially smaller efficiency losses than existing state-of-the-art mechanisms (\Cref{sec:MVNN-based Iterative CA}).
\end{enumerate}

\section{Preliminaries}\label{sec:Preliminaries}
In this section, we present our formal model and review the ML-based ICA by \cited{brero2021workingpaper}.

\subsection{Formal Model for ICAs}\label{subsec:Formal Model for ICAs}
We consider a CA with $n$ bidders and $m$ indivisible items. Let $N=\{1,\ldots,n\}$ and $M=\{1,\ldots,m\}$ denote the set of bidders and items, respectively. We denote with $x\in \X=\{0,1\}^m$ a bundle of items represented as an indicator vector, where $x_{j}=1$ iff item $j \in M$ is contained in $x$. Bidders' true preferences over bundles are represented by their (private) value functions $v_i: \X\to \R_+,\,\, i \in N$, i.e., $v_i(x)$ represents bidder $i$'s true value for bundle $x$.

By $a=(a_1,\ldots,a_n) \in \X^n$ we denote an allocation of bundles to bidders, where $a_i$ is the bundle bidder $i$ obtains. We denote the set of \emph{feasible} allocations by $\F=\left\{a \in \X^n:\sum_{i \in N}a_{ij} \le 1, \,\,\forall j \in M\right\}$. We assume that bidders' have quasilinear utilities $u_i$, i.e., for payments $p\in \R^n_+,\,  u_i(a,p) = v_i(a_i) - p_i$. This implies that the (true) \emph{social welfare} $V(a)$ of an allocation $a$ is equal to the sum of all bidders' values, i.e., 
$V(a)=\sum_{i\in N} u_i(a_i,p_i) + u_{\text{auctioneer}}(p)=\sum_{i\in N} v_i(a_i)-p_i +\sum_{i\in N} p_i=\sum_{i\in N} v_i(a_i)$. We let $a^* \in \argmax_{a \in {\F}}V(a)$ be a social-welfare maximizing, i.e., \textit{efficient}, allocation. The \emph{efficiency} of any allocation $a \in \F$ is measured as $V(a)/V(a^*)$.

An ICA \textit{mechanism} defines how the bidders interact with the auctioneer and how the final allocation and payments are determined. We denote a bidder's (possibly untruthful) reported value function by $\hvi{}:\X\to\R_+$. In this paper, we consider ICAs that ask bidders to iteratively report their values $\hvi{x}$ for particular bundles $x$ selected by the mechanism. A set of $L\in \N$ reported bundle-value pairs of bidder $i$ is denoted as ${R_i=\left\{\left(x^{(l)},\hvi{x^{(l)}}\right)\right\}_{l=1}^L,\,x^{(l)}\in \X}$. Let $R=(R_1,\ldots,R_n)$ denote the tuple of reported bundle-value pairs obtained from all bidders. We define the \textit{reported social welfare} of an allocation $a$ given $R$ as $\hV{a|R}:=\sum_{i \in N:\, \left(a_i,\hvi{a_i}\right)\in R_i}\hvi{a_i},$ where $\left(a_i,\hvi{a_i}\right)\in R_i$ ensures that only values for reported bundles contribute. The final optimal allocation $a^*_{R}\in \F$ and payments $p(R)\in \R^n_+$ are computed based on the elicited reports $R$ \emph{only}. More formally, the optimal allocation $a^*_{R}\in \F$ given the reports $R$ is defined as 
\begin{align}\label{WDPFiniteReports}
a^*_{R} \in \argmax_{a \in {\F}}\hV{a|R}.
\end{align}

As the auctioneer can generally only query each bidder $i$ a limited number of bundles $|R_i| \leq \Qmax$ (e.g., $\Qmax=100$), the mechanism needs a sophisticated preference elicitation algorithm, with the goal of finding a highly efficient final allocation $a^*_{R}$ with a limited number of value queries.

\subsection{A Machine Learning-powered ICA}\label{subsec:A Machine Learning powered ICA}
We now review the \textit{machine learning-powered combinatorial auction (MLCA)} by \cited{brero2021workingpaper}. Interested readers are referred to \Appendixref{sec:appendix_A Machine Learning powered ICA}{Appendix A}, where we present MLCA in detail.

MLCA proceeds in rounds until a maximum number of value queries per bidder $\Qmax$ is reached. In each round, a generic ML algorithm $\mathcal{A}_i$ is trained for every bidder $i\in N$ on the bidder's reports $R_i$. Next, MLCA generates new value queries $\qnew=(\qnew_i)_{i=1}^{n}$ with $\qnew_i\in \X \setminus R_i$ by solving a ML-based WDP $\qnew \in \argmax\limits_{a \in {\F}}\sum\limits_{i \in N} \mathcal{A}_i(a_i)$. 

The idea is the following: if $\mathcal{A}_i$ are good surrogate models of the bidders' true value functions then $\qnew$ should be a good proxy of the efficient allocation $a^*$ and thus provide the auctioneer with valuable information. Additionally, in a real-world MLCA, bidders' are always allowed to report values of further bundles that they deem potentially useful (``push''-bids)
. This mitigates the risk of bidders not getting asked the right queries.
In our experiments, MLCA achieves already state-of-the-art results without making use of any ``push''-bids (mathematically additional ``push''-bids can only improve the results further).


At the end of each round, MLCA receives reports $\Rnew$ from all bidders for the newly generated queries $\qnew$, and updates the overall elicited reports $R$. When $\Qmax$ is reached, MLCA computes an allocation $a^*_R$ that maximizes the \emph{reported} social welfare (see \Cref{WDPFiniteReports}) and determines VCG payments $p(R)$ based on the reported values (see \Appendixref[Appendix ]{def:vcg_payments}{Definition~B.1}).

\begin{remark}[IR, No-Deficit, and Incentives of MLCA]
\cited{brero2021workingpaper} showed that MLCA satisfies \emph{individual rationality (IR)} and \emph{no-deficit}, with any ML algorithm $\mathcal{A}_i$. Furthermore, they studied the incentive properties of MLCA; this is important, given that opportunities for manipulations might lower efficiency. Like all deployed spectrum auctions (including the CCA \cite{ausubel2017practical}) MLCA is not strategyproof. However, \cited{brero2021workingpaper} argued that it has good incentives in practice; and given two additional assumptions, bidding truthfully is an ex-post Nash equilibrium in MLCA. Their analyses apply to MLCA using any ML algorithm, and therefore also to an MVNN-based MLCA. We present a more detailed summary of their incentive analysis in \Appendixref{sec:appendix_Incentives of MLCA}{Appendix B}.
\end{remark}

\section{Monotone-Value Neural Networks}\label{sec:Utility Networks}

In combinatorial assignment problems, value functions are used to model each agent's (reported) value for a bundle of items ($\hvi{}: \{0,1\}^{m}\to \mathbb{R}_+$). However, while the bundle space grows exponentially with the number of items, the agents' value functions often exhibit useful structure that can be exploited. A common assumption is monotonicity:
\begin{enumerate}[label=D\arabic*,align=left, leftmargin=*,topsep=2pt]
    \item[\textbf{(M)}]\textbf{Monotonicity}~(\emph{``additional items increase value''}):\\ For $A,B \in 2^M$: if $A\subseteq B$ it holds that $\hvi{A}\le \hvi{B}$.\footnote{We slightly abused the notation here by using sets instead of their corresponding indicator vectors as arguments of $\hvi{}$.}
\end{enumerate}
This property is satisfied in many market domains. For example, in many CAs, bidders can freely dispose of unwanted items; in combinatorial course allocation, students can just drop courses they have been assigned. However, prior work on ML-based market design \cite{weissteiner2020deep,weissteiner2022fourier,brero2021workingpaper} has not taken this property into account, which negatively affects the performance (see \Cref{sec:Prediction Performance of MVNNs,sec:MVNN-based Iterative CA}).

For ease of exposition, we additionally assume that the value functions are normalized:  
\begin{enumerate}[label=D\arabic*,align=left, leftmargin=*,topsep=2pt]
    \item[\textbf{(N)}]\textbf{Normalization}~(\emph{''no value for empty bundle''}):\\ $\hat{v}_i(\emptyset)=\hat{v}_i((0,\ldots,0)):=0$
\end{enumerate}
Note that this property is not required by our method and can be easily adapted to be any other fixed value, or to be a learned parameter. In the following, we denote with
\begin{align}\label{eq:Vmon}
    \Vmon:=\{\hvi{}:\X \to \Rp|\, \text{satisfy \textbf{(N)} and \textbf{(M)}}\}
\end{align}
the set of all value functions, that satisfy the \emph{normalization} and \emph{monotonicity} property. Next, we introduce \textit{Monotone-Value Neural Networks (MVNNs)} and show that they span the entire set $\Vmon$. Thus, MVNNs are specifically suited to all applications with monotone value functions.

\begin{definition}[MVNN]\label{def:MVNN}
		An MVNN $\MVNNi{}:\X \to \mathbb{R}_+$  for bidder $i\in N$ is defined as
		\begin{equation}\label{eq:MVNN}
		\MVNNi{x} =W^{i,K_i}\phiu_{0,t}\left(\ldots\phiu_{0,t}(W^{i,1}x+b^{i,1})\ldots\right),
		\end{equation}
		\begin{itemize}[leftmargin=*,topsep=0pt,partopsep=0pt, parsep=0pt]
		\item $K_i+1\in\mathbb{N}$ is the number of layers ($K_i-1$ hidden layers),
		\item $\phiu_{0,t}{}$ is the \emph{bounded ReLU (bReLU)}\footnote{bReLUs have been widely used in practice, e.g., to enhance training stability in visual pattern recognition \cite{liew2016bounded}.} activation function with cutoff $t>0$:
		\begin{align}\label{itm:MVNNactivation}
		\phiu_{0,t}{}(z):=\min(t, \max(0, z))
		\end{align}
		\item $W^i:=(W^{i,k})_{k=1}^{K_i}$ with $W^{i,k}\ge0$ and $b^i:=(b^{i,k})_{k=1}^{K_i}$ with $b^{i,k}\le0$ denote a tuple of non-negative weights and non-positive biases of dimensions $d^{i,k}\times d^{i,k-1}$ and $d^{i,k}$, whose parameters are stored in $\theta=(W^i,b^i)$.\footnote{We apply a linear readout map to the last hidden layer, i.e, no $\phiu_{0,t}$ and $b^{i,K_i}:=0$. By setting $b^{i,K_i}\neq0$ with \emph{trainable=False}, the MVNN can model any other value than zero in the normalization property. By not restricting $b^{i,k}\le0$ and setting $b^{i,K_i}\neq0$ with \emph{trainable=True} one can also learn the value for the empty bundle.}
		\end{itemize}
\end{definition}
For implementation details of MVNNs we refer to \Appendixref{subsec:appendix_Implementation Details_MVNNs}{Appendix~C.6}. Unless explicitly stated, we consider from now on a cutoff of $t=1$ for the bReLU, i.e., $\phiu(z):=\phiu_{0,1}(z)=\min(1, \max(0, z))$.
Next, we discuss the choice of bReLU.

\paragraph{Choice of bReLU} (i) The constraints on the weights and biases enforce monotonicity of MVNNs (in fact for \emph{any} monotone activation).
(ii) For universality (see \Cref{thm:universality}) we need a \emph{bounded} monotone non-constant
activation, i.e., with ReLUs and our constraints one cannot express substitutabilities. (iii) for the MILP (see \Cref{theo:util_mip}), we need a \emph{piecewise linear} activation (e.g.,
with sigmoids we could not formulate a MILP). 
Taking all together, bReLU is the simplest
bounded, monotone, non-constant, piecewise-linear activation (see \Appendixref[Appendix ]{rem:thmUniversality,rem:WhybReLU}{Remarks~C.2 and~C.3} for a detailed discussion).

\begin{remark}
For applications where value functions are expected to be ''almost`` (but not completely) monotone, one can adapt MVNNs to only have \emph{soft monotonicity constraints} by implementing the constraints on the weights and biases via regularization, e.g., $\sum_{i,k,j,l}\max(0,-W^{i,k}_{j,l})$.
This results in soft-MVNNs that can model non-monotone changes in some items if the data evidence is strong.
\end{remark}

\subsection{Theoretical Analysis and MILP-Formulation}\label{subsec:Theoretical Analysis and MIP-Formulation}
Next, we provide a theoretical analysis of MVNNs and present a MILP formulation of the MVNN-based WDP.

\begin{lemma}\label{lem:normalized and monoton}
    Let $\MVNNi{}:\X\to \Rp$ be an MVNN (\Cref{def:MVNN}). Then it holds that 
    $\MVNNi[(W^{i},b^{i})]{}\in\Vmon$ for all $W^{i}\ge0$ and $b^{i}\le0$.
\end{lemma}
We provide the proof for \Cref{lem:normalized and monoton} in \Appendixref{subsec:appendix_proofLemmaUniversality}{Appendix~C.1}. Next, we state our main theorem about MVNNs.

\begin{theorem}[Universality]\label{thm:universality} \emph{Any} value function $\hvi{}:\X\to\Rp$ that satisfies \textbf{(N)} and \textbf{(M)} can be represented exactly as an MVNN~$\MVNNi{}$ from \Cref{def:MVNN}, i.e.,
{
\begin{align}
\Vmon=\left\{\MVNNi[(W^i,b^i)]{}: W^{i}\ge0, b^{i}\le0 \right\}.
\end{align}}
\end{theorem}
We present a constructive proof for \Cref{thm:universality} in \Appendixref{subsec:appendix_proofTheoremUniversality}{Appendix~C.3}. In the proof, we consider an arbitrary $(\hvi{x})_{x\in \X} \in \Vmon$ for which we construct a two hidden layer MVNN $\MVNNi{}$ of dimensions $[m,2^m-1,2^m-1,1]$ with parameters $\theta =(W^i_{\hvi{}},b^i_{\hvi{}})$ such that $\MVNNi{x}=\hvi{x}\, \forall x\in \X$.

Note that for $q$ queries, it is always possible to construct MVNN of size $[m,q,q,1]$ that perfectly fits through the $q$ data-points (see 
\Appendixref[Appendix ]{app:cor:interpolation}{Corollary~C.1}). Thus, the worst-case required architecture size grows only linearly with the number of queries. In our experiments we already achieve very good performance with even smaller architectures. 
\AppendixrefSuffix{app:example_mvnns}{Example~C.1}{ in the Appendix} nicely illustrates how exactly MVNNs capture complementarities, substitutabilities and independent items.

A key step in combinatorial assignment mechanisms is finding the social welfare-maximizing allocation, i.e., solving the \emph{Winner Determination Problem} (WDP). To use MVNNs in such mechanisms, we need to be able to solve MVNN-based WDPs in a practically feasible amount of time. To this end, we present a MILP formulation of the MVNN-based WDP which can be (approximately) solved in practice for reasonably-sized NNs (see \Cref{subsec:mip_runtime}). The key idea is to rewrite the bReLU $\phiu(z)$ as $-\max(-1,- \max(0, z))$ and encode for each bidder $i$, hidden layer $k$ and neuron $j$ both $\max(\cdot,\cdot)$ operators with two binary decision variables $y^{i,k}_j, \mu^{i,k}_j$. First, we show how to encode one \emph{single} hidden layer of an MVNN as multiple linear constraints. We provide the proof in \Appendixref{subsec:appendix_proofLemmaMIP}{Appendix~C.4}.
\begin{lemma}\label{lem:util_layer}
    Fix bidder $i\in N$, let $k\in \{1,\ldots,K_i-1\}$ and denote the pre-activated output of the $k$\textsuperscript{th} layer as $o^{i, k}:=W^{i, k}z^{i, k-1} + b^{i, k}$ with $W^{i, k} \in \mathbb{R}^{d^{i,k} \times d^{i,k-1}}, b^{i, k} \in \mathbb{R}^{d^{i,k}}$. Then the output of the $k$\textsuperscript{th} layer  $z^{i, k} := \phiu(o^{i, k}) = \min(1, \max(0, o^{i, k})) = -\max(-1, - \eta^{i, k})$, with $\eta^{i, k}:=\max(0, o^{i, k})$ can be equivalently expressed by the following linear constraints:
    \begin{align}
        & o^{i, k}\le \eta^{i, k}\le o^{i, k} + y^{i, k}\cdot L_1^{i,k}\label{eq:(i)}\\
        & 0\le \eta^{i, k} \le (1-y^{i, k}) \cdot L_2^{i,k}\label{eq:(ii)}\\
        & \eta^{i, k} -  \mu^{i,k}\cdot L_3^{i,k}\le z^{i, k}\le \eta^{i, k}\label{eq:(iii)}\\
        & 1 -  (1-\mu^{i, k}) \cdot L_4^{i,k} \le z^{i, k} \le 1\label{eq:Lemma_cutoff_dependent_ct}\\
        & y^{i,k}\in \{0,1\}^{d^{i,k}}, \quad \mu^{i,k}\in \{0,1\}^{d^{i,k}},
    \end{align}
    where $L_1^{i,k}, L_2^{i,k}, L_3^{i,k}, L_4^{i,k} \in \Rp$ are large enough constants for the respective \emph {big-M} constraints.\footnote{To account for a general cutoff $t\neq 1$ in the bReLU, one needs to adjust \eqref{eq:Lemma_cutoff_dependent_ct} by replacing the left- and rightmost $1$ with $t$.}
\end{lemma}
Finally, we formulate the MVNN-based WDP as a MILP.
\begin{theorem}[MILP]\label{theo:util_mip}
The MVNN-based WDP $\max\limits_{a\in \F}\sum_{i \in N}\MVNNi[(W^i,b^i)]{a_i}$ can be equivalently formulated as the following MILP\footnote{To account for a general cutoff $t\neq 1$ in the bReLU, one needs to adjust \eqref{eq:MIP_cutoff_dependent_ct} by replacing the left- and rightmost $1$ with $t$.}:
    \begin{align}
        &\max\limits_{a\in \F, z^{i,k},\mu^{i,k},\eta^{i,k},y^{i,k}}\left\{\sum_{i \in N} W^{i, K_{i}} z^{i, K_{i}-1}\right\}\\
        &\hspace{-1cm}\text{s.t. for } i\in N \text{ and } k \in \{1,\ldots,K_i-1\}\notag\\
        &z^{i,0}=a_i\\
        &W^{i,k}z^{i,k-1}+b^{i,k}\le \eta^{i,k}\\
        &\eta^{i,k}\le W^{i,k}z^{i,k-1}+b^{i,k} + y^{i,k}\cdot L_1^{i,k} \\
        &0\le \eta^{i,k} \le (1-y^{i,k}) \cdot L_2^{i,k} \\
        &\eta^{i,k} -  \mu^{i,k}\cdot L_3^{i,k}\le z^{i,k}\le \eta^{i,k}\\
        &1 -  (1-\mu^{i,k}) \cdot L_4^{i,k} \le z^{i,k} \le 1 \label{eq:MIP_cutoff_dependent_ct}\\
        &y^{i,k}\in \{0,1\}^{d^{i,k}}, \qquad \mu^{i,k}\in \{0,1\}^{d^{i,k}}
    \end{align}
\end{theorem}
\begin{proof}
The proof follows by iteratively applying \Cref{lem:util_layer} for each bidder and all her respective hidden MVNN layers.
\end{proof}
\begin{fact}
One can significantly reduce the solve time for the MILP by tightening the bounds of each neuron. In \Appendixref{app:sec_util:ia_bounds}{Appendix~C.5}, we present bound tightening via \emph{interval arithmetic (IA)} \cite{tjeng2018evaluating} for MVNNs.
For a plain ReLU NN, these IA bounds are not tight and calculating tighter bounds in a computationally efficient manner is very challenging.
In contrast, the MVNN-IA bounds are always \emph{perfectly tight}, because of their encoded monotonicity property. The upper and lower bound of a neuron is the value the neuron outputs for the input $(1,\dots,1)$ and $(0,\dots,0)$. This is a big advantage of MVNNs compared to plain (ReLU) NNs.
\end{fact}

\begin{remark}[MVNNs as Preference Generator]
To experimentally evaluate a new mechanism one needs to generate many different possible value functions. Preference generators like the \textit{Spectrum Auction Test Suite}~\cite{weiss2017sats} are carefully designed to capture the essential properties of spectrum auctions, but they are not available for every application. 
Instead of using such a domain-specific preference generator, one can also use MVNNs with randomly initialized weights to generate possible value functions. 
An advantage of random MVNNs is that they are universal (see \Cref{thm:universality}) and hence come with a diversity rich enough to sample any possible monotone value function with arbitrarily complex substitutabilities and complementarities (the distribution of supplements and complements can be controlled via the cutoff $t$, where the smaller/larger $t$ the more substitutabilities/complementarities). Future work could investigate this distribution, i.e., how representative it is for real-world valuations.
These types of generative test beds become increasingly important to avoid overfitting on specific simulation engines and/or real data sets \cite{osband2021epistemic}.
\end{remark}

\section{Prediction Performance of MVNNs}\label{sec:Prediction Performance of MVNNs}
In this section, we show that in all considered CA domains, MVNNs are significantly better at capturing bidders' complex value functions than plain (ReLU) NNs, which allows them to extrapolate much better in the bundle space.

\subsection{Experimental Setup - Prediction Performance}\label{subsec:Experimental Setup - Prediction Performance}
\paragraph{CA Domains} In our experiments we use simulated data from the Spectrum Auction Test Suite (SATS) version 0.7.0 \cite{weiss2017sats}. We consider the following four domains:
\begin{itemize}[leftmargin=*,topsep=0pt,partopsep=0pt, parsep=0pt]
\item \textbf{Global Synergy Value Model (GSVM)} \cite{goeree2010hierarchical} has $18$ items, $6$ \emph{regional} and $1$ \emph{national bidder}.
\item \textbf{Local Synergy Value Model (LSVM)} \cite{scheffel2012impact} has $18$ items, $5$ \emph{regional} and $1$ \emph{national bidder}. Complementarities arise from spatial proximity of items.
\item \textbf{Single-Region Value Model (SRVM)} \cite{weiss2017sats} has $29$ items and $7$ bidders (categorized as  \emph{local}, \emph{high frequency} \emph{regional}, or \emph{national}) and models large UK 4G spectrum auctions.
\item \textbf{Multi-Region Value Model (MRVM)} \cite{weiss2017sats} has $98$ items and $10$ bidders (\emph{local}, \emph{regional}, or \emph{national}) and models large Canadian 4G spectrum auctions.
\end{itemize}

When simulating bidders, we follow prior work (e.g., \cite{brero2021workingpaper}) and assume truthful bidding (i.e., $\hat{v}_i=v_i$). Details on how we collect the data and the train/val/test split can be found in \Appendixref{subec:appendix_pred_perf_data_gen}{Appendix~D.1}.

\begin{table}[b!]
	\robustify\bfseries
	\centering
	\begin{sc}
	\resizebox{1\columnwidth}{!}{
	\setlength\tabcolsep{4pt}
		\begin{tabular}{lllllll}
	\toprule
		 &     &        & \multicolumn{2}{c}{$\boldsymbol{R^2}$ \textuparrow} &  \multicolumn{2}{c}{\textbf{ Kt \textuparrow}} \\
		 \cmidrule(l{2pt}r{2pt}){4-5}
		 \cmidrule(l{2pt}r{2pt}){6-7}
  \textbf{Domain} &   \textbf{T} & \textbf{Bidder} &    \multicolumn{1}{c}{MVNN}    &     \multicolumn{1}{c}{NN}     &  \multicolumn{1}{c}{MVNN} & \multicolumn{1}{c}{NN}\\
    \cmidrule(l{2pt}r{2pt}){1-3}
	\cmidrule(l{2pt}r{2pt}){4-5}
	\cmidrule(l{2pt}r{2pt}){6-7}
	GSVM & 10  & Nat &  \ccell0.686  $\pm$\scriptsize 0.061  &  0.534 $\pm$\scriptsize 0.040&  \ccell0.668 $\pm$\scriptsize 0.027&  0.583 $\pm$\scriptsize 0.021  \\
		 &     & Reg  &  \ccell0.618 $\pm$\scriptsize 0.068 &  \ccell0.504 $\pm$\scriptsize 0.062&  \ccell0.633 $\pm$\scriptsize 0.038&  0.557 $\pm$\scriptsize 0.033  \\
		 & 20  & Nat &  \ccell0.923 $\pm$\scriptsize 0.016 &  0.818 $\pm$\scriptsize 0.032 &  \ccell0.849 $\pm$\scriptsize 0.017&  0.752 $\pm$\scriptsize 0.029  \\
		 &     & Reg &  \ccell0.940 $\pm$\scriptsize 0.018 & 0.880 $\pm$\scriptsize 0.022 &  \ccell0.882 $\pm$\scriptsize 0.020&  0.815 $\pm$\scriptsize 0.021  \\
		 & 50  & Nat &  \ccell0.992 $\pm$\scriptsize 0.001 & 0.988 $\pm$\scriptsize 0.001 &  \ccell0.962 $\pm$\scriptsize 0.003&  0.953 $\pm$\scriptsize 0.003  \\
		 &     & Reg &  \ccell0.997 $\pm$\scriptsize 0.001 & 0.988 $\pm$\scriptsize 0.001 &  \ccell0.974 $\pm$\scriptsize 0.002&  0.953 $\pm$\scriptsize 0.003  \\
    \cmidrule(l{2pt}r{2pt}){1-3}
	\cmidrule(l{2pt}r{2pt}){4-5}
	\cmidrule(l{2pt}r{2pt}){6-7}
	LSVM & 10  & Nat &  \ccell0.248 $\pm$\scriptsize 0.069 &  0.137 $\pm$\scriptsize 0.031 &  \ccell0.693 $\pm$\scriptsize 0.011&  \ccell0.710 $\pm$\scriptsize 0.023 \\
		 &     & Reg &  \ccell0.563 $\pm$\scriptsize 0.049 &  0.348 $\pm$\scriptsize 0.067 &  \ccell0.605 $\pm$\scriptsize 0.031 &  0.504 $\pm$\scriptsize 0.025 \\
		 & 50  & Nat &  \ccell0.616 $\pm$\scriptsize 0.020 &  0.199 $\pm$\scriptsize 0.031 &  \ccell0.753 $\pm$\scriptsize 0.009 &  0.678 $\pm$\scriptsize 0.035 \\
		 &     & Reg &  \ccell0.921 $\pm$\scriptsize 0.015 &  0.872 $\pm$\scriptsize 0.012 &  \ccell0.860 $\pm$\scriptsize 0.017 &  0.812 $\pm$\scriptsize 0.013 \\
		 & 100 & Nat &  \ccell0.677 $\pm$\scriptsize 0.014 &  0.396 $\pm$\scriptsize 0.033 &  \ccell0.813 $\pm$\scriptsize 0.005 &  0.706 $\pm$\scriptsize 0.018 \\
		 &     & Reg &  \ccell0.965 $\pm$\scriptsize 0.010 &  0.936 $\pm$\scriptsize 0.010 &  \ccell0.918 $\pm$\scriptsize 0.015 &  0.857 $\pm$\scriptsize 0.012 \\
		 \cmidrule(l{2pt}r{2pt}){1-3}
	\cmidrule(l{2pt}r{2pt}){4-5}
	\cmidrule(l{2pt}r{2pt}){6-7}
	SRVM & 10  & H.F. &  \ccell0.538 $\pm$\scriptsize 0.044 &  -2.123  $\pm$\scriptsize 0.268 &  \ccell0.626 $\pm$\scriptsize 0.020 &  \ccell0.607 $\pm$\scriptsize 0.012 \\
		 &     & Lo &  \ccell0.381 $\pm$\scriptsize 0.045 &  0.267 $\pm$\scriptsize 0.042 &  \ccell0.559 $\pm$\scriptsize 0.030 &  0.489 $\pm$\scriptsize 0.032 \\
		 &     & Nat &  \ccell0.389 $\pm$\scriptsize 0.063 &  \ccell0.341 $\pm$\scriptsize 0.038 &  \ccell0.560 $\pm$\scriptsize 0.026 &  \ccell0.535 $\pm$\scriptsize 0.012 \\
		 &     & Reg &  \ccell0.422 $\pm$\scriptsize 0.051 &  \ccell0.372 $\pm$\scriptsize 0.036 &  \ccell0.562 $\pm$\scriptsize 0.023 &  \ccell0.544 $\pm$\scriptsize 0.014 \\
		 & 50  & H.F. &  \ccell0.860 $\pm$\scriptsize 0.015 &  0.773 $\pm$\scriptsize 0.034 &  \ccell0.853 $\pm$\scriptsize 0.013 &  0.803 $\pm$\scriptsize 0.020 \\
		 &     & Lo &  \ccell0.895 $\pm$\scriptsize 0.020 & 0.588 $\pm$\scriptsize 0.031 &  \ccell0.902 $\pm$\scriptsize 0.000 &  0.771 $\pm$\scriptsize 0.030 \\
		 &     & Nat &  \ccell0.988 $\pm$\scriptsize 0.004 & 0.828 $\pm$\scriptsize 0.015 &  \ccell0.918 $\pm$\scriptsize 0.005&  0.801 $\pm$\scriptsize 0.009  \\
		 &     & Reg &  \ccell0.989 $\pm$\scriptsize 0.004 & 0.872 $\pm$\scriptsize 0.047 &  \ccell0.931 $\pm$\scriptsize 0.004&  0.823 $\pm$\scriptsize 0.022  \\
		 & 100 & H.F. &  \ccell0.911 $\pm$\scriptsize 0.008 & 0.849 $\pm$\scriptsize 0.011 &  \ccell0.908 $\pm$\scriptsize 0.006&  0.896 $\pm$\scriptsize 0.006  \\
		 &     & Lo&  \ccell0.948 $\pm$\scriptsize 0.014  & 0.723 $\pm$\scriptsize 0.005 &  \ccell0.903 $\pm$\scriptsize 0.000&  0.900 $\pm$\scriptsize 0.002  \\
		 &     & Nat &  \ccell0.998 $\pm$\scriptsize  0.000 & 0.913 $\pm$\scriptsize 0.008 &  \ccell0.952 $\pm$\scriptsize 0.003&  0.841 $\pm$\scriptsize 0.008  \\
		 &     & Reg &  \ccell0.996 $\pm$\scriptsize 0.001 &  \ccell0.945 $\pm$\scriptsize 0.004 &  \ccell0.948 $\pm$\scriptsize 0.021&  0.895 $\pm$\scriptsize 0.012 \\
		 \cmidrule(l{2pt}r{2pt}){1-3}
	\cmidrule(l{2pt}r{2pt}){4-5}
	\cmidrule(l{2pt}r{2pt}){6-7}
 	MRVM & 10  & Lo &  \ccell-0.055 $\pm$\scriptsize 0.058 & \ccell-0.018 $\pm$\scriptsize 0.050 &  \ccell0.262 $\pm$\scriptsize 0.017&  0.200 $\pm$\scriptsize 0.015  \\
         &     & Nat &  \ccell0.182 $\pm$\scriptsize 0.045 &  \ccell-0.556 $\pm$\scriptsize 1.355 &  0.365 $\pm$\scriptsize 0.018 &  \ccell0.414 $\pm$\scriptsize 0.008 \\
         &     & Reg &  \ccell0.036 $\pm$\scriptsize 0.085 & \ccell-0.048 $\pm$\scriptsize 0.092 &  \ccell0.322 $\pm$\scriptsize 0.022&  0.255 $\pm$\scriptsize 0.038  \\
         & 100 & Local &  \ccell0.831 $\pm$\scriptsize 0.023 & 0.493 $\pm$\scriptsize 0.027  &  \ccell0.786 $\pm$\scriptsize 0.019 &  0.545 $\pm$\scriptsize 0.012\\
         &     & Nat &  \ccell0.778 $\pm$\scriptsize 0.022 &  0.560 $\pm$\scriptsize 0.019 &  \ccell0.726 $\pm$\scriptsize 0.014 &  0.581 $\pm$\scriptsize 0.010 \\
         &     & Reg &  \ccell0.832 $\pm$\scriptsize 0.028 &  0.447 $\pm$\scriptsize 0.053 &  \ccell0.779 $\pm$\scriptsize 0.027 &  0.572 $\pm$\scriptsize 0.018 \\
         & 300 & Local &  \ccell0.944 $\pm$\scriptsize 0.006 &  0.871 $\pm$\scriptsize 0.009 &  \ccell0.883 $\pm$\scriptsize 0.010 &  0.819 $\pm$\scriptsize 0.009 \\
         &     & Nat &  \ccell0.868 $\pm$\scriptsize 0.016 &  \ccell0.855 $\pm$\scriptsize 0.029 &  \ccell0.814 $\pm$\scriptsize 0.009&  \ccell0.808 $\pm$\scriptsize 0.013  \\
         &     & Reg &  \ccell0.917 $\pm$\scriptsize 0.017 &  0.847 $\pm$\scriptsize 0.010 &  \ccell0.851 $\pm$\scriptsize 0.019 &  0.809 $\pm$\scriptsize 0.012 \\
	\bottomrule
	\end{tabular}}
	\vskip -0.2cm
	\caption{Prediction performance measured via R-squared ($\boldsymbol{R^2}$) and Kendall tau (\textbf{\textsc{Kt}}) with a $95\%$-CI in four SATS domains with corresponding bidder types: high frequency (\textsc{H.F.}), local (\textsc{Lo}), regional (\textsc{Reg}) and national (\textsc{Nat}), averaged over 30 auction instances. Both MVNNs and plain NNs are trained on $T$ and evaluated on $209\,715 - T$ random bundles. Winners are marked in grey.}
	\label{tab:full_pred_performance_table}
	\end{sc}
\end{table}
\paragraph{HPO} To efficiently optimize the hyperparameters and fairly compare MVNNs and plain NNs for best generalization across different instances of each SATS domain, we frame the \emph{hyperparameter optimization (HPO)} problem as an algorithm configuration problem and use the well-established \emph{sequential model-based algorithm configuration (SMAC)}~\cite{hutter2011sequential}. SMAC quickly discards hyperparameters which already perform poorly on a few SATS instances and proposes more promising ones via Bayesian optimization. It is flexible enough for the parameterization of NNs as it naturally handles a mixture of categorical, integer and float hyperparameters. Further details on the setting including hyperparameter ranges can be found in \Appendixref{subec:appendix_pred_perf_hpo}{Appendix~D.2}.
\begin{table*}[ht]
	\robustify\bfseries
	\centering
	\begin{sc}
	\resizebox{1\textwidth}{!}{
	\small
    \begin{tabular}{crcccccc}
    \toprule
         &  & \multicolumn{4}{c}{\textbf{Efficiency Loss in \%\,\,\textdownarrow}} & \multicolumn{2}{c}{\textbf{T-Test for Efficiency:}}\\
        \cmidrule(l{2pt}r{2pt}){3-6}
        \cmidrule(l{2pt}r{2pt}){7-8}
    \textbf{Domain} &$\boldsymbol{\Qmax}$ &    \multicolumn{1}{c}{MVNN}  &   \multicolumn{1}{c}{NN}  & \multicolumn{1}{c}{FT}& \multicolumn{1}{c}{RS}  & $\mathcal{H}_0:\mu_{\text{NN}}\le\mu_{\text{MVNN}}$ & $\mathcal{H}_0:\mu_{\text{FT}}\le\mu_{\text{MVNN}}$\\
        \cmidrule(l{2pt}r{2pt}){1-2}
        \cmidrule(l{2pt}r{2pt}){3-6}
        \cmidrule(l{2pt}r{2pt}){7-8}
    GSVM  & 100& \ccell00.00 $\pm$\scriptsize\, 0.00 & \ccell00.00 $\pm$\scriptsize\, 0.00 & 01.77$\pm$\scriptsize\, 0.96&30.34 $\pm$\scriptsize\, 1.61 & &$p_{\text{val}}=3\mathrm{e}{-6}$\\
    
    LSVM & 100&   \ccell00.70 $\pm$\scriptsize\, 0.40 & 02.91 $\pm$\scriptsize\, 1.44 & 01.54$\pm$\scriptsize\,0.65&31.73 $\pm$\scriptsize\, 2.15 &  $p_{\text{val}}=2\mathrm{e}{-03}$&$p_{\text{val}}=5\mathrm{e}{-3}$\\
    
    SRVM& 100&   \ccell00.23 $\pm$\scriptsize\, 0.06  & 01.13 $\pm$\scriptsize\, 0.22 &00.72$\pm$\scriptsize\,0.16 &28.56 $\pm$\scriptsize\, 1.74 & $p_{\text{val}}=5\mathrm{e}{-10}$&$p_{\text{val}}=2\mathrm{e}{-8}$\\
    
    MRVM& 100&  \ccell08.16 $\pm$\scriptsize\, 0.41 & 09.05 $\pm$\scriptsize\, 0.53 & 10.37$\pm$\scriptsize\, 0.57 &48.79 $\pm$\scriptsize\, 1.13 & $p_{\text{val}}=9\mathrm{e}{-03}$&$p_{\text{val}}=1\mathrm{e}{-7}$\\
    
    \bottomrule
    \end{tabular}
}
    \end{sc}
    \vskip -0.2cm
    \caption{Efficiency loss of MVNN vs plain NNs, the Fourier Transform (FT) benchmark and random search (RS). Shown are averages and a 95\% CI on a test set of $50$ CA instances. Winners based on a (paired) t-test  with significance level of 1\% are marked in grey.}
    \label{tab:efficiency_loss_mlca}
\end{table*}
\subsection{Prediction Performance Results}\label{subsec:Results - Prediction Performance}
For ease of exposition, we only present our results for the \textsc{MVNN-ReLU-Projected} implementation of our MVNNs (termed MVNN in the following). Results for other MVNN implementations can be found in the \Appendixref{subec:appendix_pred_perf_detailed_results}{Appendix~D.3}.
In Table~\ref{tab:full_pred_performance_table}, we compare the prediction performance of the winning models that the HPO found for different amounts of training data (T) on the test data. We see that, compared to plain NNs, MVNNs provide both a significantly better fit in terms of R-squared $\boldsymbol{R^2}$ as well as a better Kendall Tau rank correlation \textbf{\textsc{Kt}} (i.e., a better ordinal ranking of the predicted test bundle values). Thus, enforcing the monotonicity property in MVNNs significantly improves the learning performance.

Figure~\ref{fig:pred_perf_scatter} illustrates our findings by providing a visual comparison of the prediction performance for the highly non-unimodal SRVM.\footnote{We provide corresponding plots for the other domains and bidder types in \Appendixref{subec:appendix_pred_perf_detailed_results}{Appendix~D.3}; the results are qualitatively similar.} We see that the MVNN fits the training data exactly (blue crosses), although the HPO only optimized generalization performance on the validation data. This is a strong indication that MVNNs correspond to a more realistic prior, since for a realistic prior, it is optimal to exactly fit the training data in noiseless settings \cite[Proposition D.2.a]{heiss2022nomu}. In contrast, the HPO has selected hyperparameters for the plain NNs that result in a worse fit of the training data (otherwise generalization to unseen data would be even worse). Moreover, the plain NNs show a particularly bad fit on the less frequent lower and higher valued bundles.

\begin{figure}[t!]
    \centering
    \includegraphics[width=1\columnwidth]{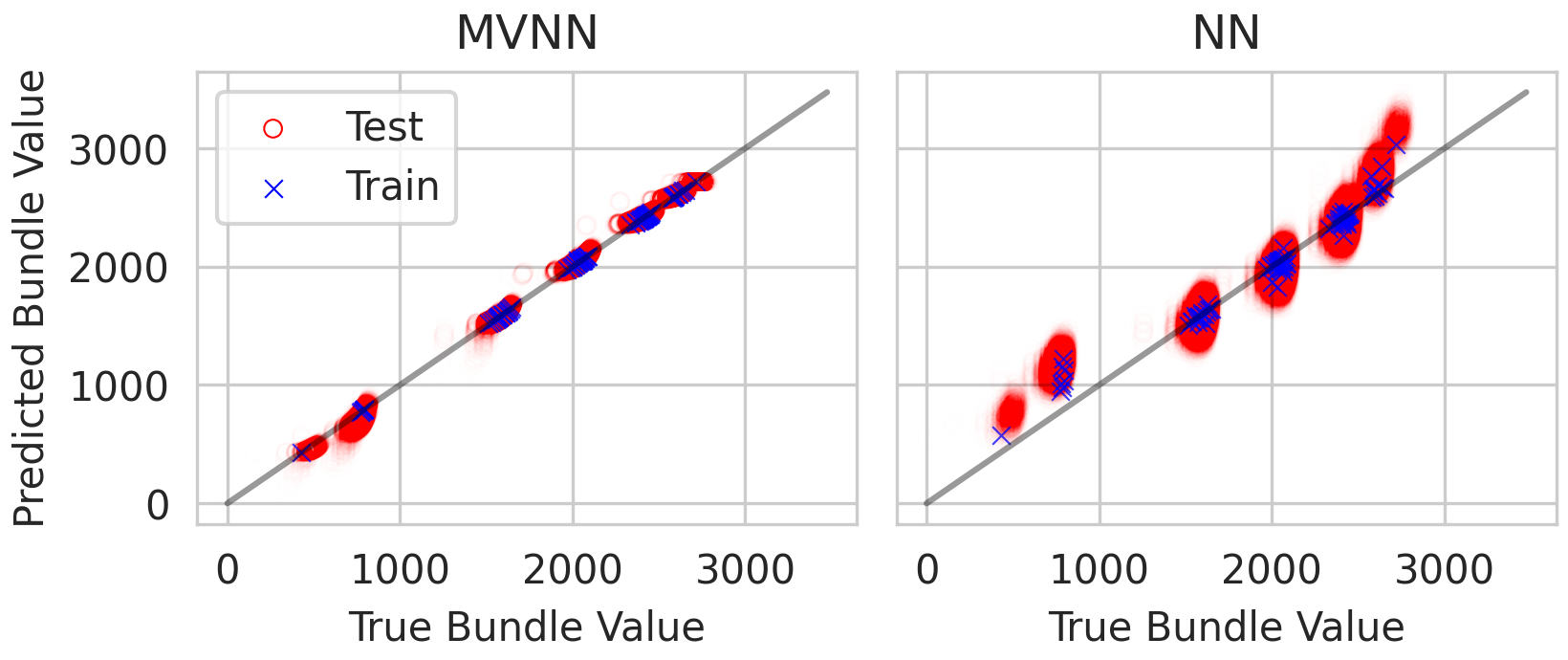}
    \vskip -0.2cm
    \caption{Prediction performance of MVNNs vs plain NNs in SRVM (national bidder). The identity is shown in grey.
    }
    \label{fig:pred_perf_scatter}
\end{figure}

\section{MVNN-powered Iterative CA}\label{sec:MVNN-based Iterative CA}
To evaluate the performance of MVNNs when used inside a combinatorial market mechanism, we have integrated MVNNs into  MLCA (see \Cref{subsec:A Machine Learning powered ICA}), yielding an \emph{MVNN-powered Iterative CA}. In this section, we compare the economic efficiency of our MVNN-based MLCA against the previously proposed NN-based MLCA. For solving the MVNN-based WDPs in MLCA, we use our MILP from Theorem~\ref{theo:util_mip}.

\subsection{Experimental Setup - MLCA}\label{subsec:Experimental Setup - MLCA}
To generate synthetic CA instances, we use the same four SATS domains as in \Cref{sec:Prediction Performance of MVNNs}. SATS also gives us access to the true optimal allocation $a^*$, which we use to measure the \emph{efficiency loss}, i.e., $1-V(a^*_R)/V(a^*)$ and \emph{relative revenue} $\sum_{i\in N}p(R)_i/V(a^*)$ of an allocation $a^*_R\in \mathcal{F}$ and payments $p(R) \in \mathbb{R}^n_+$ determined by MLCA when eliciting reports $R$. 
Due to the long run-time of a single evaluation of MLCA, we perform a restricted HPO, which, for each domain, only uses the winners of the prediction performance experiment (\Cref{tab:full_pred_performance_table}) for the three amounts of training data $T$. This is a reasonable choice, because in the prediction performance we optimize the generalization performance for bidders' value functions that is also key in MLCA.

For each domain, we use $\Qinit=40$ initial random queries and set the query budget to $\Qmax=100$. We terminate MLCA in an intermediate iteration if it already found an efficient allocation (i.e., with 0 efficiency loss).

\subsection{Efficiency Results}\label{subsec:Results - MLCA}
In Table~\ref{tab:efficiency_loss_mlca}, we present the efficiency results of MVNN-based MLCA and NN-based MLCA, averaged over 50 auction instances. We focus on efficiency rather than revenue, since spectrum auctions are government-run auctions with a mandate to maximize efficiency and not revenue \cite{cramton2013spectrumauctions}. In \Appendixref{subsec:appendix_detailed_results_MLCA,subsec:appendix_detailed_revenue_analysis}{Appendices~E.2 and~E.3}, we also present and discuss detailed revenue results.

For each domain, we present results corresponding to the best MVNNs and NNs amongst the three incumbents obtained from the prediction performance experiments. We present results for all three incumbents in \Appendixref{subsec:appendix_detailed_results_MLCA}{Appendix~E.2}. Overall, we see that the better prediction performance of MVNNs (Table~\ref{tab:full_pred_performance_table}) translates to smaller efficiency losses in MLCA. In LSVM and SRVM, MVNNs significantly outperform NNs and have a more than four times lower efficiency loss. In MRVM, MVNN's average efficiency loss is approximately 1\% point smaller than the NN's loss. Given that in the 2014 Canadian 4G auction the total revenue was on the order of 5 billion USD \cite{ausubel2017practical}, an efficiency loss decrease of 1\% point in MRVM can translate to welfare gains on the order of 50 million USD. Finally, in GSVM, the simplest domain, where bidders' value functions have at most two-way interactions between items, both MVNNs and plain NNs incur no efficiency loss. As further baselines, we evaluate the Fourier transform (FT) auction \cite{weissteiner2022fourier} using their proposed optimal hyperparameters and random search (RS). We do not compare to SVRs \cite{brero2021workingpaper} since they were already outperformed by plain NNs in \cite{weissteiner2020deep}.
We observe that RS incurs efficiency losses of 30--50\% illustrating the need for smart preference elicitation. Moreover, we see that MVNNs also significantly outperform FTs in all domains.

In Figure~\ref{fig:efficiency_path_plot_summary}, we present the efficiency loss path of MVNNs vs NNs (i.e., the regret curve) corresponding to \Cref{tab:efficiency_loss_mlca}. We see that in LSVM and SRVM, MVNNs lead to a smaller (average) efficiency loss for \emph{every} number of queries. In MRVM, the same holds true for $50$ and more queries. In GSVM, both networks have no efficiency loss in every instance after only $56$ queries. Since a single query can be very costly in real-world CAs, it makes sense to ask few queries. Figure~\ref{fig:efficiency_path_plot_summary} shows that MVNNs consistently outperform plain NNs also in settings with a small number of queries (i.e., reduced $\Qmax$)
\begin{figure}[t!]
    \centering
     \includegraphics[width=1\columnwidth]{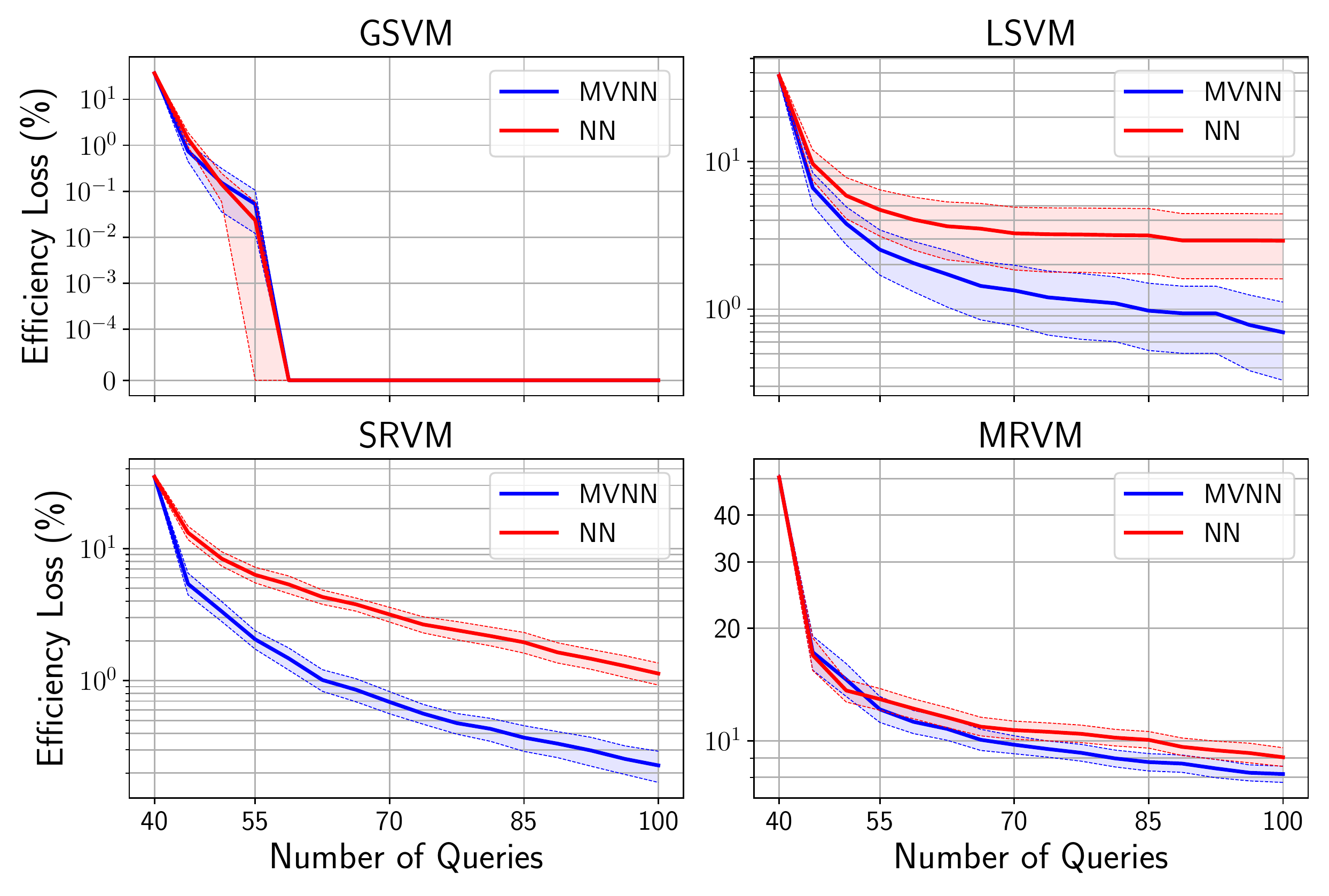}
        \vskip -0.2cm
    \caption{Efficiency loss paths of MLCA with MVNNs vs plain NNs. Shown are averages with 95\% CIs over 50 auction instances.}
    \label{fig:efficiency_path_plot_summary}
\end{figure}

\subsection{MILP Runtime Analysis}\label{subsec:mip_runtime} When integrating MVNNs into MLCA or another iterative combinatorial assignment mechanism, one needs to solve the MVNN-based WDP multiple times in one full run.
Thus, the key computational challenge when integrating MVNNs in such mechanisms is to make solving the MVNN-based WDP practically feasible. In Theorem~\ref{theo:util_mip}, we have shown how to encode the MVNN-based WDP as a succinct MILP, which can be (approximately) solved in practice for reasonably-sized NNs. However, due to the bReLU, i.e., the two cutoffs at $0$ and $t>0$, the MVNN-based MILP has twice the number of binary variables ($y^{i,k}$ and $\mu^{i,k}$) than the MILP encoding of a plain NN with ReLUs \cite{weissteiner2020deep}. 

Figure~\ref{fig:mip_runtime_analysis} presents MILP runtimes of MVNNs vs plain ReLU NNs for selected architectures. We observe two effects: First, even though the MVNN-based MILPs have twice the number of binary variables, they can be solved faster than the plain NN-based MILPs. Second, the deeper the architecture or the more neurons, the larger this difference becomes. One hypothesis to explain these effects is that (i) MVNNs are more regular functions than plain NNs (due to their monotonicity property) and (ii) the constraints on their parameters yield structure that might be exploited by the MILP solver.
\begin{figure}[t!]
    \centering
    \includegraphics[width=1\columnwidth]{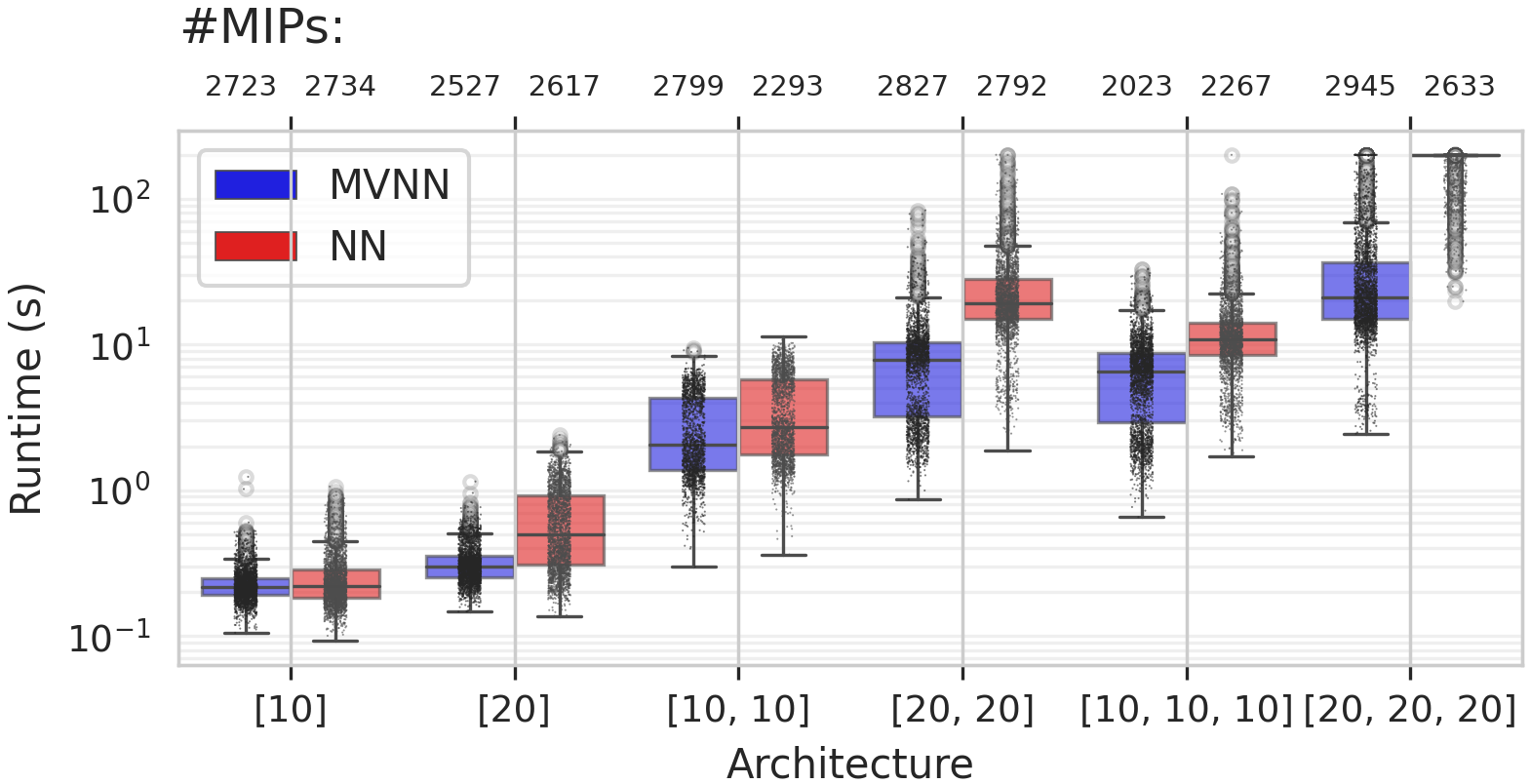}
    \vskip -0.2cm
    \caption{MILP runtime (200s time limit) of MVNNs vs plain NNs in MLCA on $10$ LSVM instances for a selection of architectures.
    }
    \label{fig:mip_runtime_analysis}
\end{figure}

\section{Conclusion}\label{sec:Conclusion}
In this paper, we have introduced MVNNs, a new class of NNs that is specifically designed to model normalized and monotone value functions in combinatorial assignment problems. We have experimentally evaluated the performance of MVNNs in four combinatorial spectrum auction domains and shown that MVNNs outperform plain NNs with respect to prediction performance, economic efficiency, and runtime. Overall, our experiments suggest that MVNNs are the best currently available model for preference elicitation in combinatorial assignment (also compared to FTs and SVRs).
Thus, incorporating important structural knowledge in the ML algorithm plays an important role in combinatorial assignment.

MVNNs enable us to incorporate an informative \textit{prior} into a market mechanism. Future work could use such informative priors and enhance existing mechanisms (e.g., MLCA) by also using the \textit{posterior} estimates in a more principled way than just the mean prediction. For example, one could frame an ICA as a (combinatorial) Bayesian optimization task and integrate a well-defined notion of posterior uncertainty to foster exploration \cite{heiss2022nomu}.\footnote{After the publication of the present paper, \cited{weissteiner2023bayesian}  introduced uncertainty-based exploration for ML-powered ICAs using the method by \cited{heiss2022nomu}.} Finally, it would be interesting to also evaluate the performance of MVNNs in other combinatorial assignment problems such as course allocation.\footnote{After the publication of the present paper, \cited{soumalias2023machine} used MVNNs to learn students' preferences in the course allocation domain.}

\section*{Acknowledgments}
We thank the anonymous reviewers for helpful comments. This paper is part of a project that has received funding from the European Research Council (ERC)
under the European Union’s Horizon 2020 research and innovation program (Grant agreement No. 805542).

\bibliographystyle{named}
\bibliography{references_arXiv}

\clearpage
\appendix
\section*{\centering Appendix}
\counterwithin{definition}{section}
\counterwithin{corollary}{section}
\counterwithin{problem}{section}
\counterwithin{example}{section}
\counterwithin{remark}{section}
\counterwithin{fact}{section}
\section{A Machine Learning-powered ICA}\label{sec:appendix_A Machine Learning powered ICA}
In this section, we present in detail the \textit{machine learning-powered combinatorial auction (MLCA)} by \cited{brero2021workingpaper}.

At the core of MLCA is a \textit{query module} (Algorithm~\ref{alg:QueryModule}), which, for each bidder $i\in I\subseteq N$, determines a new value query $q_i$. First, in the \textit{estimation step} (Line 1), an ML algorithm $\mathcal{A}_i$ is used to learn bidder $i$'s valuation from reports $R_i$. Next, in the \textit{optimization step} (Line 2), an \textit{ML-based WDP} is solved to find a candidate $q$ of value queries. In principle, any ML algorithm $\mathcal{A}_i$ that allows for solving the corresponding ML-based WDP in a fast way could be used. Finally, if $q_i$ has already been queried before (Line 4), another, more restricted NN-based WDP (Line 6) is solved and $q_i$ is updated correspondingly. This ensures that all final queries $q$ are new.
\setlength{\textfloatsep}{5pt}
\begin{algorithm}[h!]
        \DontPrintSemicolon
        \SetKwInOut{inputs}{Inputs}
        \inputs{~Index set of bidders $I$ and reported values $R$}
    \lForEach(\Comment*[f]{\color{blue}Estimation step}){$i \in I$}{
    \hspace{-0.07cm}Fit $\mathcal{A}_i$ on $R_i$: $\mathcal{A}_i[R_i]$
    }
    Solve $q \in \argmax\limits_{a \in {\F}}\sum\limits_{i \in I} \mathcal{A}_i[R_i](a_i)$ \hspace{-0.03cm}\Comment*[r]{\color{blue}Optimization step}
    \ForEach{$i \in I$}{
        \If(\Comment*[f]{\color{blue} Bundle already queried}){$(q_i,\hvi{q_i})\in R_i$}{ 
        Define $\pF=  \{a\in \F : a_i \neq x, \forall (x,\hvi{x})\in R_i\}$\;
        Re-solve $\pq \in \argmax_{a \in \pF}\sum_{l \in I} \mathcal{A}_l[R_l](a_l)$\;
        Update $q_i = \pqi\;$
        }
    }
    \Return{Profile of new queries $q=(q_1,\ldots,q_n)$}
    \caption{\textsc{NextQueries}$(I,R)$\, {\scriptsize (Brero et al. 2021)}}
    \label{alg:QueryModule}
\end{algorithm}

In Algorithm~\ref{MLCA}, we present \textsc{Mlca}. In the following, let $R_{-i}=(R_1,\ldots,R_{i-1},R_{i+1},\ldots, R_n)$. \textsc{Mlca} proceeds in rounds until a maximum number of queries per bidder $\Qmax$ is reached. In each round, it calls Algorithm \ref{alg:QueryModule}  $(\Qround-1)n+1$ times: for each bidder $i\in N$, $\Qround-1$ times excluding a different bidder $j\neq i$ (Lines 5--10,  sampled \textit{marginal economies}) and once including all bidders (Line 11, \textit{main economy}). In total each bidder is queried $\Qround$ bundles per round in \textsc{MLCA}. At the end of each round, the mechanism receives reports $\Rnew$ from all bidders for the newly generated queries $\qnew$, and updates the overall elicited reports $R$ (Lines 12--14). In Lines 16--17, \textsc{Mlca} computes an allocation $a^*_R$ that maximizes the \emph{reported} social welfare (see \Cref{WDPFiniteReports}) and determines VCG payments $p(R)$ based on the reported values $R$ (see \Appendixref[Appendix ]{def:vcg_payments}{Definition~B.1}).
\setlength{\textfloatsep}{5pt}
\begin{algorithm}[t!]
        \DontPrintSemicolon
        \SetKwInOut{parameters}{Params}
        \parameters{$\Qinit,\Qmax,\Qround$ {initial, max and \#queries/round}}
    \ForEach{$i \in N$}{Receive reports $R_i$ for $\Qinit$ randomly drawn bundles}
    \For(\Comment*[f]{\hspace{-0.05cm}\color{blue}Round iterator}){$k=1,...,\floor{(\Qmax-\Qinit)/\Qround}$}{
        \ForEach(\Comment*[f]{\color{blue}Marginal economy queries}){$i \in N$}{
            {Draw uniformly without replacement $(\Qround\hspace{-0.1cm}-\hspace{-0.1cm}1)$ bidders from $N\setminus\{i\}$ and store them in $\tilde{N}$}\;
            \ForEach{$j \in \tilde{N}$}{
            $\qnew=\qnew\cup$ \textit{NextQueries$(N\setminus\{j\},R_{-j})$}
            }
        }
        $\qnew=$ \textit{NextQueries$(N,R)$} \Comment*[r]{\color{blue}Main economy queries}
        \ForEach{$i \in N$}{
         Receive reports $\Rnewi$ for $\qnew_i$, set $R_i=R_i\cup\Rnewi$
        }
    }
    Given elicited reports $R$ compute $a^*_{R}$ as in \Cref{WDPFiniteReports}\;
    Given elicited reports $R$ compute VCG-payments $p(R)$\;
    \Return{Final allocation $a^*_{R}$ and payments $p(R)$}
    \caption{\small \textsc{Mlca}($\Qinit,\Qmax,\Qround$)\, {\scriptsize (Brero et al. 2021)}}
    \label{MLCA}
\end{algorithm}

\section{Incentives of MLCA}\label{sec:appendix_Incentives of MLCA}
In this section, we briefly review the key arguments by \cited{brero2021workingpaper} why MLCA has good incentives in practice. First, we define VCG-payments given bidder's reports.

\begin{definition}{\textsc{(VCG Payments from Reports)}}\label{def:vcg_payments}
Let $R=(R_1,\ldots,R_n)$ denote an elicited set of reported bundle-value pairs from each bidder obtained from \textsc{Mlca} (\Cref{MLCA}) and let $R_{-i}\coloneqq(R_1,\ldots,R_{i-1},R_{i+1},\ldots,R_n)$. We then calculate the VCG payments $p(R)=(p(R)_1\ldots,p(R)_n) \in \R_+^n$ as follows:
\begin{align}\label{VCGPayments}
&p(R)_i \coloneqq \hspace{-0.2cm}\sum_{j \in N \setminus \{i\}} \hvj{}{}\left(\left(a^*_{R_{-i}}\right)_j\right) - \hspace{-0.2cm}\sum_{j \in N \setminus \{i\}}\hvj{}{}\left(\left(a^*_{R}\right)_j\right).
\end{align}
where $a^*_{R_{-i}}$ maximizes the reported social welfare when excluding bidder $i$, i.e.,
\begin{align}
&a^*_{R_{-i}}\in \argmax_{a \in \F} \hV{a|R_{-i}} = \argmax_{a \in \F}\hspace{-0.4cm}\sum_{\substack{j \in N\setminus\{i\}:\\ \left(a_j,\hvj{}(a_j)\right)\in R_j}}\hspace{-0.4cm}\hvj{}(a_j),
\end{align}
and $a^*_R$ is a reported-social-welfare-maximizing allocation (including all bidders), i.e,
\begin{align}
&a^*_{R}\in \argmax_{a \in \F} \hV{a|R} = \argmax_{a \in \F} \hspace{-1.3cm}\sum_{\hspace{1cm}i \in N:\, \left(a_i,\hvi{}(a_i)\right)\in R_i}\hspace{-1.3cm}\hvi{}(a_i).
\end{align}
\end{definition}

Therefore, when using VCG, bidder $i$'s utility is:
{\small\begin{align*}
u_i \hspace{-0.05cm}=&v_i((a_R^*)_i)-p(R)_i\\
=&\hspace{-0.05cm}  \underbrace{v_i((a_R^*)_i) + \hspace{-0.3cm}\sum_{j \in N \setminus \{i\}}\hspace{-0.25cm}\hvj{}((a^*_{R})_j)}_{\textrm{\scriptsize (a) Reported SW of main economy}} - \hspace{-0.1cm} \hspace{-0.1cm}\underbrace{\sum_{j \in N \setminus \{i\}}\hspace{-0.25cm} \hvj{}((a^*_{R_{-i}})_j).}_\textrm{{\scriptsize (b) Reported SW of marginal economy}}
\end{align*}}
Any beneficial misreport must increase the difference (a) $-$ (b). 

MLCA has two features that mitigate manipulations. First, MLCA explicitly queries each bidder's marginal economy (\Cref{MLCA}, Line 5), which implies that (b) is practically independent of bidder $i$'s bid (Section 7.3 in \cite{brero2021workingpaper} provides experimental support for this). Second, MLCA enables bidders to ``push'' information to the auction which they deem useful. This mitigates certain manipulations that target (a), as it allows bidders to increase (a) with truthful information. \cited{brero2021workingpaper} argue that any remaining manipulation would be implausible as it would require almost complete information.

If we are willing to make two assumptions, we also obtain a theoretical incentive guarantee. Assumption 1 requires that, if all bidders bid truthfully, then MLCA finds an efficient allocation (we show in \Cref{subsec:appendix_detailed_results_MLCA} that in two of our domains: GSVM, LSVM, we indeed find the efficient allocation in the majority of cases). Assumption 2 requires that, for all bidders $i$, if all other bidders report truthfully, then the social welfare of bidder $i$'s marginal economy is independent of his value reports. If both assumptions hold, then bidding truthfully is an ex-post Nash equilibrium in MLCA.

\section{Monotone-Value Neural Networks}\label{sec:appendix_MVNNs}
In this section, we provide proofs for all mathematical claims made in \Cref{subsec:Theoretical Analysis and MIP-Formulation} (\Cref{subsec:appendix_proofLemmaUniversality}, \ref{subsec:appendix_proofTheoremUniversality} and \ref{subsec:appendix_proofLemmaMIP}), provide a toy example on MVNNs (\Cref{subsec:appendix_example_mvnns}) and give an overview of bounds tightening via \emph{interval arithmetic} for the MVNN-based MILP (\Cref{app:sec_util:ia_bounds}).
\subsection{Proof of \Cref{lem:normalized and monoton}}\label{subsec:appendix_proofLemmaUniversality}
\setcounter{lemma}{0}
\begin{lemma}\label{app_lem:normalized and monoton}
    Let $\MVNNi{}:\X\to \Rp$ be an MVNN from \Cref{def:MVNN}. Then it holds that  $\MVNNi[(W^{i},b^{i})]{}\in\Vmon$ for all $W^{i}\ge0$ and $b^{i}\le0$.
\end{lemma}
\begin{proof}\

\begin{enumerate}
    \item \textbf{Monotonicity (M)}:\\
    This property immediately follows, since (component wise) the weights $W^{i,k}\ge0$ for all $k\in \{1\ldots,K_i\}$ and $\phiu_{0,t}(z)$ is monotonically increasing.
    \item \textbf{Normalization (N)}:\\
    Since $\phiu_{0,t}(z)=0$ for $z\le0$ and the biases fulfill (component wise) $b^{i,k}\le0$, we can conclude that $\MVNNi{}((\underbrace{0,\ldots,0}_{m\text{-times}}))=0$.
\end{enumerate}
\end{proof}

\begin{remark}[The role of the activation-function in \Cref{app_lem:normalized and monoton}]\label{rem:leMonotoneANdNormalized}
For \Cref{app_lem:normalized and monoton} it would be sufficient to only assume that the activation function is monotonically increasing to get (M) and maps negative numbers to zero to get (N). So ReLU would also be a valid activation function for \Cref{app_lem:normalized and monoton}, but not for \Cref{app:thm:Universality} (see \Cref{rem:thmUniversality}).
\end{remark}

\subsection{Example MVNNs}\label{subsec:appendix_example_mvnns}
The following example illustrates how we can capture complementarities, substitutabilities and independent items via an MVNN.
\begin{example}[MVNN]\label{app:example_mvnns}
Consider the set of items $M = \{x_1, x_2, x_3\}$ and the associated (reported) value function $\hvi{}$ shown in \Cref{tab:example} (where we use $001$ as a shorthand notation for $(0,0,1)$):

\begin{table}[h]
\centering
\begin{sc}
\resizebox{0.98\columnwidth}{!}{
	\begin{tabular}{
			c
			S[table-format=1.0]
			S[table-format=1.0]
			S[table-format=1.0]
			S[table-format=1.0]
			S[table-format=1.0]
			S[table-format=1.0]
			S[table-format=1.0]
			S[table-format=1.0]
			}
		\toprule
		& {000} & {100} & {010} & {001}& {110} & {101} & {011} & {111}\\
		\cmidrule(lr){2-2}
		\cmidrule(lr){3-5}
        \cmidrule(lr){6-8}
        \cmidrule(lr){9-9}
		$\hvi{}$   & {0}& {\vphantom{-}1} & {\vphantom{-}1}  & {\vphantom{-}1}  & {1} & {3} & {2} & {\vphantom{-}4}\\
		\bottomrule
	\end{tabular}
}
\vskip -0.2cm
\caption{Example on flexibility of MVNNs.}
\label{tab:example}
\end{sc}
\end{table}

In this example, $x_1$ and $x_2$ are \emph{substitutes}, i.e., $2=\hvi{\{x_1\}}+\hvi{\{x_2\}}>\hvi{\{x_1,x_2\}}=1$; $x_1$ and $x_3$ are \emph{complements}, i.e., $2=\hvi{\{x_1\}}+\hvi{\{x_3\}}<\hvi{\{x_1,x_3\}}=3$; and $x_2$ and $x_3$ are \emph{independent}, i.e., $2=\hvi{\{x_2\}}+\hvi{\{x_3\}}=\hvi{\{x_2,x_3\}}=2$. This reported value function can be exactly captured by an MVNN $\MVNNi{x}$ in the following way:

\begin{figure}[h!]
\begin{center}
        \resizebox{0.7\columnwidth}{!}{
                \begin{tikzpicture}
                [cnode/.style={draw=black,fill=#1,minimum width=3mm,circle}]
                \node[cnode=gray,label=180:$x_{1}$](x1) at (0,-1) {};
                \node[cnode=gray,label=180:$x_{2}$](x2) at (0,-2.5) {};
                \node[cnode=gray,label=180:$x_{3}$](x3) at (0,-4) {};
                \node[cnode=gray,label={90:$0$},label={10:$\phiu{}$}] (h1) at (3,-1.75) {};
                \node[cnode=gray,label={90:$-1$},label={-10:$\phiu{}$}] (h2) at (3,-3.25) {};
                \node[cnode=gray,label={90:$0$}](out) at (6,-2.5) {};
                
                \draw (x1)--(h1) node[midway,above] {$1$};
                \draw (x2)--(h1);
                \node at (1.25,-1.65) {$0.5$};
                \draw (x3)--(h1); 
                \node at (.5,-3.4) {$1$};
                
                \draw (x1)--(h2);
                \node at (.9,-2.1) {$1$};
                \draw (x2)--(h2);
                \node at (1,-2.55) {$0.25$};
                \draw (x3)--(h2) node[midway,above] {$1$};
                
                \draw (h1)--(out) node[midway,above] {$1$};
                \draw (h2)--(out) node[midway,above] {$4$};
                
                \node at (1.5,-0.5) {$W^{i,1}(\cdot)+b^{i,1}$};
                \node at (4.5,-0.5) {$W^{i,2}(\cdot)$};
                \end{tikzpicture}
        }
        \end{center}
        \vskip -0.2cm
        \caption{MVNN $\protect\MVNNi{}$ with $\protect\MVNNi{x}=\protect\hvi{x}\quad \forall x\in \X$.}
        \label{fig:mvnn_tik}
\end{figure}
Biases are marked at the top of each neuron and weights are marked above the corresponding connections. A missing connection denotes a weight of 0. We see that the second kink (at $t=1$) of the \emph{bReLU} together with the $0$ bias of the top neuron in the hidden layer implements the substitutability between $x_1$ and $x_2$. Furthermore, the complementarity between $x_1$ and $x_3$ is implemented by the bottom neuron in the hidden layer via the negative bias $-1$ and the first kink (at $0$) of the \emph{bReLU}. 
\end{example}

\subsection{Proof of \Cref{thm:universality}}\label{subsec:appendix_proofTheoremUniversality}
\setcounter{theorem}{0}
\begin{theorem}[Universality]\label{app:thm:Universality}\emph{Any} value function $\hvi{}:\X\to\Rp$ that satisfies \textbf{(N)} and \textbf{(M)} can be represented exactly as an MVNN~$\MVNNi{}$ from \Cref{def:MVNN}, i.e.,
{
\begin{align}
\Vmon=\left\{\MVNNi[(W^i,b^i)]{}: W^{i}\ge0, b^{i}\le0\} \right\}.
\end{align}}
\end{theorem}
\begin{proof}\label{appproof:thm:Universality}\
\begin{enumerate}
    \item {\small $\Vmon \supseteq \left\{\MVNNi[(W^i,b^i)]{}: W^{i,k}\ge0, b^{i,k}\le0\,\, \forall k \in \{1,\ldots,K_i\} \right\}$}\\
    
    This direction follows immediately from \Cref{app_lem:normalized and monoton}.\\
    
    \item {\small$\Vmon \subseteq \left\{\MVNNi[(W^i,b^i)]{}: W^{i,k}\ge0, b^{i,k}\le0\,\, \forall k \in \{1,\ldots,K_i\} \right\}$}\\
    
    Let $(\hvi{x})_{x\in \X} \in \Vmon$.
    For the reverse direction, we give a constructive proof, i.e., we construct an MVNN $\MVNNi{}$ with $\theta=(W^i_{\hvi{}},b^i_{\hvi{}})$ such that $\MVNNi{x}=\hvi{x}$ for all $x\in \X$.
    
    Let $(w_j)_{j=1}^{2^m}$ denote the values corresponding to $(\hvi{x})_{x\in \X}$ sorted in increasing order, i.e, let
    $x_1=(0,\ldots,0)$ with 
    \begin{align}\label{eq:w_1}
    w_1:=\hvi{x_1}=0,
    \end{align}
    let $x_{2^m}=(1,\ldots,1)$ with \begin{align}\label{eq:w_2m}
    w_{2^m}:=\hvi{x_{2^m}},
    \end{align}
    and $x_j,x_l \in \X\setminus\{x_1,x_{2^m}\}$\ for $1< l \le j \le 2^m-1$ with
    \begin{align}\label{eq:w_remaining}
    w_j:=\hvi{x_j}\, \leq\, w_l:=\hvi{x_l}.
    \end{align}
    In the following, we slightly abuse the notation and write for $x_l,x_j\in \X$ $x_l\subseteq x_j$ iff for the corresponding sets $A_j, A_l\in 2^M$ it holds that $A_j\subseteq A_l$. Furthermore, we denote by $ \left<\cdot,\cdot\right>$ the Euclidean scalar product on $\R^m$. Then, for all $x\in\X$:
    {\small
    \begin{align}
    &\hvi{x}
    =\hspace{-0.05cm}\sum_{l=1}^{2^m-1}\hspace{-0.05cm} \left(w_{l+1}-w_{l}\right)\1{\forall j\in\{1,\dots,l\}\,:\,x\not\subseteq x_{j}}\\
    &=\sum_{l=1}^{2^m-1}\hspace{-0.05cm} \left(w_{l+1}-w_{l}\right)		  \phiu_{0,1}{}\hspace{-0.05cm}\left(\sum_{j=1}^{l}	\phiu_{0,1}{}\left( \left<1-x_{j},x\right>\right)-(l-1)\right)\label{eq:last_term},	      
    \end{align}
    }%
    where the second equality follows since
    {\small
    \begin{align}
        x\not\subseteq x_{j} & \iff \left<1-x_{j},x\right> \ge1\\
        &\iff \phiu_{0,1}{}\left( \left<1-x_{j},x\right>\right)=1,
    \end{align}
    }%
    which implies that
    {\small
    \begin{align}
        \forall j\in\{1,\dots,l\}: x\not\subseteq x_{j} &\iff \sum_{j=1}^{l}\phiu_{0,1}{}\left( \left<1-x_{j},x\right>\right)=l,
    \end{align}
    }%
    and
    {\small
    \begin{align}
        \1{\forall j\in\{1,\dots,l\}\,:\,x\not\subseteq x_{j}}=\phiu_{0,1}{}\left(\sum_{j=1}^{l}\phiu_{0,1}{}\left( \left<1-x_{j},x\right>\right)-(l-1)\right)
    \end{align}
    }%
    
    Finally, \Cref{eq:last_term} can be equivalently written in matrix notation as
    
    \resizebox{!}{0.18\columnwidth}{\parbox{\columnwidth}{
    \begin{align*}
    \hspace{-0.1cm}\underbrace{\begin{bmatrix} w_2 - w_1  \\ w_3 - w_2  \\ \vdots \\ w_{2^m} - w_{2^m - 1} \end{bmatrix}^T}_{(W^{i,3}_{\hvi{}})^T\in \R_{\ge0}^{2^m-1}}\hspace{-0.3cm}\phiu_{0,1}{}\hspace{-0.1cm}\left(\hspace{-0.1cm}W^{i,2}_{\hvi{}} \phiu_{0,1}{}\hspace{-0.1cm}\left(\underbrace{\begin{bmatrix} 1- x_{1} \\1- x_{2} \\ \vdots \\ 1 - x_{2^m- 1}\end{bmatrix}}_{W^{i,1}_{\hvi{}}\in \R_{\ge0}^{(2^m-1)\times m }} \hspace{-0.25cm}x\hspace{-0.05cm}\right) + \underbrace{\begin{bmatrix}  0\\-1\\ \vdots \\ -(2^m - 2)\end{bmatrix}}_{b^{i,2}_{\hvi{}}\in \R_{\le0}^{2^m-1}}\hspace{-0.1cm}\right)
    \end{align*}
    }}%
    
    with $W^{i,2}_{\hvi{}}\in \R_{\ge0}^{(2^m-1)\times (2^m-1)}$ a lower triangular matrix of ones, i.e.,
    {\scriptsize
    \begin{align*}
    W^{i,2}_{\hvi{}}:=\begin{bmatrix}
    1       &  0 &  \ldots &0      \\
    \vdots  &\ddots  & \ddots  &\vdots     \\
    \vdots   &  &  \ddots &0     \\
    1       &   \ldots     & \ldots & 1
    \end{bmatrix}.
    \end{align*}
    }%
    From that, we can see that the last term is indeed an MVNN $\MVNNi{x}=W^{i,3}_{\hvi{}} \phiu_{0,1}{}\left(W^{i,2}_{\hvi{}} \phiu_{0,1}{}\left(W^{i,1}_{\hvi{}} x\right) + b^{i,2}_{\hvi{}}\right)$ with four layers in total (i.e., two hidden layers) and respective dimensions $[m,2^m-1,2^m-1,1]$.
\end{enumerate}
\end{proof}

\begin{remark}[The role of the activation-function in \Cref{app:thm:Universality}]\label{rem:thmUniversality}
If the classical ReLU-activation function was used, we would not get universality, because of the non-negativity constraints of the weights:
A ReLU-network with non-negative weights can only express convex functions.\footnote{A ReLU-MVNN could only express convex functions, since ReLU is convex and non-decreasing, and linear combinations of convex non-decreasing functions are convex and non-decreasing if all the coefficients are non-negative, and \href{https://math.stackexchange.com/a/287725/688715}{compositions of convex non-decreasing functions are convex} (and non-decreasing). Thus any ReLU-MVNN with our weight constraints could never express any non-convex functions. However, agent's value functions are very often non-convex, i.e., they typically admit substitutabilities. E.g.~\Cref{app:example_mvnns} could not be expressed by a ReLU-MVNN.}
This means that it could only express complementarities, but no substitutabilities.
For approximate universality any monotonically increasing, bounded, non-constant activation-function would be sufficient (e.g., bReLU or sigmoid). However, sigmoid would not be a valid activation function for \Cref{app:thm:Universality}, because it would not allow \emph{exact} normalization for a non-zero function. (Firstly, for an activation function such as sigmoid that does not map negative number to zero, networks that fulfil the sign constraints for the weights and biases could be arbitrarily far away from being normalized. Secondly, for universality, one could only approximate the true function up to an arbitrarily small epsilon but never exactly for a normalized non-constant true function.) While these problems of sigmoid-activation-functions for \Cref{app:thm:Universality} might be rather theoretical, what is more important in practice is that non-piecewise-linear activations functions are not MILP-formalizable.
\end{remark}

\begin{remark}[Why bReLU?]\label{rem:WhybReLU}
Every monotonically increasing, non-constant, bounded activation function that maps negative numbers to 0 would be a valid choice, since it would fulfill \Cref{app_lem:normalized and monoton} (see \Cref{rem:leMonotoneANdNormalized}), \Cref{app:thm:Universality} (see \Cref{rem:thmUniversality}) and be MILP-formalizable (cp.
\Cref{theo:util_mip}). The complexity of the MILP grows with the number of kinks in the activation function. In this sense, bReLU is the simplest monotonically increasing, non-constant, bounded activation function that maps negative numbers to 0.\footnote{One can easily see that any bounded piece-wise linear non-constant activation function has to have at least two kinks, so at least 3 linear segments. bReLU has two kinks (one at zero and one at the cut-off $t$), so there is no other suitable activation function with less kinks.} In addition, bReLU shares many properties of ReLU and sigmoid, which are both popular choices for activation-functions.
\Cref{rem:thmUniversality} explains why ReLU or sigmoid would not be valid choices.
\end{remark}

\begin{corollary}\label{app:cor:interpolation}
For every dataset $\left( (x_j),\hvi{x_j}\right)_{j\in\{1\dots,q\}}$ there exist an MVNN of dimensions $[m,q,q,1]$ that perfectly fits the dataset, i.e. $\MVNNi{x_j}=\hvi{x_j} \  \forall j \in \{1\dots,q\}$.
\end{corollary}
\begin{proof}
This interpolating MVNN can be explicitly constructed such as the MVNN in the \hyperref[appproof:thm:Universality]{proof} of \Cref{app:thm:Universality} by replacing the collection of all bundles and their values by the dataset $\left( (x_j),\hvi{x_j}\right)_{j\in\{0\dots,q\}}$, where $x_0$ is the empty bundle and $\hvi{x_0}=0$, i.e., \Cref{eq:last_term} has to be replaced by:
{\small
    \begin{align}
    \MVNNi{x}=\sum_{l=1}^{q}\hspace{-0.05cm} \left(w_{l}-w_{l-1}\right)		  \phiu_{0,1}{}\hspace{-0.05cm}\left(\sum_{j=0}^{l-1}	\phiu_{0,1}{}\left( \left<1-x_{j},x\right>\right)-(l-1)\right)\label{eq:last_term_cor:interpolation},
    \end{align}}%
    where we again assume that the bundles are sorted by their values and for $0< l \le j \le q$ we define:
    \begin{align}
    w_j:=\hvi{x_j}\, \leq\, w_l:=\hvi{x_l}.
    \end{align}

\end{proof}

\begin{remark}[Number of neurons needed]
$[m,2^m-1,2^m-1,1]$ is just an upper bound for the size of the network to learn the true value function exactly on $\X$. For NNs on a continuous domain no finite upper bound exist, but they are still widely used in practice. Most value functions could be expressed with significantly fewer nodes.

More importantly, in an ICA (e.g., MLCA) one doesn't need to learn the value function everywhere perfectly but only approximately where the precision mainly matters for relevant bundles. \Cref{app:cor:interpolation} states that for $q$ queries, there is an MVNN with size $[m,q,q,1]$ that can perfectly fit through these $q$ queries. So we have a mathematical guarantee that the number of nodes only grows \emph{linearly} with the number of queries in the worst case. This is the much more relevant growth rate in practice, because we will typically not be able to do $2^m-1$ queries anyway and thus the MILP will also never be exponentially large.
Thus, for MVNNs with size $[m,100,100,1]$, we have a mathematical guarantee that we can always get a perfect monotonic normalized interpolation of all available 100 queries even in the worst cases.
However, our experiments show that for data coming from realistic value functions, significantly smaller architectures are sufficient to (almost) perfectly interpolate the training data most of the time and even often almost perfectly approximate the true value function on $\X$ (see \Cref{tab:full_pred_performance_table_appendix,tab:pred_perf_plots_appendix}) leading to better than state-of-the-art performance in MLCA (see \Cref{tab:efficiency_loss_mlca,tab:efficiency_loss_mlca_appendix}). This allows very fast MILP-solving times in practice (see \Cref{tab:efficiency_loss_mlca_appendix}, where the average run-times of full auctions are given in hours or \Cref{fig:mip_runtime_analysis}, where the average run-times per MILP are given in seconds).
\end{remark}

\subsection{Proof of \Cref{lem:util_layer}}\label{subsec:appendix_proofLemmaMIP}
In this section, we proove \Cref{lem:util_layer} from \Cref{sec:Utility Networks}.
\begin{proof}
For all $j \in \{1,\ldots,d^{i,k}\}$ we distinguish the following three cases:
\begin{align}
     \textrm{\textbf{Case 1:~}}& o^{i,k}_j \in (-\infty, 0] \Rightarrow \phiu(o^{i,k})_j = 0 \notag\\
    & \text{\eqref{eq:(i)}} \And \text{\eqref{eq:(ii)}} \Rightarrow y^{i,k}_j = 1 \Rightarrow \eta^{i, k}_j = 0 \notag \\
    & \text{\eqref{eq:(iii)}} \And \text{\eqref{eq:Lemma_cutoff_dependent_ct}} \Rightarrow \mu^{i,k}_j = 0 \Rightarrow z^{i, k}_j = \eta^{i, k}_j = 0 \notag \\
     \textrm{\textbf{Case 2:~}}& o^{i,k}_j \in (0, 1] \Rightarrow \phiu(o^{i,k})_j = o^{i,k}_j \notag\\
    & \text{\eqref{eq:(i)}} \And \text{\eqref{eq:(ii)}} \Rightarrow y^{i,k}_j = 0 \Rightarrow \eta^{i, k}_j = o^{i,k}_j \notag \\
    & \text{\eqref{eq:(iii)}} \And \text{\eqref{eq:Lemma_cutoff_dependent_ct}} \Rightarrow \mu^{i,k}_j = 0 \Rightarrow z^{i, k}_j = \eta^{i, k}_j = o^{i,k}_j \notag \\
     \textrm{\textbf{Case 3:~}}& o^{i,k}_j \in (1, +\infty] \Rightarrow \phiu(o^{i,k})_j = 1 \notag\\
    & \text{\eqref{eq:(i)}} \And \text{\eqref{eq:(ii)}} \Rightarrow y^{i,k}_j = 0 \Rightarrow \eta^{i, k}_j = o^{i,k} \notag \\
    & \text{\eqref{eq:(iii)}} \And \text{\eqref{eq:Lemma_cutoff_dependent_ct}} \Rightarrow \mu^{i,k}_j = 1 \Rightarrow z^{i, k}_j = 1 \notag
\end{align}
Thus, in total $z^{i, k}=\phiu(o^{i,k})$.
\end{proof}

\subsection{Interval Arithmetic Bounds Tightening for MVNNs}\label{app:sec_util:ia_bounds}

In this section, we consider a bReLU $\phiu_{0,t}$ with cutoff $t>0$ and mark it in red. This helps when implementing the MILP, i.e. it particularly shows more clearly where the cutoff $\cutoff$ propagates in the respective equations. First, we recall \Cref{lem:normalized and monoton} for a general cutoff $\cutoff$, where for the sake of readability we remove the bidder index $i\in N$ from all variables.

\setcounter{lemma}{1}
\begin{lemma}\label{lem:util_layer_appendix_abbreviated}
    Let $k\in \{1,\ldots,K-1\}$ and let the pre-activated output of the $k$\textsuperscript{th} layer be given as $W^{k}z^{k-1} + b^{k}$ with $W^{k} \in \mathbb{R}^{d^{k} \times d^{k-1}}, b^{k} \in \mathbb{R}^{d^{k}}$. Then the output of the $k$\textsuperscript{th} layer  $z^{k} := \phiu_{0,t}(W^{k}z^{k-1} + b^{k}) = \min(\cutoff, \max(0, (W^{k}z^{k-1} + b^{k})) = -\max(-\cutoff, - \eta^{k})$, with $\eta^{k}:=\max(0,W^{k}z^{k-1} + b^{k} )$ can be equivalently expressed by the following linear constraints:
    {\small
    \begin{align}
        \hspace{-0.3cm} W^{k}z^{k-1}+b^{k}&\le \eta^{k}\le W^{k}z^{k-1}+b^{k} + y^{k}L_1^{k}\label{eq:MIP_lemma_1}\\
        0&\le \eta^{k} \le (1-y^{k})L_2^{k}\label{eq:MIP_lemma_2}\\
         \eta^{k} -  \mu^{k}L_3^{k}&\le z^{k}\le \eta^{k}\label{eq:MIP_lemma_3}\\
         \cutoff -  (1-\mu^{k})L_4^{k} &\le z^{k} \le \cutoff\label{eq:MIP_lemma_4}
    \end{align}
    }%
    where $y^{k}\in \{0,1\}^{d^k}$, $\mu^{k}\in \{0,1\}^{d^k}$, and $L_1^{k}, L_2^{k}, L_3^{k}, L_4^{k} \in \Rp$ are large enough constants for the respective \emph {big-M} constraints.
\end{lemma}

\subsubsection{Interval Arithmetic (IA) for Single Neurons}\label{app:subsubsec_Interval Arithmetic (IA) for Single Neurons}

For any  $k\in \{1,\ldots,K-1\}$ let $\Wpre\ge0$ and $\bpre\le0$ be the weights of its affine linear transformation. Furthermore, let $\zpre$ be the output of the previous layer.
and, let 
\begin{align}
\zaft=\phiu_{0,t}\left(\Wpre\zpre+\bpre\right)
\end{align}
be the output of the current layer.

Given already computed IA bounds for $\zpre_i$, i.e., $[L(\zpre_l), U(\zpre_l)]\subseteq[0,\cutoff]$, we can then calculate the IA bounds $L(\zaft_l), U(\zaft_l)$ for $\zaft_l$ such that 
\begin{align}
\zaft \in \prod_{l=1}^{\daft}[L(\zaft_l), U(\zaft_l)]
\end{align}
as follows:

\begin{enumerate}[leftmargin=*,topsep=5pt,partopsep=0pt, parsep=0pt]
    \item \textbf{Upper bound (pre-activated):}
    {\small
    \begin{align}
    U^{\text{pre}}(\zaft)&=\\
    &=\hspace{-1.5cm}\max_{\zpre \in \prod_{l=1}^{\dpre}[L(\zpre_l), U(\zpre_l)]} \left \{ \Wpre\zpre + \bpre \right \}\\
     &= \left(\sum_{l} \Wpre_{j,l}\cdot U(\zpre_l)\right)_{j=1}^{\daft} +  \bpre
    \end{align}
    }
      \item \textbf{Upper bound:}
      {\small
      \begin{align}
      U(\zaft)=\phiu_{0,t}\left(U^{\text{pre}}(\zaft)\right)
      \end{align}
      }
      \item \textbf{Lower bound (pre-activated):}
      {\small
    \begin{align}
    L^{\text{pre}}(\zaft)&=\\
    &=\hspace{-1.5cm}\min_{\zpre \in \prod_{l=1}^{\dpre}[L(\zpre_l), U(\zpre_l)]} \left \{ \Wpre\zpre + \bpre \right \}\\
    &=
      \left(\sum_{l} \Wpre_{j,l}\cdot L(\zpre_l)\right)_{j=1}^{\daft} +  \bpre
      \end{align}
      }
      Note, that the preactivated lower bound $L^\text{pre}(\zaft)$ is always non-positive, i.e. $L^\text{pre}(\zaft)\le0$. This can be seen as follows: start with a $L(z^1)=\boldsymbol{0}_{d^1}$ 
      . Then since all biases are non-positive $L^{\text{pre}}(z^2)\le0 \implies L(z^2)=\boldsymbol{0}_{d^2}$, and so forth. Thus, we get for the lower bounds:
      \item \textbf{Lower bound:}
      {\small
      \begin{align}
      L(\zaft)=\phiu_{0,t}\left(L^{\text{pre}}(\zaft)\right)=\boldsymbol{0}_{d^k}
      \end{align}
      }
\end{enumerate}


\subsubsection{Removing Constraints with IA}\label{app:subsubsec_Removing constraints with IA}
In the following cases, we can remove the constraints and corresponding variables in \Cref{lem:util_layer_appendix_abbreviated}:
\begin{enumerate}[leftmargin=*,topsep=5pt,partopsep=0pt, parsep=0pt]
    \item \textbf{Case:}\\
    If $U^{\text{pre}}(\zaft_l)\le0 \implies \zaft_l=0$ and one can remove the $l$\textsuperscript{th} components from all constraints \eqref{eq:MIP_lemma_1} -- \eqref{eq:MIP_lemma_4} and the corresponding variables for layer $k$.
    \item \textbf{Case:}\\
    If $L^{\text{pre}}(\zaft_l)\in [0,\cutoff)$ and $U^{\text{pre}}(\zaft_l)\in(0,\cutoff] \implies \zaft_l=\left(\Waft\zpre + \baft\right)_l$ and one and can remove the $l$\textsuperscript{th} components from all constraints \eqref{eq:MIP_lemma_1} -- \eqref{eq:MIP_lemma_4} and the corresponding variables for layer $k$.
    \item \textbf{Case:}\\
    If $L^{\text{pre}}(\zaft_l)=0$ and $U^{\text{pre}}(\zaft_l)>\cutoff \implies \eta^k_l=\left(\Waft\zpre + \baft\right)_l$ and one and can remove the $l$\textsuperscript{th} components from all constraints \eqref{eq:MIP_lemma_1} -- \eqref{eq:MIP_lemma_2} and the corresponding variables for layer $k$.
    \item \textbf{Case:}\\
    If $L^{\text{pre}}(\zaft_l)<0$ and $U^{\text{pre}}(\zaft_l)\in(0,\cutoff] \implies \zaft_l=\eta^k_l$ and one and can remove the $l$\textsuperscript{th} components from all constraints \eqref{eq:MIP_lemma_3} -- \eqref{eq:MIP_lemma_4} and the corresponding variables for layer $k$.
    \footnote{In all experiments presented in this paper, we have not taken advantage of Case 3 and 4. We believe that we can further improve the computational performance of MVNN by incorporating these cases too.
    }
     \end{enumerate} 

\begin{fact}
For a plain ReLU NN, the above calculated IA bounds are not tight and calculating tighter bounds in a computationally efficient manner is very challenging and an open research question.
However, for MVNNs, the IA bounds are always \emph{perfectly tight}, because of their encoded monotonicity property. This is a big advantage of MVNN-based MILPs compared to plain (ReLU) NN-based MILPs.
\end{fact}

\subsubsection{Interval Arithmetic for the Four Big-M Constraints}\label{app:subsubsec_Interval Arithmetic for the 4 Big-M constraints}

Recall, that $L(\zaft)=\boldsymbol{0}_{\daft}$ for all $k\in \{1,\ldots,K-1\}$. Let $m=:d^0$ denote the number of items to be allocated. 

We present standard IA bounds where one starts for $z^0\in \{0,1\}^m$ with $L(z^0)=\boldsymbol{0}_m$ and $U(z^0)=\boldsymbol{1}_m$ (per definition) and iteratively for $k \in \{1,\ldots,K-1\}$ \emph{propagates} through the network for given bounds on $\zpre$, i.e., $\zpre\in \mathcal{Z}^{k}:=\prod_{l=1}^{\dpre}[0, U(\zpre_l)]\subset[0,\cutoff]^{\dpre}$.

\begin{enumerate}[leftmargin=*,topsep=5pt,partopsep=0pt, parsep=0pt]
\item $L^k_1$ only appears when $y^k = 1$, which implies that $\eta^k = 0$ and $\Wpre\zpre+\bpre \le 0 $. Thus, \Cref{eq:MIP_lemma_1} implies that $L^k_1 + \min_{\zpre\in \mathcal{Z}^{k-1}} \left \{ \Wpre\zpre+\bpre \right \}  \ge 0$ and we get

{\small
\begin{align}
L^k_1 =& \max \left \{0, - \min_{\zpre\in \mathcal{Z}^{k-1}} \left \{ \Wpre\zpre+\bpre  \right \}  \right \}\\
=&\max \left \{  0,  \left(\sum_{l: \Wpre_{j,l} < 0} |\Wpre_{j,l}|\cdot U(\zpre_l) \right)_{j=1}^{\daft} \hspace{-0.6cm}- \bpre \right \}\\
=&\max \left \{  0, - \bpre \right \}
\end{align}
}
where the last equality follows since per definition any MVNN only has positive weights.

\item $L^k_2$ only appears when $y^k = 0$ which implies that $\eta^k = \Wpre\zpre+\bpre \ge  0 $. Thus, \Cref{eq:MIP_lemma_2} implies that $L^k_2 \ge \eta^k = \Wpre\zpre+\bpre$ and we get

{\small
    \begin{align}
    L^k_2 =& \max_{\zpre \in \mathcal{Z}^{k-1}} \left \{ \Wpre\zpre+\bpre \right \}\\
    =&\max \left \{ 0, \left(\sum_{l: \Wpre_{j,l} > 0} \Wpre_{j,l}\cdot U(\zpre_l)\right)_{j=1}^{\daft} + \bpre \right \}\\
    =&\max \left \{ 0, \left(\sum_{l} \Wpre_{j,l} \cdot U(\zpre_l)\right)_{j=1}^{\daft} + \bpre  \right \},
    \end{align}
}%
where the last equality follows since per definition any MVNN only has positive weights.

\item $L^k_3$ only appears when $ \mu^k  = 1$ which implies that $ \zaft = \cutoff$ and $\eta^k = \Wpre \zpre + \bpre  \ge \cutoff $. Thus, \Cref{eq:MIP_lemma_3} implies that $\cutoff = \zaft \ge \eta^k - L^k_3 \iff L^k_3 \ge \eta^k - \cutoff \iff L^k_3 \ge \Wpre \zpre + \bpre - \cutoff$ and we get
{\small
\begin{align}
L^k_3&= \max_{\zpre \in \mathcal{Z}^{k-1}} \left \{ \Wpre \zpre + \bpre - \cutoff \right \}\\
&=\max \left \{ 0, L^k_2 - \cutoff  \right \}
\end{align}
}

\item For $L^k_4$ only appears when $ \mu^k  = 0$. In this case, we get from  \Cref{eq:MIP_lemma_4} that $\cutoff - L^k_4  \le  \zaft$ and we get
{\small
\begin{align}
L^k_4 = \max_{\zaft \in \mathcal{Z}^k}\{ \cutoff - \zaft\} = \cutoff
\end{align}
}
\end{enumerate}

\subsection{Implementation Details - MVNNs}\label{subsec:appendix_Implementation Details_MVNNs}
Recall that the three key building blocks of MVNNs are the bReLU $\phiu{}$, the (element-wise) non-negative weights $W^i$, and the (element-wise) non-positive biases $b^i$. The bReLU $\phiu{}$ can be straightforwardly implemented as a custom activation function in \textsc{PyTorch} 1.8.1. However, the constraints on the weights and biases to be element-wise positive and negative respectively, can be implemented in several different ways. We first describe the method we call \textsc{MVNN-ReLU-Projected}, which we have found experimentally to perform best. Finally, we describe the other options we have explored since they may be of independent interest to some readers.

\paragraph{\textsc{MVNN-ReLU-Projected}} In this implementation, we project $W^i$ and $b^i$ prior to every forward pass to be non-negative and non-positive, respectively, using ReLU $z\mapsto\pm\max(0,\pm z)$. Thus, (in contrast to some of the other methods we explored), gradient descent (GD) updates are performed on the already transformed weights and biases and we \emph{do not differentiate} through the ReLUs. After the last GD step, we apply \emph{post-processing} and project $W^i$ and $b^i$ again via $z\mapsto\pm\max(0,\pm z)$.

\paragraph{\textsc{MVNN-Abs} and \textsc{MVNN-ReLU}} For these methods, we add an additional node to the computational graph and element-wisely transform the weights $W^i$ via the absolute value $z\mapsto |z|$ or ReLU $z\mapsto max(0, z)$ to ensure non-negativity. For the biases $b^i$ we analogously use $z\mapsto -|z|$ or $z\mapsto -max(0, -z)$ to ensure non-positivity. Importantly, we \emph{differentiate} through $z\mapsto\pm|z|$ or $z\mapsto\pm\max(0,\pm z)$ in every gradient descent (GD) step. We refer to these implementation variants as \textsc{MVNN-Abs} and \textsc{MVNN-ReLU}. After the last GD step we apply \emph{post-processing} and project $W^i$ and $b^i$ again via $z\mapsto\pm|z|$ and $z\mapsto\pm\max(0,\pm z)$.

\paragraph{\textsc{MVNN-Abs-Projected}} We do not consider a \textsc{MVNN-Abs-Projected} implementation, since one can prove that in the classic GD algorithm this version is mathematically equivalent to \textsc{MVNN-Abs}.

\section{Details Prediction Performance}\label{sec:appendix_pred_perf}

\subsection{Data Generation - Prediction Performance}\label{subec:appendix_pred_perf_data_gen} 
\paragraph{Train/Val/Test-Split}
For every domain, bidder type and considered amount of training data $T$ we create the data in the following way:
For each seed we let SATS create a value function and uniformly at random select $T$ different bundles from the bidder-specific-feasible (using the SATS method \textit{get\_uniform\_random\_bids}) bundle space $\mathcal{X}$ (training set).
We measure the metrics of a method trained on the training data of a seed based on randomly selected different bundles from the same bundle space~$\mathcal{X}$ (approx.\ $52\,000$ for the HPO seeds (validation set) and approx.\ $210\,000$ for the test seeds (test set)). 
We use HPO seeds 0-20 only for the HPO (validation sets) and test seeds 21-50 only for reporting the values in this paper (test sets).



\subsection{HPO - Prediction Performance}\label{subec:appendix_pred_perf_hpo}
Table~\ref{tab:hpo_space} shows the hyperparameter ranges from our HPO. Cosine Annealing~\cite{Loshchilov2017SGDRSG} was used to decay the learning to zero within the number of training epochs. Occasionally, the neural networks (both MVNNs and plain NNs) diverged during training (determined if \emph{the Pearson correlation coefficient} $r_{xy}<0.9$ on the train set) hence we added $20$ retries to the training to make sure that we obtain a valid network. 

Experimentally, we found that MVNNs need the training data to be normalized to be within [0, 1], as a result of the bReLU activation functions bounded output. On the other hand, plain NNs worked better if the data was normalized to be within [0, 500] for all considered SATS domains. 

The computational budget of the HPO was 12 hours. 
All experiments were conducted on a compute cluster running Debian GNU/Linux 10 (buster) with Intel Xeon E5-2650 v4 2.20GHz processors with 24 cores and 128GB RAM and Intel E5 v2 2.80GHz processors with 20 cores and 128GB RAM and Python 3.7.10.

\begin{table}[t!]
\resizebox{1\linewidth}{!}{
\begin{tabular}{@{}llll@{}}
\toprule
\textbf{Hyperparameters}    & \textbf{Type}        & \textbf{Range}             & \textbf{Scale} \\ \midrule
Optimizer          & Categorical & {[}Adam, SGD{]}   &       \\
Batch Size         & Integer     & {[}1, 4{]}        &       \\
Num. Hidden Layers & Integer     & {[}1, 3{]}        &       \\
Total Num. Neurons & Integer     & {[}1, 64{]}       &       \\
L2                 & Float       & {[}1e-10, 1e-6{]} & log   \\
Learning Rate      & Float       & {[}1e-4, 1e-2{]}  & log   \\
Epochs             & Integer     & {[}50, 400{]}     &       \\
Loss Function      & Categorical & {[}MSE, MAE{]}    &       \\ \bottomrule
\end{tabular}}
\vskip -0.2cm
\caption{HPO space used in SMAC.}
\label{tab:hpo_space}
\end{table}

\begin{table*}[htbp]
	\robustify\bfseries
	\centering
	\begin{sc}
	\resizebox{1\textwidth}{!}{
    \begin{tabular}{crrrrcccccc}
    \toprule
         &  & & & &  \multicolumn{2}{c}{\textbf{Efficiency Loss in \%\,\,\textdownarrow}} & \multicolumn{2}{c}{\textbf{Revenue in \%\,\,\textuparrow}}  &\multicolumn{2}{c}{\textbf{Runtime in Hrs.}}   \\
        \cmidrule(l{2pt}r{2pt}){6-7}
        \cmidrule(l{2pt}r{2pt}){8-9}
        \cmidrule(l{2pt}r{2pt}){10-11}
    \textbf{Domain} & \textbf{T}&$\boldsymbol{\Qinit}$&$\boldsymbol{\Qround}$&$\boldsymbol{\Qmax}$ &    \multicolumn{1}{c}{MVNN}  &   \multicolumn{1}{c}{NN}    & \multicolumn{1}{c}{MVNN} & \multicolumn{1}{c}{NN}    & \multicolumn{1}{c}{MVNN} & \multicolumn{1}{c}{NN}\\
        \cmidrule(l{2pt}r{2pt}){1-5}
        \cmidrule(l{2pt}r{2pt}){6-7}
        \cmidrule(l{2pt}r{2pt}){8-9}
        \cmidrule(l{2pt}r{2pt}){10-11}
    GSVM & 10& 40 & 4&  100 & 00.00 $\pm$\scriptsize\, 0.00 & 00.01  $\pm$\scriptsize\, 0.02 &  60.11 $\pm$\scriptsize\, 3.86 & 58.59 $\pm$\scriptsize\, 4.35 & 00.14 $\pm$\scriptsize\, 0.03 & 00.15 $\pm$\scriptsize\, 0.05\\
         &  20& 40& 4& 100& 00.00 $\pm$\scriptsize\, 0.00 & $\llap{*}$00.00 $\pm$\scriptsize\, 0.00 &        59.07 $\pm$\scriptsize\, 3.84 & 55.71 $\pm$\scriptsize\, 4.46  &         00.08 $\pm$\scriptsize\, 0.01 &  00.06 $\pm$\scriptsize\, 0.01 \\
         &  50& 40& 4& 100& $\llap{*}$00.00 $\pm$\scriptsize\, 0.00 & 00.00 $\pm$\scriptsize\, 0.00 &        52.77 $\pm$\scriptsize\, 5.04 & 56.09 $\pm$\scriptsize\, 4.57  &         00.06 $\pm$\scriptsize\, 0.01 &  00.07 $\pm$\scriptsize\, 0.01 \\
         \cmidrule(l{2pt}r{2pt}){1-5}
        \cmidrule(l{2pt}r{2pt}){6-7}
        \cmidrule(l{2pt}r{2pt}){8-9}
        \cmidrule(l{2pt}r{2pt}){10-11}
    LSVM & 10& 40& 4& 100&  01.63 $\pm$\scriptsize\, 0.75 &   03.19 $\pm$\scriptsize\, 01.59 &   70.14 $\pm$\scriptsize\, 4.23 &   65.23 $\pm$\scriptsize\, 3.83 &00.93 $\pm$\scriptsize\, 0.36 &  01.16 $\pm$\scriptsize\, 0.23 \\
         & 50& 40& 4& 100&   $\llap{*}$00.70 $\pm$\scriptsize\, 0.40 & 03.11 $\pm$\scriptsize\, 01.52 &    70.70 $\pm$\scriptsize\, 4.63  & 64.07 $\pm$\scriptsize\, 4.24 & 00.54 $\pm$\scriptsize\, 0.15 &  01.28 $\pm$\scriptsize\, 0.32 \\
         &  100& 40& 4& 100&   01.27 $\pm$\scriptsize\, 0.55 &  $\llap{*}$02.91 $\pm$\scriptsize\, 01.44 &   71.09 $\pm$\scriptsize\, 4.35 & 65.10 $\pm$\scriptsize\, 3.90 & 00.52 $\pm$\scriptsize\, 0.12 &  00.44 $\pm$\scriptsize\, 0.10\\
         \cmidrule(l{2pt}r{2pt}){1-5}
        \cmidrule(l{2pt}r{2pt}){6-7}
        \cmidrule(l{2pt}r{2pt}){8-9}
        \cmidrule(l{2pt}r{2pt}){10-11}
    SRVM &  10& 40& 4& 100& 00.67 $\pm$\scriptsize\, 0.10 & 01.27 $\pm$\scriptsize\, 0.23   & 50.86 $\pm$\scriptsize\, 2.31     &  46.93 $\pm$\scriptsize\, 2.07 & 01.86 $\pm$\scriptsize\, 0.10& 01.59 $\pm$\scriptsize\, 0.09\\
         &  50& 40& 4& 100&  00.49 $\pm$\scriptsize\, 0.07 & 01.15 $\pm$\scriptsize\, 0.31 &   50.70 $\pm$\scriptsize\, 2.44   & 45.59 $\pm$\scriptsize\, 2.72& 20.72 $\pm$\scriptsize\, 0.52&02.17 $\pm$\scriptsize\, 0.21\\
         &  100& 40& 4& 100&   $\llap{*}$00.23 $\pm$\scriptsize\, 0.06  & $\llap{*}$01.13 $\pm$\scriptsize\, 0.22 &  51.18 $\pm$\scriptsize\, 2.52   & 46.30 $\pm$\scriptsize\, 2.87& 20.12 $\pm$\scriptsize\, 1.77&00.68 $\pm$\scriptsize\, 0.06\\
          \cmidrule(l{2pt}r{2pt}){1-5}
        \cmidrule(l{2pt}r{2pt}){6-7}
        \cmidrule(l{2pt}r{2pt}){8-9}
        \cmidrule(l{2pt}r{2pt}){10-11}
    MRVM &  10& 40& 4& 100&  $\llap{*}$08.16 $\pm$\scriptsize\, 0.41 & 10.96 $\pm$\scriptsize\, 0.76 &  34.45 $\pm$\scriptsize\, 2.06   & 41.68 $\pm$\scriptsize\, 1.34& 02.59 $\pm$\scriptsize\, 0.13&06.08 $\pm$\scriptsize\, 0.40\\
         &  100& 40& 4& 100&    08.67 $\pm$\scriptsize\, 0.43 & 10.19 $\pm$\scriptsize\, 0.89 & 35.04 $\pm$\scriptsize\, 1.91  & 41.91 $\pm$\scriptsize\, 1.39& 01.96 $\pm$\scriptsize\, 0.12 &03.33 $\pm$\scriptsize\, 0.37\\
         &  300& 40& 4& 100&    09.96 $\pm$\scriptsize\, 0.49 &  $\llap{*}$09.05 $\pm$\scriptsize\, 0.53 &  32.72 $\pm$\scriptsize\, 2.17 & 39.66 $\pm$\scriptsize\, 1.20 &00.96 $\pm$\scriptsize\, 0.04 &01.84 $\pm$\scriptsize\, 0.24\\
    \bottomrule
    \end{tabular}
}
    \end{sc}
    \vskip -0.2cm
    \caption{Efficiency loss, relative revenue and runtime of MLCA with MVNNs vs MLCA with plain NNs. Shown are averages including a 95\% CI on a test set of $50$ auction instances in all four SATS domains. The best MVNN and plain NN per domain based on the lowest efficiency loss are marked with a star (if the final efficiency loss is the same for multiple incumbents we selected the incumbent that reached 0\% with the fewest number of queries).}
    \label{tab:efficiency_loss_mlca_appendix}
\end{table*}

\subsection{Detailed Results - Prediction Performance}\label{subec:appendix_pred_perf_detailed_results}
Table~\ref{tab:full_pred_performance_table_appendix} shows the detailed prediction performance results including all optimized hyperparameters for the \textsc{MVNN-Abs} and the \textsc{MVNN-ReLU-Projected} implementation of MVNNs. The \textsc{MVNN-ReLU} implementation led to similar results as the \textsc{MVNN-ReLU-Projected} implementation, and therefore we do not present them in this table.

In Table~\ref{tab:pred_perf_plots_appendix}, we visualize for all other SATS domains the prediction performance capabilities of MVNNs vs plain NNs. Overall, Table~\ref{tab:full_pred_performance_table_appendix} and \ref{tab:pred_perf_plots_appendix} show that MVNNs have a superior generalization performance across all SATS domains and bidder types.

\section{Details MVNN-based Iterative CA}\label{sec:appendix_details_incorporating_MVNNs_into_MLCA}

\subsection{Details Experimental Setup -  MVNN-based Iterative CA}\label{subsec:appendix_details_setup_MLCA}
We used 50 auction instances for evaluation with seeds (for generating the SATS instances) 10001-10050 which do not intersect with the seeds used in \Cref{subec:appendix_pred_perf_data_gen} for prediction performance.
All experiments were conducted on the same compute cluster as in \Cref{subec:appendix_pred_perf_hpo}.

Following prior work \cite{brero2021workingpaper} we set $\Qround=4$ (see \Cref{MLCA}), i.e., in each iteration of MLCA we ask each bidder $1$ query in the main economy (including all bidders) and $3$ queries from randomly sampled marginal economies (excluding one bidder). This choice ensures a trade-off between efficiency (more main economy queries decrease the efficiency loss) and revenue (more marginal queries increase revenue). Note that we add the empty bundle to the initial elicited bundles $\Qinit$, as we know its value to be $0$ a priori. This has no impact for MVNNs as they estimate empty bundles to have zero value by definition but adds extra prior information for plain NNs.

For all SATS domains we specified a minimum relative gap of 1e-2 and a timeout of 300s.

\subsection{Detailed Results -  MVNN-based Iterative CA}\label{subsec:appendix_detailed_results_MLCA}
In Table~\ref{tab:efficiency_loss_mlca_appendix}, we present detailed results of MLCA with MVNNs vs MLCA with plain NNs.
The revenue and runtime should be considered with care.
The revenue is influenced by our choice to use early stopping whenever we already found an efficient allocation at an intermediary iteration as discussed in the main paper.
The shown runtime is not comparable between MVNNs and NNs as they do not use the same architecture, but the one found during HPO. See \Cref{subsec:mip_runtime} in the main paper for a fair runtime comparison. 

For the presentation in the main paper (i.e., \Cref{tab:efficiency_loss_mlca} in the main paper), we selected the best MVNN and plain NN per domain, i.e. the ones marked with a star. In Figures \ref{fig:efficiency_loss_path_GSVM}, \ref{fig:efficiency_loss_path_LSVM}, \ref{fig:efficiency_loss_path_SRVM} and \ref{fig:efficiency_loss_path_MRVM} we present a detailed boxplot version of the efficiency loss path plots for these best models per domain.
\begin{figure}[b!]
    \centering
    \includegraphics[width=1\linewidth]{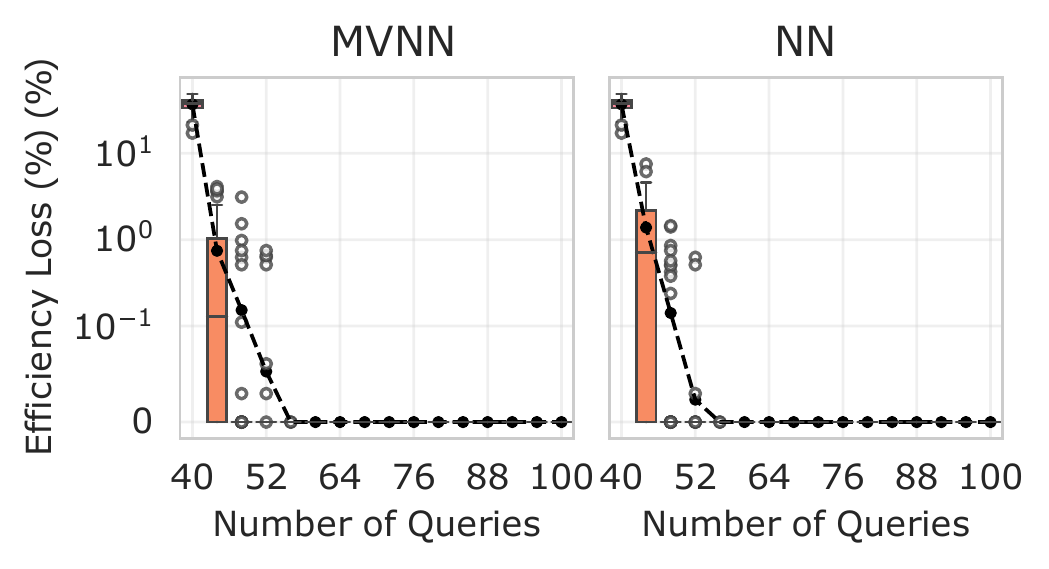}
    \vskip -0.2cm
    \caption{Efficiency loss path of MVNN vs plain NN in GSVM with the corresponding best MVNN ($T=50$) and plain NN ($T=20$) from Table~\ref{tab:efficiency_loss_mlca_appendix}. Averages are shown as black dots. We use a semi-logarithmic scale with linear range [0,1e-2].}
    \label{fig:efficiency_loss_path_GSVM}
\end{figure}
\begin{figure}[t!]
    \centering
    \includegraphics[width=1\linewidth]{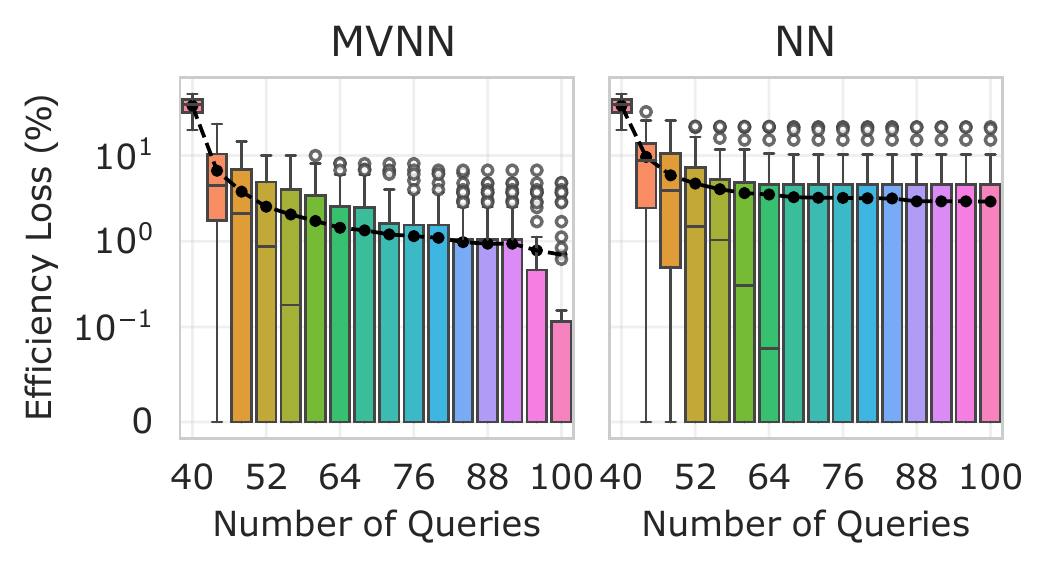}
    \vskip -0.2cm
    \caption{Efficiency loss path of MVNN vs plain NN in LSVM with the corresponding best MVNN ($T=50$) and plain NN ($T=100$) from Table~\ref{tab:efficiency_loss_mlca_appendix}. Averages are shown as black dots. We use a semi-logarithmic scale with linear range [0,1e-1].}
    \label{fig:efficiency_loss_path_LSVM}
\end{figure}
\begin{figure}[t!]
    \centering
    \includegraphics[width=1\linewidth]{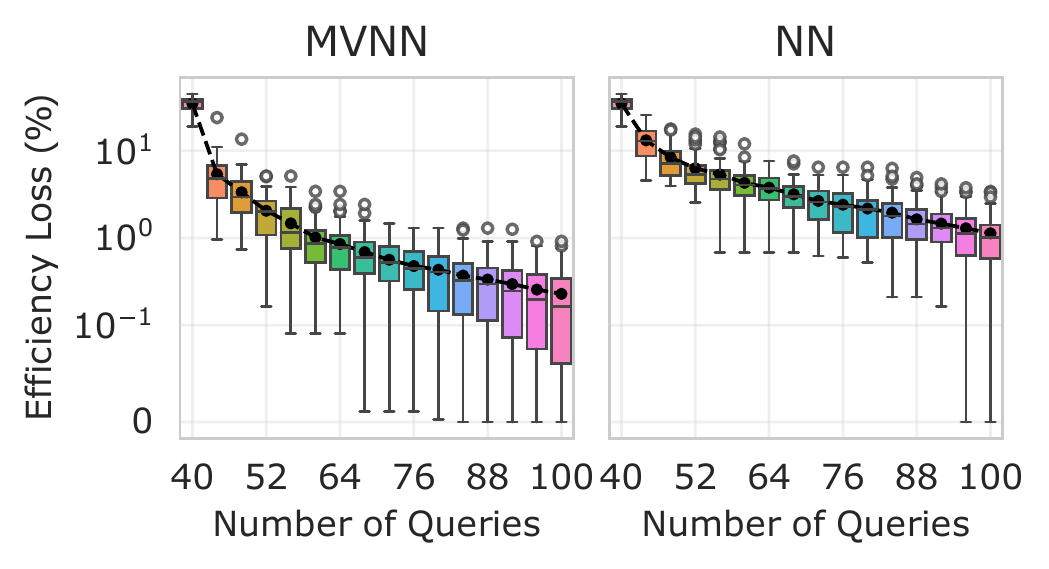}
    \vskip -0.2cm
    \caption{Efficiency loss path of MVNN vs plain NN in SRVM with the corresponding best MVNN ($T=100$) and best plain NN ($T=100$) from Table~\ref{tab:efficiency_loss_mlca_appendix}. Averages are shown as black dots. We use a semi-logarithmic scale with linear range [0,1e-1].}
    \label{fig:efficiency_loss_path_SRVM}
\end{figure}
\begin{figure}[t!]
    \centering
    \includegraphics[width=1.2\linewidth]{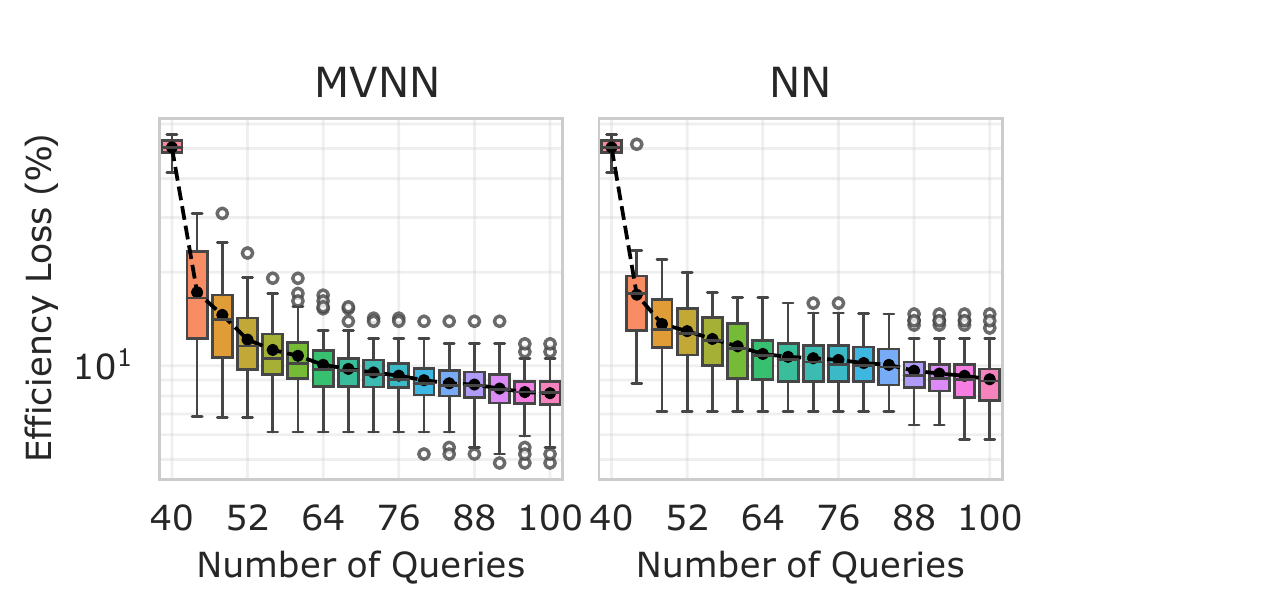}
    \vskip -0.2cm
    \caption{Efficiency loss path of MVNN vs plain NN in MRVM with the corresponding best MVNN ($T=10$) and plain NN ($T=300$) from Table~\ref{tab:efficiency_loss_mlca_appendix}. Averages are shown as black dots. Only logarithmic scale.}
    \label{fig:efficiency_loss_path_MRVM}
\end{figure}

\begin{table*}[ht!]
	\robustify\bfseries
	\centering
	\begin{sc}
	\resizebox{1\textwidth}{!}{
    \begin{tabular}{crrrrccccccccc}
    \toprule
         &  & \multicolumn{3}{c}{\textbf{Number of Queries}}  &\multicolumn{3}{c}{\textbf{Avg. Marg. SCWs (\%)}} & \multicolumn{3}{c}{\textbf{Avg. Main SCWs (\%)}}& \multicolumn{3}{c}{\textbf{Revenue in \%\,\,\textuparrow}}\\
        \cmidrule(l{2pt}r{2pt}){1-2}
        \cmidrule(l{2pt}r{2pt}){3-5}
        \cmidrule(l{2pt}r{2pt}){6-8}
        \cmidrule(l{2pt}r{2pt}){9-11}
        \cmidrule(l{2pt}r{2pt}){12-14}
    \textbf{Domain} & \textbf{T}&  \multicolumn{1}{c}{MVNN}  &   \multicolumn{1}{c}{NN}  &  \multicolumn{1}{c}{RS} & \multicolumn{1}{c}{MVNN}  &   \multicolumn{1}{c}{NN}  &  \multicolumn{1}{c}{RS} & \multicolumn{1}{c}{MVNN} & \multicolumn{1}{c}{NN}&  \multicolumn{1}{c}{RS}&\multicolumn{1}{c}{MVNN} & \multicolumn{1}{c}{NN}&  \multicolumn{1}{c}{RS}\\
        \cmidrule(l{2pt}r{2pt}){1-2}
     \cmidrule(l{2pt}r{2pt}){3-5}
        \cmidrule(l{2pt}r{2pt}){6-8}
        \cmidrule(l{2pt}r{2pt}){9-11}
        \cmidrule(l{2pt}r{2pt}){12-14}
    GSVM & 10 &  59.1& 52.6& 100.0&94.31 & 94.09& 67.05 & 85.71 & 85.71& 59.90& 60.11 $\pm$\scriptsize\, 3.86& 58.59 $\pm$\scriptsize\, 4.35& 52.19 $\pm$\scriptsize\,  2.39\\
         &  20&  54.3& 47.8& 100.0&94.16  &  93.69& 67.05 & 85.71 & 85.71& 59.90&59.07 $\pm$\scriptsize\, 3.84 & 55.71 $\pm$\scriptsize\, 4.46  & 52.19 $\pm$\scriptsize\,  2.39\\
         &  50&  47.2& 48.3& 100.0&93.25 & 93.74& 67.05 & 85.71 & 85.71& 59.90&52.77 $\pm$\scriptsize\, 5.04 & 56.09 $\pm$\scriptsize\, 4.57 & 52.19 $\pm$\scriptsize\,  2.39\\
         \cmidrule(l{2pt}r{2pt}){1-2}
       \cmidrule(l{2pt}r{2pt}){3-5}
        \cmidrule(l{2pt}r{2pt}){6-8}
        \cmidrule(l{2pt}r{2pt}){9-11}
        \cmidrule(l{2pt}r{2pt}){12-14}
    LSVM & 10&  74.2& 75.5& 100.0&93.58& 91.41& 64.57& 81.95& 80.62& 55.67&70.14 $\pm$\scriptsize\, 4.23 &   65.23 $\pm$\scriptsize\, 3.83 &53.58 $\pm$\scriptsize\,  1.84\\
         & 50&  68.6& 73.2& 100.0&94.48&91.28& 64.57& 82.74& 80.69& 55.67&70.70 $\pm$\scriptsize\, 4.63  & 64.07 $\pm$\scriptsize\, 4.24 &53.58 $\pm$\scriptsize\,  1.84 \\
         &  100&  73.2& 73.0& 100.0&94.04 & 91.63& 64.57 & 82.26&  80.86& 55.67&71.09 $\pm$\scriptsize\, 4.35 & 65.10 $\pm$\scriptsize\, 3.90 &53.58 $\pm$\scriptsize\,  1.84\\
         \cmidrule(l{2pt}r{2pt}){1-2}
    \cmidrule(l{2pt}r{2pt}){3-5}
        \cmidrule(l{2pt}r{2pt}){6-8}
        \cmidrule(l{2pt}r{2pt}){9-11}
        \cmidrule(l{2pt}r{2pt}){12-14}
        SRVM &  10&100.0  & 100.0& 100.0&92.53 & 91.34& 69.43 & 85.13&84.62& 62.12&50.86 $\pm$\scriptsize\, 2.31     &  46.93 $\pm$\scriptsize\, 2.07 &51.56 $ \pm$\scriptsize\,  2.07\\
         &  50&  99.4 & 99.9& 100.0&92.66& 91.35& 69.43 & 85.29&84.72& 62.12&50.70 $\pm$\scriptsize\, 2.44   & 45.59 $\pm$\scriptsize\, 2.72&51.56 $ \pm$\scriptsize\,  2.07\\
         &  100&  98.9 & 99.8&  100.0&92.98 & 91.47& 69.43 & 85.52& 84.74& 62.12&51.18 $\pm$\scriptsize\, 2.52   & 46.30 $\pm$\scriptsize\, 2.87&51.56 $ \pm$\scriptsize\,  2.07\\
          \cmidrule(l{2pt}r{2pt}){1-2}
        \cmidrule(l{2pt}r{2pt}){3-5}
        \cmidrule(l{2pt}r{2pt}){6-8}
        \cmidrule(l{2pt}r{2pt}){9-11}
        \cmidrule(l{2pt}r{2pt}){12-14}
     MRVM &  10&  100.0& 100.0& 100.0&86.12& 84.39& 50.39 & 82.68 &80.22& 46.03&34.45 $\pm$\scriptsize\, 2.06   & 41.68 $\pm$\scriptsize\, 1.34&43.58 $\pm$\scriptsize\,  0.65 \\
         &  100&   100.0& 100.0& 100.0&85.70 & 85.04& 50.39 & 82.22 &80.83&46.03&35.04 $\pm$\scriptsize\, 1.91  & 41.91 $\pm$\scriptsize\, 1.39&43.58 $\pm$\scriptsize\,  0.65 \\
         &  300&   100.0& 100.0& 100.0&84.29& 85.86& 50.39 & 81.03&81.87& 46.03&32.72 $\pm$\scriptsize\, 2.17 & 39.66 $\pm$\scriptsize\, 1.20 &43.58 $\pm$\scriptsize\,  0.65\\
    \bottomrule
    \end{tabular}
}
    \end{sc}
    \vskip -0.2cm
    \caption{Normalized average social welfare in the marginal economies and main economies when excluding one bidder for MVNNs, NNs and random search (RS). Additionally, we print the average number of queries, that can differ between MVNN and NN incumbents due to early termination of MLCA. Shown are averages over the same test set of 50 auction instances as in \Cref{tab:efficiency_loss_mlca_appendix}.}
    \label{tab:detailed_revenue_mlca_appendix}
\end{table*}

\subsection{Detailed Revenue Analysis -  MVNN-based Iterative CA}\label{subsec:appendix_detailed_revenue_analysis}
Recall the VCG-payments from \Cref{def:vcg_payments}, i.e., for a set of elicited reports $R$ bidder $i$'s VCG-payment is given as:
 \begin{align*}
&p_i(R) \coloneqq \hspace{-0.2cm}\underbrace{\sum_{j \in N \setminus \{i\}} \hvj{}{}\left(\left(a^*_{R_{-i}}\right)_j\right)}_{=:\textrm{Marg}_{-i}} - \hspace{-0.2cm}\underbrace{\sum_{j \in N \setminus \{i\}}\hvj{}{}\left(\left(a^*_{R}\right)_j\right)}_{=:\textrm{Main}_{-i}},
\end{align*}
 where $\textrm{Marg}_{-i}$ is for a given set of reports $R$ the optimal social welfare (SCW) in the \emph{marginal} economy when excluding bidder $i\in N$ and $\textrm{Main}_{-i}$ denotes the SCW of the \emph{main} economy when excluding bidder $i$. The larger the differences $\textrm{Marg}_{-i}-\textrm{Main}_{-i}, i\in N$ the more revenue is generated in the auction.
 In \Cref{tab:detailed_revenue_mlca_appendix}, we print the normalized average marginal SCW $(\frac{1}{n}\sum_{i\in N}\textrm{Marg}_{-i})/V(a^*)$ (\textsc{\textbf{Avg. Marg. SCWs (\%)}}) and the normalized average main SCW when excluding one bidder $(\frac{1}{n}\sum_{i\in N}\textrm{Main}_{-i})/V(a^*)$ (\textsc{\textbf{Avg. Main SCWs (\%)}})  averaged over the $50$ auction instances corresponding to \Cref{tab:efficiency_loss_mlca_appendix}.
 Using this notation, the average total relative revenue of the auction is then the difference of the normalized average marginal SCW and the normalized average main SCW when excluding one bidder times the number of bidders, i.e.,
 {\small\begin{align*}
     n\left((\frac{1}{n}\sum_{i\in N}\textrm{Marg}_{-i})/V(a^*)-(\frac{1}{n}\sum_{i\in N}\textrm{Main}_{-i})/V(a^*)\right).
 \end{align*}}

Recall, that we terminate MLCA in an intermediate iteration, if it already found an efficient allocation (to save computational costs), i.e., incurred no efficiency loss. Due to this early termination the revenue can be worse off since fewer bundles are elicited in the marginal economies.
If one runs a full auction without early termination, the revenue can only improve after the efficient allocation was already found (because $\textrm{Main}_{-i}$ cannot increase anymore after the efficient allocation was already found, but $\textrm{Marg}_{-i}$ can still improve).

In \Cref{tab:detailed_revenue_mlca_appendix}, we present a detailed revenue analysis for all domains. \Cref{tab:detailed_revenue_mlca_appendix} shows that overall the MVNN's achieved (normalized) SCWs are larger both in the marginal economies as well as in the main economies when excluding one bidder.

In GSVM, we find the efficient allocation very early (see \Cref{fig:efficiency_loss_path_GSVM} or \Cref{tab:detailed_revenue_mlca_appendix}), so the revenue is mainly determined by the number of queries. Without early determination, $60\%$ revenue should be easily achievable for MVNNs.
For LSVM and SRVM, we see that even with less queries MVNNs can outperform plain NNs consistently in terms of revenue. Without early termination, we expect this margin to be even higher (because this would add more queries for MVNNs).

In MRVM, we see that the best MVNN ($T=10$) achieves an average difference of $86.12-82.68=3.44$ while the best plain NN (T=300) achieves an average difference of $85.86-81.87=3.98$. This implies the larger revenue ($\approx 5\% = 10\cdot0.5\%$) of plain NNs in MRVM from \Cref{tab:efficiency_loss_mlca} in the main paper. A possible explanation for the larger revenue of plain NNs in MRVM is that, with each query MVNNs much more exploit the economy for that we solve the WDP, while plain NNs are more random and thus the query that solves the WDP for the main economy is not perfectly specialized for the main economy, but only slightly more helpful for the main economy than for the other economies. Recall that MLCA asks in each iteration each bidder $\Qround$-many queries: one main economy query and $\Qround-1$ marginal economy queries.
Thus, if $\Qround<n$, MLCA generates more queries for the main economy than for each marginal economy and MVNNs improve the main economy even more than they improve the marginal economies, since their queries are highly specialized and exploiting. This effect is the strongest for MRVM as it has the largest ratio of the number of bidders $n$ and $\Qround$.

However, if the objective would be to maximize revenue rather than efficiency, MVNNs could achieve as good or better revenue than plain NNs when we set $\Qround=n$, since then both main and marginal economies would be equally improved by the advantages of MVNNs (again with $\Qround=4$ the main economy profits more from the advantages of the MVNNs than the marginal economies; see \Cref{tab:detailed_revenue_mlca_appendix}).

As expected, for random search (RS) the (normalized) SCWs in both the marginal economies and the main economies when excluding one bidder are much lower compared to MVNNs and NNs. 
However, since \emph{each} economies' SCW is bad, their differences and thus the revenue does not have to be worse.
Moreover,  since RS treats all economies equally and particularly is not specialized towards the main economies (in contrast to MVNN-MLCA or NN-MLCA with $\Qround<n$), RS's revenue can be even higher compared to MVNNs and NNs (e.g., in SRVM and MRVM).

\clearpage
\newgeometry{top = 0.3cm,bottom=0.3cm}

\thispagestyle{empty} 
\atxy{\dimexpr\paperwidth-1in}{.5\paperheight}{\rotatebox[origin=center]{90}{\thepage}}

\begin{sidewaystable*}
    \renewcommand\arraystretch{1.5}
	\centering
	\begin{sc}
    \resizebox{1\textwidth}{!}{%
    \tabcolsep=4pt
    \begin{tabular}{cccllllllllllllllllllllllllllllll}
    \toprule
     &     & {} & \multicolumn{3}{c}{\LARGE$\boldsymbol{r_{xy}}$} & \multicolumn{3}{c}{\LARGE\textbf{Kt}} & \multicolumn{3}{c}{\LARGE\textbf{MAE}} & \multicolumn{3}{c}{\LARGE\textbf{Architecture}} & \multicolumn{3}{c}{\LARGE\textbf{Batch Size}} & \multicolumn{3}{c}{\LARGE\textbf{L2-Reg.}} & \multicolumn{3}{c}{\LARGE\textbf{Optimizer}} & \multicolumn{3}{c}{\LARGE\textbf{Learning Rate}} & \multicolumn{3}{c}{\LARGE\textbf{Loss}} & \multicolumn{3}{c}{\LARGE\textbf{Epochs}} \\
     \cmidrule(lr){4-6}
     \cmidrule(lr){7-9}
     \cmidrule(lr){10-12}
     \cmidrule(lr){13-15}
     \cmidrule(lr){16-18}
     \cmidrule(lr){19-21}
     \cmidrule(lr){22-24}
     \cmidrule(lr){25-27}
     \cmidrule(lr){28-30}
     \cmidrule(lr){31-33}
     &  &  &             \multicolumn{1}{c}{NN} &         \multicolumn{1}{c}{MVNN} &       \multicolumn{1}{c}{MVNN} &              \multicolumn{1}{c}{NN} &         \multicolumn{1}{c}{MVNN} &       \multicolumn{1}{c}{MVNN} &                      \multicolumn{1}{c}{NN} &                 \multicolumn{1}{c}{MVNN} &               \multicolumn{1}{c}{MVNN} &            \multicolumn{1}{c}{NN} &       \multicolumn{1}{c}{MVNN} &     \multicolumn{1}{c}{MVNN} &         \multicolumn{1}{c}{NN} & \multicolumn{1}{c}{MVNN} & \multicolumn{1}{c}{MVNN} &         \multicolumn{1}{c}{NN} &    \multicolumn{1}{c}{MVNN} &  \multicolumn{1}{c}{MVNN}&        \multicolumn{1}{c}{NN} & \multicolumn{1}{c}{MVNN} & \multicolumn{1}{c}{MVNN} &            \multicolumn{1}{c}{NN} &    \multicolumn{1}{c}{MVNN} &  \multicolumn{1}{c}{MVNN} &            \multicolumn{1}{c}{NN} & \multicolumn{1}{c}{MVNN} & \multicolumn{1}{c}{MVNN} &     \multicolumn{1}{c}{NN} & \multicolumn{1}{c}{MVNN} & \multicolumn{1}{c}{MVNN} \\
      \multicolumn{1}{c}{\LARGE\textbf{Domain}} & \multicolumn{1}{c}{\LARGE\textbf{T}} & \multicolumn{1}{c}{\LARGE\textbf{Bidder}} &             \multicolumn{1}{c}{Plain} &         \multicolumn{1}{c}{Abs} &       \multicolumn{1}{c}{ReLU-Proj..} &              \multicolumn{1}{c}{Plain} &         \multicolumn{1}{c}{Abs} &       \multicolumn{1}{c}{ReLU-Proj.} &                      \multicolumn{1}{c}{Plain} &                 \multicolumn{1}{c}{Abs} &               \multicolumn{1}{c}{ReLU-Proj.} &            \multicolumn{1}{c}{Plain} &       \multicolumn{1}{c}{Abs} &     \multicolumn{1}{c}{ReLU-Proj.} &         \multicolumn{1}{c}{Plain} & \multicolumn{1}{c}{Abs} & \multicolumn{1}{c}{ReLU-Proj.} &         \multicolumn{1}{c}{Plain} &    \multicolumn{1}{c}{Abs} &  \multicolumn{1}{c}{ReLU-Proj.}&        \multicolumn{1}{c}{Plain} & \multicolumn{1}{c}{Abs} & \multicolumn{1}{c}{ReLU} &            \multicolumn{1}{c}{Plain} &    \multicolumn{1}{c}{Abs} &  \multicolumn{1}{c}{ReLU-Proj.} &            \multicolumn{1}{c}{Plain} & \multicolumn{1}{c}{Abs} & \multicolumn{1}{c}{ReLU-Proj.} &     \multicolumn{1}{c}{Plain} & \multicolumn{1}{c}{Abs} & \multicolumn{1}{c}{ReLU-Proj.} \\
     \cmidrule(lr){1-3}
     \cmidrule(lr){4-6}
     \cmidrule(lr){7-9}
     \cmidrule(lr){10-12}
     \cmidrule(lr){13-15}
     \cmidrule(lr){16-18}
     \cmidrule(lr){19-21}
     \cmidrule(lr){22-24}
     \cmidrule(lr){25-27}
     \cmidrule(lr){28-30}
     \cmidrule(lr){31-33}
GSVM & 10  & Nat &  0.782 $\pm$\scriptsize 0.020 &  0.861 $\pm$\scriptsize 0.017 &  0.856 $\pm$\scriptsize 0.022 &  0.583 $\pm$\scriptsize 0.021 &  0.672 $\pm$\scriptsize 0.022 &  0.668 $\pm$\scriptsize 0.027 &  2.197$\mathrm{e}{+1}$ $\pm$\scriptsize 1.713$\mathrm{e}{+0}$ &  1.731$\mathrm{e}{+1}$ $\pm$\scriptsize 1.605$\mathrm{e}{+0}$ &  1.776$\mathrm{e}{+1}$ $\pm$\scriptsize 1.960$\mathrm{e}{+0}$ &      [16, 16] &          [62] &          [63] &          1 &       4 &         1 &  4.844$\mathrm{e}{-10}$ &  1.160$\mathrm{e}{-7}$ &  1.002$\mathrm{e}{-10}$ &      Adam &    Adam &      Adam &      3.558$\mathrm{e}{-3}$ &  9.945$\mathrm{e}{-3}$ &  4.297$\mathrm{e}{-3}$ &           MAE &     MAE &       MAE &    219 &      55 &        60 \\
     &     & Reg &  0.746 $\pm$\scriptsize 0.033 &  0.799 $\pm$\scriptsize 0.032 &  0.816 $\pm$\scriptsize 0.034 &  0.557 $\pm$\scriptsize 0.033 &  0.614 $\pm$\scriptsize 0.034 &  0.633 $\pm$\scriptsize 0.038 &  1.376$\mathrm{e}{+1}$ $\pm$\scriptsize 9.002$\mathrm{e}{-1}$ &  1.268$\mathrm{e}{+1}$ $\pm$\scriptsize 1.077$\mathrm{e}{+0}$ &  1.185$\mathrm{e}{+1}$ $\pm$\scriptsize 1.095$\mathrm{e}{+0}$ &     [5, 5, 5] &          [11] &          [61] &          2 &       2 &         1 &  3.817$\mathrm{e}{-7}$ &  1.124$\mathrm{e}{-7}$ &  5.884$\mathrm{e}{-8}$ &      Adam &    Adam &      Adam &      9.973$\mathrm{e}{-3}$ &  6.599$\mathrm{e}{-3}$ &  6.587$\mathrm{e}{-3}$ &           MAE &     MAE &       MAE &    321 &     181 &       114 \\
     & 20  & Nat &  0.911 $\pm$\scriptsize 0.017 &  0.966 $\pm$\scriptsize 0.008 &  0.965 $\pm$\scriptsize 0.007 &  0.752 $\pm$\scriptsize 0.029 &  0.865 $\pm$\scriptsize 0.020 &  0.849 $\pm$\scriptsize 0.017 &  1.347$\mathrm{e}{+1}$ $\pm$\scriptsize 1.616$\mathrm{e}{+0}$ &  8.479$\mathrm{e}{+0}$ $\pm$\scriptsize 1.354$\mathrm{e}{+0}$ &  8.568$\mathrm{e}{+0}$ $\pm$\scriptsize 1.089$\mathrm{e}{+0}$ &          [44] &  [21, 21, 21] &          [53] &          3 &       1 &         3 &  3.864$\mathrm{e}{-9}$ &  8.866$\mathrm{e}{-8}$ &  1.344$\mathrm{e}{-7}$ &       SGD &    Adam &      Adam &      2.935$\mathrm{e}{-3}$ &  3.584$\mathrm{e}{-3}$ &  5.982$\mathrm{e}{-3}$ &           MAE &     MSE &       MAE &    398 &     338 &       117 \\
     &     & Reg &  0.944 $\pm$\scriptsize 0.011 &  0.965 $\pm$\scriptsize 0.009 &  0.973 $\pm$\scriptsize 0.009 &  0.815 $\pm$\scriptsize 0.021 &  0.878 $\pm$\scriptsize 0.018 &  0.882 $\pm$\scriptsize 0.020 &  6.594$\mathrm{e}{+0}$ $\pm$\scriptsize 6.773$\mathrm{e}{-1}$ &  4.844$\mathrm{e}{+0}$ $\pm$\scriptsize 7.184$\mathrm{e}{-1}$ &  4.324$\mathrm{e}{+0}$ $\pm$\scriptsize 6.252$\mathrm{e}{-1}$ &          [57] &  [21, 21, 21] &          [61] &          4 &       1 &         2 &  2.300$\mathrm{e}{-10}$ &  1.137$\mathrm{e}{-8}$ &  4.712$\mathrm{e}{-7}$ &       SGD &    Adam &      Adam &      2.945$\mathrm{e}{-3}$ &  1.415$\mathrm{e}{-3}$ &  6.956$\mathrm{e}{-3}$ &           MAE &     MAE &       MAE &    390 &     381 &       138 \\
     & 50  & Nat &  0.995 $\pm$\scriptsize 0.000 &  0.998 $\pm$\scriptsize 0.001 &  0.997 $\pm$\scriptsize 0.000 &  0.953 $\pm$\scriptsize 0.003 &  0.978 $\pm$\scriptsize 0.003 &  0.962 $\pm$\scriptsize 0.003 &  3.294$\mathrm{e}{+0}$ $\pm$\scriptsize 2.693$\mathrm{e}{-1}$ &  1.540$\mathrm{e}{+0}$ $\pm$\scriptsize 1.774$\mathrm{e}{-1}$ &  2.506$\mathrm{e}{+0}$ $\pm$\scriptsize 2.261$\mathrm{e}{-1}$ &          [61] &  [20, 20, 20] &          [61] &          1 &       1 &         2 &  4.465$\mathrm{e}{-8}$ &  2.117$\mathrm{e}{-7}$ &  1.653$\mathrm{e}{-7}$ &       SGD &    Adam &      Adam &      8.275$\mathrm{e}{-3}$ &  7.401$\mathrm{e}{-4}$ &  2.246$\mathrm{e}{-3}$ &           MAE &     MAE &       MAE &    223 &     247 &       223 \\
     &     & Reg &  0.995 $\pm$\scriptsize 0.001 &  0.998 $\pm$\scriptsize 0.001 &  0.999 $\pm$\scriptsize 0.000 &  0.953 $\pm$\scriptsize 0.003 &  0.978 $\pm$\scriptsize 0.002 &  0.974 $\pm$\scriptsize 0.002 &  1.936$\mathrm{e}{+0}$ $\pm$\scriptsize 1.204$\mathrm{e}{-1}$ &  7.172$\mathrm{e}{-1}$ $\pm$\scriptsize 9.339$\mathrm{e}{-2}$ &  8.704$\mathrm{e}{-1}$ $\pm$\scriptsize 8.185$\mathrm{e}{-2}$ &          [48] &  [21, 21, 21] &          [62] &          4 &       3 &         2 &  3.013$\mathrm{e}{-7}$ &  5.333$\mathrm{e}{-7}$ &  3.351$\mathrm{e}{-7}$ &       SGD &    Adam &      Adam &      6.250$\mathrm{e}{-3}$ &  1.676$\mathrm{e}{-3}$ &  2.002$\mathrm{e}{-3}$ &           MAE &     MAE &       MAE &    348 &     392 &       355 \\
     \cmidrule(lr){1-3}
     \cmidrule(lr){4-6}
     \cmidrule(lr){7-9}
     \cmidrule(lr){10-12}
     \cmidrule(lr){13-15}
     \cmidrule(lr){16-18}
     \cmidrule(lr){19-21}
     \cmidrule(lr){22-24}
     \cmidrule(lr){25-27}
     \cmidrule(lr){28-30}
     \cmidrule(lr){31-33}
LSVM & 10  & Nat &  0.701 $\pm$\scriptsize 0.013 &  0.803 $\pm$\scriptsize 0.010 &  0.598 $\pm$\scriptsize 0.019 &  0.710 $\pm$\scriptsize 0.023 &  0.712 $\pm$\scriptsize 0.013 &  0.693 $\pm$\scriptsize 0.011 &  3.225$\mathrm{e}{+1}$ $\pm$\scriptsize 1.366$\mathrm{e}{+0}$ &  2.811$\mathrm{e}{+1}$ $\pm$\scriptsize 1.447$\mathrm{e}{+0}$ &  3.263$\mathrm{e}{+1}$ $\pm$\scriptsize 1.256$\mathrm{e}{+0}$ &          [24] &          [47] &      [29, 29] &          4 &       2 &         1 &  7.358$\mathrm{e}{-8}$ &  1.224$\mathrm{e}{-8}$ &  5.951$\mathrm{e}{-8}$ &      Adam &     SGD &       SGD &      3.919$\mathrm{e}{-3}$ &  9.634$\mathrm{e}{-3}$ &  2.398$\mathrm{e}{-4}$ &           MAE &     MSE &       MAE &    213 &     372 &       321 \\
     &     & Reg &  0.679 $\pm$\scriptsize 0.027 &  0.766 $\pm$\scriptsize 0.026 &  0.785 $\pm$\scriptsize 0.028 &  0.504 $\pm$\scriptsize 0.025 &  0.585 $\pm$\scriptsize 0.028 &  0.605 $\pm$\scriptsize 0.031 &  2.917$\mathrm{e}{+1}$ $\pm$\scriptsize 4.131$\mathrm{e}{+0}$ &  2.457$\mathrm{e}{+1}$ $\pm$\scriptsize 3.299$\mathrm{e}{+0}$ &  2.385$\mathrm{e}{+1}$ $\pm$\scriptsize 3.382$\mathrm{e}{+0}$ &          [59] &          [54] &          [44] &          1 &       4 &         1 &  6.406$\mathrm{e}{-7}$ &  2.702$\mathrm{e}{-7}$ &  1.287$\mathrm{e}{-8}$ &      Adam &    Adam &       SGD &      9.853$\mathrm{e}{-3}$ &  9.752$\mathrm{e}{-3}$ &  1.323$\mathrm{e}{-3}$ &           MAE &     MAE &       MAE &    113 &     121 &       273 \\
     & 50  & Nat &  0.681 $\pm$\scriptsize 0.024 &  0.793 $\pm$\scriptsize 0.011 &  0.796 $\pm$\scriptsize 0.012 &  0.678 $\pm$\scriptsize 0.035 &  0.765 $\pm$\scriptsize 0.011 &  0.753 $\pm$\scriptsize 0.009 &  3.052$\mathrm{e}{+1}$ $\pm$\scriptsize 1.116$\mathrm{e}{+0}$ &  2.327$\mathrm{e}{+1}$ $\pm$\scriptsize 7.769$\mathrm{e}{-1}$ &  2.292$\mathrm{e}{+1}$ $\pm$\scriptsize 6.508$\mathrm{e}{-1}$ &           [4] &  [11, 11, 11] &  [20, 20, 20] &          1 &       1 &         1 &  5.092$\mathrm{e}{-7}$ &  6.468$\mathrm{e}{-8}$ &  2.702$\mathrm{e}{-10}$ &      Adam &    Adam &       SGD &      8.593$\mathrm{e}{-3}$ &  1.084$\mathrm{e}{-3}$ &  3.355$\mathrm{e}{-3}$ &           MAE &     MAE &       MAE &     92 &     306 &       220 \\
     &     & Reg &  0.938 $\pm$\scriptsize 0.006 &  0.964 $\pm$\scriptsize 0.006 &  0.963 $\pm$\scriptsize 0.007 &  0.812 $\pm$\scriptsize 0.013 &  0.859 $\pm$\scriptsize 0.015 &  0.860 $\pm$\scriptsize 0.017 &  1.274$\mathrm{e}{+1}$ $\pm$\scriptsize 2.013$\mathrm{e}{+0}$ &  9.379$\mathrm{e}{+0}$ $\pm$\scriptsize 2.020$\mathrm{e}{+0}$ &  9.511$\mathrm{e}{+0}$ $\pm$\scriptsize 2.114$\mathrm{e}{+0}$ &          [26] &      [24, 24] &          [62] &          1 &       2 &         1 &  1.997$\mathrm{e}{-10}$ &  4.041$\mathrm{e}{-7}$ &  6.958$\mathrm{e}{-7}$ &       SGD &    Adam &      Adam &      1.067$\mathrm{e}{-3}$ &  1.994$\mathrm{e}{-3}$ &  1.917$\mathrm{e}{-3}$ &           MAE &     MAE &       MAE &    158 &     281 &       258 \\
     & 100 & Nat &  0.719 $\pm$\scriptsize 0.015 &  0.825 $\pm$\scriptsize 0.008 &  0.834 $\pm$\scriptsize 0.008 &  0.706 $\pm$\scriptsize 0.018 &  0.817 $\pm$\scriptsize 0.005 &  0.813 $\pm$\scriptsize 0.005 &  2.805$\mathrm{e}{+1}$ $\pm$\scriptsize 7.422$\mathrm{e}{-1}$ &  1.994$\mathrm{e}{+1}$ $\pm$\scriptsize 6.458$\mathrm{e}{-1}$ &  1.998$\mathrm{e}{+1}$ $\pm$\scriptsize 6.605$\mathrm{e}{-1}$ &     [9, 9, 9] &      [32, 32] &      [24, 24] &          4 &       3 &         1 &  9.546$\mathrm{e}{-10}$ &  7.747$\mathrm{e}{-7}$ &  8.972$\mathrm{e}{-10}$ &      Adam &    Adam &      Adam &      5.013$\mathrm{e}{-3}$ &  1.780$\mathrm{e}{-3}$ &  2.406$\mathrm{e}{-4}$ &           MAE &     MAE &       MAE &     79 &     250 &       391 \\
     &     & Reg &  0.969 $\pm$\scriptsize 0.005 &  0.983 $\pm$\scriptsize 0.004 &  0.983 $\pm$\scriptsize 0.005 &  0.857 $\pm$\scriptsize 0.012 &  0.912 $\pm$\scriptsize 0.014 &  0.918 $\pm$\scriptsize 0.015 &  9.082$\mathrm{e}{+0}$ $\pm$\scriptsize 1.764$\mathrm{e}{+0}$ &  6.003$\mathrm{e}{+0}$ $\pm$\scriptsize 1.649$\mathrm{e}{+0}$ &  5.952$\mathrm{e}{+0}$ $\pm$\scriptsize 1.719$\mathrm{e}{+0}$ &          [14] &      [22, 22] &          [55] &          4 &       2 &         4 &  6.641$\mathrm{e}{-10}$ &  7.392$\mathrm{e}{-9}$ &  2.273$\mathrm{e}{-9}$ &       SGD &    Adam &      Adam &      2.109$\mathrm{e}{-3}$ &  1.450$\mathrm{e}{-3}$ &  3.640$\mathrm{e}{-3}$ &           MAE &     MAE &       MAE &    392 &     387 &       204 \\
     \cmidrule(lr){1-3}
     \cmidrule(lr){4-6}
     \cmidrule(lr){7-9}
     \cmidrule(lr){10-12}
     \cmidrule(lr){13-15}
     \cmidrule(lr){16-18}
     \cmidrule(lr){19-21}
     \cmidrule(lr){22-24}
     \cmidrule(lr){25-27}
     \cmidrule(lr){28-30}
     \cmidrule(lr){31-33}
MRVM & 10  & Lo &  0.296 $\pm$\scriptsize 0.021 &  0.438 $\pm$\scriptsize 0.018 &  0.383 $\pm$\scriptsize 0.023 &  0.200 $\pm$\scriptsize 0.015 &  0.303 $\pm$\scriptsize 0.012 &  0.262 $\pm$\scriptsize 0.017 &  7.447$\mathrm{e}{+6}$ $\pm$\scriptsize 1.383$\mathrm{e}{+6}$ &  7.093$\mathrm{e}{+6}$ $\pm$\scriptsize 1.330$\mathrm{e}{+6}$ &  7.511$\mathrm{e}{+6}$ $\pm$\scriptsize 1.445$\mathrm{e}{+6}$ &           [4] &        [9, 9] &          [21] &          2 &       1 &         4 &  3.510$\mathrm{e}{-10}$ &  1.161$\mathrm{e}{-7}$ &  4.206$\mathrm{e}{-8}$ &       SGD &    Adam &      Adam &      1.222$\mathrm{e}{-4}$ &  5.775$\mathrm{e}{-4}$ &  2.028$\mathrm{e}{-3}$ &           MSE &     MSE &       MSE &    377 &     276 &       271 \\
     &     & Nat &  0.597 $\pm$\scriptsize 0.009 &  0.588 $\pm$\scriptsize 0.016 &  0.529 $\pm$\scriptsize 0.022 &  0.414 $\pm$\scriptsize 0.008 &  0.409 $\pm$\scriptsize 0.014 &  0.365 $\pm$\scriptsize 0.018 &  3.708$\mathrm{e}{+8}$ $\pm$\scriptsize 1.415$\mathrm{e}{+8}$ &  2.754$\mathrm{e}{+8}$ $\pm$\scriptsize 2.316$\mathrm{e}{+7}$ &  2.863$\mathrm{e}{+8}$ $\pm$\scriptsize 2.287$\mathrm{e}{+7}$ &        [2, 2] &          [29] &          [32] &          4 &       3 &         1 &  4.675$\mathrm{e}{-8}$ &  1.400$\mathrm{e}{-7}$ &  3.966$\mathrm{e}{-8}$ &      Adam &    Adam &      Adam &      1.681$\mathrm{e}{-3}$ &  1.320$\mathrm{e}{-3}$ &  2.529$\mathrm{e}{-4}$ &           MSE &     MAE &       MAE &    330 &     121 &       352 \\
     &     & Reg &  0.368 $\pm$\scriptsize 0.052 &  0.482 $\pm$\scriptsize 0.028 &  0.459 $\pm$\scriptsize 0.031 &  0.255 $\pm$\scriptsize 0.038 &  0.340 $\pm$\scriptsize 0.020 &  0.322 $\pm$\scriptsize 0.022 &  2.809$\mathrm{e}{+7}$ $\pm$\scriptsize 6.838$\mathrm{e}{+6}$ &  2.455$\mathrm{e}{+7}$ $\pm$\scriptsize 6.173$\mathrm{e}{+6}$ &  2.599$\mathrm{e}{+7}$ $\pm$\scriptsize 6.790$\mathrm{e}{+6}$ &           [6] &          [13] &          [28] &          2 &       1 &         2 &  6.042$\mathrm{e}{-9}$ &  4.163$\mathrm{e}{-8}$ &  1.333$\mathrm{e}{-7}$ &       SGD &    Adam &      Adam &      1.360$\mathrm{e}{-4}$ &  7.934$\mathrm{e}{-4}$ &  2.344$\mathrm{e}{-3}$ &           MSE &     MSE &       MAE &     86 &     285 &       184 \\
     & 100 & Lo &  0.729 $\pm$\scriptsize 0.012 &  0.928 $\pm$\scriptsize 0.010 &  0.918 $\pm$\scriptsize 0.012 &  0.545 $\pm$\scriptsize 0.012 &  0.793 $\pm$\scriptsize 0.017 &  0.786 $\pm$\scriptsize 0.019 &  5.177$\mathrm{e}{+6}$ $\pm$\scriptsize 9.344$\mathrm{e}{+5}$ &  2.582$\mathrm{e}{+6}$ $\pm$\scriptsize 5.289$\mathrm{e}{+5}$ &  2.735$\mathrm{e}{+6}$ $\pm$\scriptsize 5.360$\mathrm{e}{+5}$ &        [1, 1] &          [32] &      [13, 13] &          1 &       1 &         1 &  1.515$\mathrm{e}{-8}$ &  6.048$\mathrm{e}{-8}$ &  3.342$\mathrm{e}{-7}$ &      Adam &    Adam &      Adam &      1.405$\mathrm{e}{-3}$ &  9.018$\mathrm{e}{-4}$ &  4.855$\mathrm{e}{-3}$ &           MAE &     MAE &       MSE &    225 &     157 &        61 \\
     &     & Nat &  0.776 $\pm$\scriptsize 0.010 &  0.894 $\pm$\scriptsize 0.012 &  0.887 $\pm$\scriptsize 0.011 &  0.581 $\pm$\scriptsize 0.010 &  0.737 $\pm$\scriptsize 0.015 &  0.726 $\pm$\scriptsize 0.014 &  2.094$\mathrm{e}{+8}$ $\pm$\scriptsize 1.582$\mathrm{e}{+7}$ &  1.338$\mathrm{e}{+8}$ $\pm$\scriptsize 1.393$\mathrm{e}{+7}$ &  1.390$\mathrm{e}{+8}$ $\pm$\scriptsize 1.339$\mathrm{e}{+7}$ &        [6, 6] &      [13, 13] &          [23] &          2 &       2 &         4 &  6.091$\mathrm{e}{-7}$ &  1.455$\mathrm{e}{-8}$ &  2.330$\mathrm{e}{-8}$ &      Adam &    Adam &      Adam &      2.654$\mathrm{e}{-4}$ &  2.368$\mathrm{e}{-4}$ &  5.983$\mathrm{e}{-4}$ &           MAE &     MAE &       MAE &    393 &     399 &       245 \\
     &     & Reg &  0.747 $\pm$\scriptsize 0.016 &  0.915 $\pm$\scriptsize 0.013 &  0.917 $\pm$\scriptsize 0.015 &  0.572 $\pm$\scriptsize 0.018 &  0.777 $\pm$\scriptsize 0.023 &  0.779 $\pm$\scriptsize 0.027 &  1.880$\mathrm{e}{+7}$ $\pm$\scriptsize 3.774$\mathrm{e}{+6}$ &  8.807$\mathrm{e}{+6}$ $\pm$\scriptsize 1.401$\mathrm{e}{+6}$ &  8.504$\mathrm{e}{+6}$ $\pm$\scriptsize 1.121$\mathrm{e}{+6}$ &        [1, 1] &      [15, 15] &      [13, 13] &          1 &       1 &         1 &  5.212$\mathrm{e}{-7}$ &  5.562$\mathrm{e}{-8}$ &  1.336$\mathrm{e}{-8}$ &      Adam &    Adam &      Adam &      2.176$\mathrm{e}{-3}$ &  3.952$\mathrm{e}{-4}$ &  9.247$\mathrm{e}{-4}$ &           MAE &     MAE &       MAE &    117 &     130 &       189 \\
     & 300 & Lo &  0.934 $\pm$\scriptsize 0.005 &  0.968 $\pm$\scriptsize 0.004 &  0.972 $\pm$\scriptsize 0.003 &  0.819 $\pm$\scriptsize 0.009 &  0.874 $\pm$\scriptsize 0.011 &  0.883 $\pm$\scriptsize 0.010 &  2.417$\mathrm{e}{+6}$ $\pm$\scriptsize 4.241$\mathrm{e}{+5}$ &  1.616$\mathrm{e}{+6}$ $\pm$\scriptsize 3.194$\mathrm{e}{+5}$ &  1.484$\mathrm{e}{+6}$ $\pm$\scriptsize 2.932$\mathrm{e}{+5}$ &      [11, 11] &          [30] &      [16, 16] &          4 &       4 &         2 &  7.798$\mathrm{e}{-9}$ &  2.857$\mathrm{e}{-8}$ &  3.990$\mathrm{e}{-8}$ &       SGD &    Adam &      Adam &      2.904$\mathrm{e}{-3}$ &  1.571$\mathrm{e}{-3}$ &  1.195$\mathrm{e}{-3}$ &           MAE &     MSE &       MAE &     93 &     228 &       161 \\
     &     & Nat &  0.927 $\pm$\scriptsize 0.014 &  0.933 $\pm$\scriptsize 0.008 &  0.933 $\pm$\scriptsize 0.008 &  0.808 $\pm$\scriptsize 0.013 &  0.818 $\pm$\scriptsize 0.008 &  0.814 $\pm$\scriptsize 0.009 &  9.967$\mathrm{e}{+7}$ $\pm$\scriptsize 1.006$\mathrm{e}{+7}$ &  9.519$\mathrm{e}{+7}$ $\pm$\scriptsize 8.113$\mathrm{e}{+6}$ &  9.584$\mathrm{e}{+7}$ $\pm$\scriptsize 8.548$\mathrm{e}{+6}$ &        [1, 1] &        [4, 4] &           [9] &          3 &       2 &         4 &  5.880$\mathrm{e}{-9}$ &  3.681$\mathrm{e}{-9}$ &  6.226$\mathrm{e}{-9}$ &       SGD &    Adam &      Adam &      3.670$\mathrm{e}{-4}$ &  1.346$\mathrm{e}{-3}$ &  6.650$\mathrm{e}{-4}$ &           MAE &     MAE &       MAE &    292 &     149 &       191 \\
     &     & Reg &  0.923 $\pm$\scriptsize 0.005 &  0.963 $\pm$\scriptsize 0.007 &  0.958 $\pm$\scriptsize 0.009 &  0.809 $\pm$\scriptsize 0.012 &  0.862 $\pm$\scriptsize 0.018 &  0.851 $\pm$\scriptsize 0.019 &  9.405$\mathrm{e}{+6}$ $\pm$\scriptsize 1.883$\mathrm{e}{+6}$ &  5.257$\mathrm{e}{+6}$ $\pm$\scriptsize 6.323$\mathrm{e}{+5}$ &  5.501$\mathrm{e}{+6}$ $\pm$\scriptsize 5.956$\mathrm{e}{+5}$ &          [16] &          [32] &        [9, 9] &          1 &       1 &         3 &  1.246$\mathrm{e}{-10}$ &  7.778$\mathrm{e}{-7}$ &  8.111$\mathrm{e}{-8}$ &       SGD &    Adam &      Adam &      4.827$\mathrm{e}{-3}$ &  5.867$\mathrm{e}{-4}$ &  1.064$\mathrm{e}{-3}$ &           MAE &     MAE &       MAE &    103 &     292 &       298 \\
     \cmidrule(lr){1-3}
     \cmidrule(lr){4-6}
     \cmidrule(lr){7-9}
     \cmidrule(lr){10-12}
     \cmidrule(lr){13-15}
     \cmidrule(lr){16-18}
     \cmidrule(lr){19-21}
     \cmidrule(lr){22-24}
     \cmidrule(lr){25-27}
     \cmidrule(lr){28-30}
     \cmidrule(lr){31-33}
SRVM & 10  & H.F. &  0.762 $\pm$\scriptsize 0.011 &  0.805 $\pm$\scriptsize 0.010 &  0.795 $\pm$\scriptsize 0.018 &  0.607 $\pm$\scriptsize 0.012 &  0.642 $\pm$\scriptsize 0.014 &  0.626 $\pm$\scriptsize 0.020 &  4.268$\mathrm{e}{+1}$ $\pm$\scriptsize 2.940$\mathrm{e}{+0}$ &  1.359$\mathrm{e}{+1}$ $\pm$\scriptsize 9.009$\mathrm{e}{-1}$ &  1.436$\mathrm{e}{+1}$ $\pm$\scriptsize 9.693$\mathrm{e}{-1}$ &     [2, 2, 2] &          [29] &      [30, 30] &          4 &       2 &         2 &  5.455$\mathrm{e}{-9}$ &  1.884$\mathrm{e}{-7}$ &  3.639$\mathrm{e}{-7}$ &      Adam &    Adam &      Adam &      3.995$\mathrm{e}{-3}$ &  9.079$\mathrm{e}{-3}$ &  9.808$\mathrm{e}{-3}$ &           MSE &     MAE &       MSE &    394 &     206 &       253 \\
     &     & Lo &  0.570 $\pm$\scriptsize 0.030 &  0.678 $\pm$\scriptsize 0.027 &  0.660 $\pm$\scriptsize 0.029 &  0.489 $\pm$\scriptsize 0.032 &  0.571 $\pm$\scriptsize 0.027 &  0.559 $\pm$\scriptsize 0.030 &  3.298$\mathrm{e}{+0}$ $\pm$\scriptsize 3.677$\mathrm{e}{-1}$ &  2.708$\mathrm{e}{+0}$ $\pm$\scriptsize 2.180$\mathrm{e}{-1}$ &  2.756$\mathrm{e}{+0}$ $\pm$\scriptsize 2.047$\mathrm{e}{-1}$ &          [18] &  [18, 18, 18] &      [27, 27] &          1 &       2 &         3 &  2.506$\mathrm{e}{-9}$ &  1.564$\mathrm{e}{-7}$ &  9.703$\mathrm{e}{-7}$ &       SGD &    Adam &      Adam &      1.369$\mathrm{e}{-4}$ &  3.551$\mathrm{e}{-3}$ &  5.858$\mathrm{e}{-3}$ &           MSE &     MAE &       MAE &     89 &     203 &       290 \\
     &     & Nat &  0.636 $\pm$\scriptsize 0.018 &  0.692 $\pm$\scriptsize 0.028 &  0.676 $\pm$\scriptsize 0.032 &  0.535 $\pm$\scriptsize 0.012 &  0.579 $\pm$\scriptsize 0.023 &  0.560 $\pm$\scriptsize 0.026 &  4.082$\mathrm{e}{+2}$ $\pm$\scriptsize 2.859$\mathrm{e}{+1}$ &  3.582$\mathrm{e}{+2}$ $\pm$\scriptsize 2.826$\mathrm{e}{+1}$ &  3.691$\mathrm{e}{+2}$ $\pm$\scriptsize 2.935$\mathrm{e}{+1}$ &     [2, 2, 2] &      [25, 25] &          [55] &          3 &       2 &         1 &  5.466$\mathrm{e}{-9}$ &  1.471$\mathrm{e}{-7}$ &  1.295$\mathrm{e}{-10}$ &      Adam &    Adam &      Adam &      4.268$\mathrm{e}{-3}$ &  7.093$\mathrm{e}{-3}$ &  5.177$\mathrm{e}{-3}$ &           MAE &     MAE &       MAE &    400 &      59 &       135 \\
     &     & Reg &  0.675 $\pm$\scriptsize 0.020 &  0.704 $\pm$\scriptsize 0.029 &  0.693 $\pm$\scriptsize 0.027 &  0.544 $\pm$\scriptsize 0.014 &  0.572 $\pm$\scriptsize 0.024 &  0.562 $\pm$\scriptsize 0.023 &  2.234$\mathrm{e}{+2}$ $\pm$\scriptsize 1.753$\mathrm{e}{+1}$ &  1.994$\mathrm{e}{+2}$ $\pm$\scriptsize 1.590$\mathrm{e}{+1}$ &  2.042$\mathrm{e}{+2}$ $\pm$\scriptsize 1.624$\mathrm{e}{+1}$ &  [10, 10, 10] &          [35] &          [44] &          2 &       2 &         1 &  3.237$\mathrm{e}{-9}$ &  2.667$\mathrm{e}{-8}$ &  6.439$\mathrm{e}{-9}$ &      Adam &    Adam &      Adam &      9.590$\mathrm{e}{-3}$ &  3.521$\mathrm{e}{-3}$ &  3.023$\mathrm{e}{-3}$ &           MSE &     MAE &       MAE &    320 &     219 &       396 \\
     & 50  & H.F. &  0.895 $\pm$\scriptsize 0.012 &  0.945 $\pm$\scriptsize 0.006 &  0.933 $\pm$\scriptsize 0.007 &  0.803 $\pm$\scriptsize 0.020 &  0.870 $\pm$\scriptsize 0.011 &  0.853 $\pm$\scriptsize 0.013 &  9.761$\mathrm{e}{+0}$ $\pm$\scriptsize 1.216$\mathrm{e}{+0}$ &  5.789$\mathrm{e}{+0}$ $\pm$\scriptsize 5.648$\mathrm{e}{-1}$ &  6.451$\mathrm{e}{+0}$ $\pm$\scriptsize 6.152$\mathrm{e}{-1}$ &        [2, 2] &      [21, 21] &  [21, 21, 21] &          1 &       4 &         2 &  3.989$\mathrm{e}{-8}$ &  8.670$\mathrm{e}{-9}$ &  8.080$\mathrm{e}{-9}$ &       SGD &    Adam &      Adam &      7.067$\mathrm{e}{-4}$ &  1.847$\mathrm{e}{-3}$ &  2.820$\mathrm{e}{-3}$ &           MAE &     MAE &       MAE &    256 &     211 &       204 \\
     &     & Lo &  0.794 $\pm$\scriptsize 0.019 &  0.992 $\pm$\scriptsize 0.002 &  0.955 $\pm$\scriptsize 0.008 &  0.771 $\pm$\scriptsize 0.030 &  0.903 $\pm$\scriptsize 0.000 &  0.902 $\pm$\scriptsize 0.000 &  2.128$\mathrm{e}{+0}$ $\pm$\scriptsize 2.477$\mathrm{e}{-1}$ &  2.976$\mathrm{e}{-1}$ $\pm$\scriptsize 4.422$\mathrm{e}{-2}$ &  5.608$\mathrm{e}{-1}$ $\pm$\scriptsize 8.987$\mathrm{e}{-2}$ &           [2] &  [20, 20, 20] &  [20, 20, 20] &          3 &       4 &         4 &  8.099$\mathrm{e}{-7}$ &  7.348$\mathrm{e}{-8}$ &  5.486$\mathrm{e}{-8}$ &       SGD &    Adam &      Adam &      1.196$\mathrm{e}{-4}$ &  2.398$\mathrm{e}{-3}$ &  3.818$\mathrm{e}{-3}$ &           MAE &     MAE &       MAE &    224 &     245 &       333 \\
     &     & Nat &  0.916 $\pm$\scriptsize 0.008 &  0.996 $\pm$\scriptsize 0.001 &  0.995 $\pm$\scriptsize 0.002 &  0.801 $\pm$\scriptsize 0.009 &  0.925 $\pm$\scriptsize 0.007 &  0.918 $\pm$\scriptsize 0.005 &  2.014$\mathrm{e}{+2}$ $\pm$\scriptsize 1.438$\mathrm{e}{+1}$ &  3.280$\mathrm{e}{+1}$ $\pm$\scriptsize 5.011$\mathrm{e}{+0}$ &  3.528$\mathrm{e}{+1}$ $\pm$\scriptsize 4.707$\mathrm{e}{+0}$ &           [2] &  [19, 19, 19] &  [19, 19, 19] &          1 &       2 &         2 &  9.253$\mathrm{e}{-7}$ &  6.578$\mathrm{e}{-8}$ &  4.919$\mathrm{e}{-7}$ &       SGD &    Adam &      Adam &      2.687$\mathrm{e}{-4}$ &  2.340$\mathrm{e}{-3}$ &  4.937$\mathrm{e}{-3}$ &           MAE &     MAE &       MAE &    112 &     220 &       225 \\
     &     & Reg &  0.945 $\pm$\scriptsize 0.017 &  0.995 $\pm$\scriptsize 0.001 &  0.995 $\pm$\scriptsize 0.001 &  0.823 $\pm$\scriptsize 0.022 &  0.939 $\pm$\scriptsize 0.006 &  0.931 $\pm$\scriptsize 0.004 &  9.217$\mathrm{e}{+1}$ $\pm$\scriptsize 1.343$\mathrm{e}{+1}$ &  2.119$\mathrm{e}{+1}$ $\pm$\scriptsize 2.486$\mathrm{e}{+0}$ &  2.206$\mathrm{e}{+1}$ $\pm$\scriptsize 2.423$\mathrm{e}{+0}$ &     [2, 2, 2] &      [30, 30] &  [19, 19, 19] &          1 &       1 &         3 &  5.244$\mathrm{e}{-10}$ &  2.352$\mathrm{e}{-9}$ &  6.164$\mathrm{e}{-8}$ &      Adam &    Adam &      Adam &      6.539$\mathrm{e}{-3}$ &  1.920$\mathrm{e}{-3}$ &  4.493$\mathrm{e}{-3}$ &           MAE &     MAE &       MAE &    230 &     208 &       297 \\
     & 100 & H.F. &  0.927 $\pm$\scriptsize 0.005 &  0.964 $\pm$\scriptsize 0.003 &  0.956 $\pm$\scriptsize 0.004 &  0.896 $\pm$\scriptsize 0.006 &  0.917 $\pm$\scriptsize 0.005 &  0.908 $\pm$\scriptsize 0.006 &  6.204$\mathrm{e}{+0}$ $\pm$\scriptsize 3.830$\mathrm{e}{-1}$ &  3.867$\mathrm{e}{+0}$ $\pm$\scriptsize 2.865$\mathrm{e}{-1}$ &  4.337$\mathrm{e}{+0}$ $\pm$\scriptsize 3.420$\mathrm{e}{-1}$ &           [1] &      [32, 32] &  [20, 20, 20] &          4 &       4 &         1 &  8.548$\mathrm{e}{-8}$ &  2.419$\mathrm{e}{-8}$ &  3.730$\mathrm{e}{-9}$ &       SGD &    Adam &      Adam &      5.262$\mathrm{e}{-3}$ &  1.822$\mathrm{e}{-3}$ &  7.410$\mathrm{e}{-4}$ &           MAE &     MAE &       MAE &    290 &      86 &       186 \\
     &     & Lo &  0.872 $\pm$\scriptsize 0.003 &  0.998 $\pm$\scriptsize 0.001 &  0.978 $\pm$\scriptsize 0.006 &  0.900 $\pm$\scriptsize 0.002 &  0.903 $\pm$\scriptsize 0.000 &  0.903 $\pm$\scriptsize 0.000 &  1.313$\mathrm{e}{+0}$ $\pm$\scriptsize 1.255$\mathrm{e}{-1}$ &  1.173$\mathrm{e}{-1}$ $\pm$\scriptsize 1.749$\mathrm{e}{-2}$ &  2.839$\mathrm{e}{-1}$ $\pm$\scriptsize 4.946$\mathrm{e}{-2}$ &           [1] &  [21, 21, 21] &  [20, 20, 20] &          4 &       4 &         3 &  1.034$\mathrm{e}{-9}$ &  6.892$\mathrm{e}{-8}$ &  2.486$\mathrm{e}{-8}$ &       SGD &    Adam &      Adam &      5.180$\mathrm{e}{-4}$ &  1.882$\mathrm{e}{-3}$ &  1.684$\mathrm{e}{-3}$ &           MAE &     MAE &       MAE &    376 &     125 &       372 \\
     &     & Nat &  0.959 $\pm$\scriptsize 0.004 &  0.999 $\pm$\scriptsize 0.000 &  0.999 $\pm$\scriptsize 0.000 &  0.841 $\pm$\scriptsize 0.008 &  0.958 $\pm$\scriptsize 0.002 &  0.952 $\pm$\scriptsize 0.003 &  1.293$\mathrm{e}{+2}$ $\pm$\scriptsize 1.099$\mathrm{e}{+1}$ &  1.342$\mathrm{e}{+1}$ $\pm$\scriptsize 1.168$\mathrm{e}{+0}$ &  1.502$\mathrm{e}{+1}$ $\pm$\scriptsize 1.181$\mathrm{e}{+0}$ &          [34] &  [14, 14, 14] &  [20, 20, 20] &          1 &       2 &         4 &  6.302$\mathrm{e}{-8}$ &  1.520$\mathrm{e}{-7}$ &  3.698$\mathrm{e}{-7}$ &      Adam &    Adam &      Adam &      5.768$\mathrm{e}{-3}$ &  2.262$\mathrm{e}{-3}$ &  4.292$\mathrm{e}{-3}$ &           MAE &     MAE &       MAE &    221 &     289 &       244 \\
     &     & Reg &  0.975 $\pm$\scriptsize 0.002 &  0.998 $\pm$\scriptsize 0.000 &  0.981 $\pm$\scriptsize 0.033 &  0.895 $\pm$\scriptsize 0.012 &  0.962 $\pm$\scriptsize 0.003 &  0.948 $\pm$\scriptsize 0.021 &  5.510$\mathrm{e}{+1}$ $\pm$\scriptsize 4.958$\mathrm{e}{+0}$ &  1.206$\mathrm{e}{+1}$ $\pm$\scriptsize 1.335$\mathrm{e}{+0}$ &  1.910$\mathrm{e}{+1}$ $\pm$\scriptsize 1.299$\mathrm{e}{+1}$ &           [5] &      [23, 23] &  [21, 21, 21] &          1 &       3 &         2 &  1.617$\mathrm{e}{-8}$ &  2.641$\mathrm{e}{-9}$ &  3.085$\mathrm{e}{-8}$ &       SGD &    Adam &      Adam &      2.324$\mathrm{e}{-4}$ &  1.790$\mathrm{e}{-3}$ &  3.954$\mathrm{e}{-3}$ &           MAE &     MAE &       MAE &    330 &     381 &        76 \\
\bottomrule
\end{tabular}
}
\end{sc}
\vskip -0.2cm
\caption{Prediction performance of the MVNN implementations \textsc{MVNN-Abs} and \textsc{MVNN-ReLU-Projected} measured via Pearson correlation coefficient ($\boldsymbol{r_{xy}}$), Kendall tau (\textbf{\textsc{Kt}}) and mean absolute error \textbf{(MAE)} with a $95\%$-CI in four SATS domains with corresponding bidder types: high frequency (\textsc{H.F.}), local (\textsc{Lo}), regional (\textsc{Reg}) and national (\textsc{Nat}). Both MVNNs and plain NNs are trained on $T$, validated on $0.2\cdot262,144 $ and evaluated on $262,144-T-0.2\cdot262,144$ random bundles.}
\label{tab:full_pred_performance_table_appendix}
\end{sidewaystable*}
\vspace{-10cm}
\restoregeometry
\clearpage

\newgeometry{top = 2cm,bottom=3cm,left=0.2cm, right=0.2cm}
\begin{table*}[t!]
\renewcommand\arraystretch{0.2}
\centering
\tabcolsep=0pt
	\begin{sc}
    \resizebox{1\textwidth}{!}{%
\begin{tabular}{@{}cccc@{}}
\toprule
Regional Bidder & National Bidder & Local Bidder & High Frequency Bidder \\\midrule
\multicolumn{4}{c}{GSVM} \\
\includegraphics[width=0.2\textwidth]{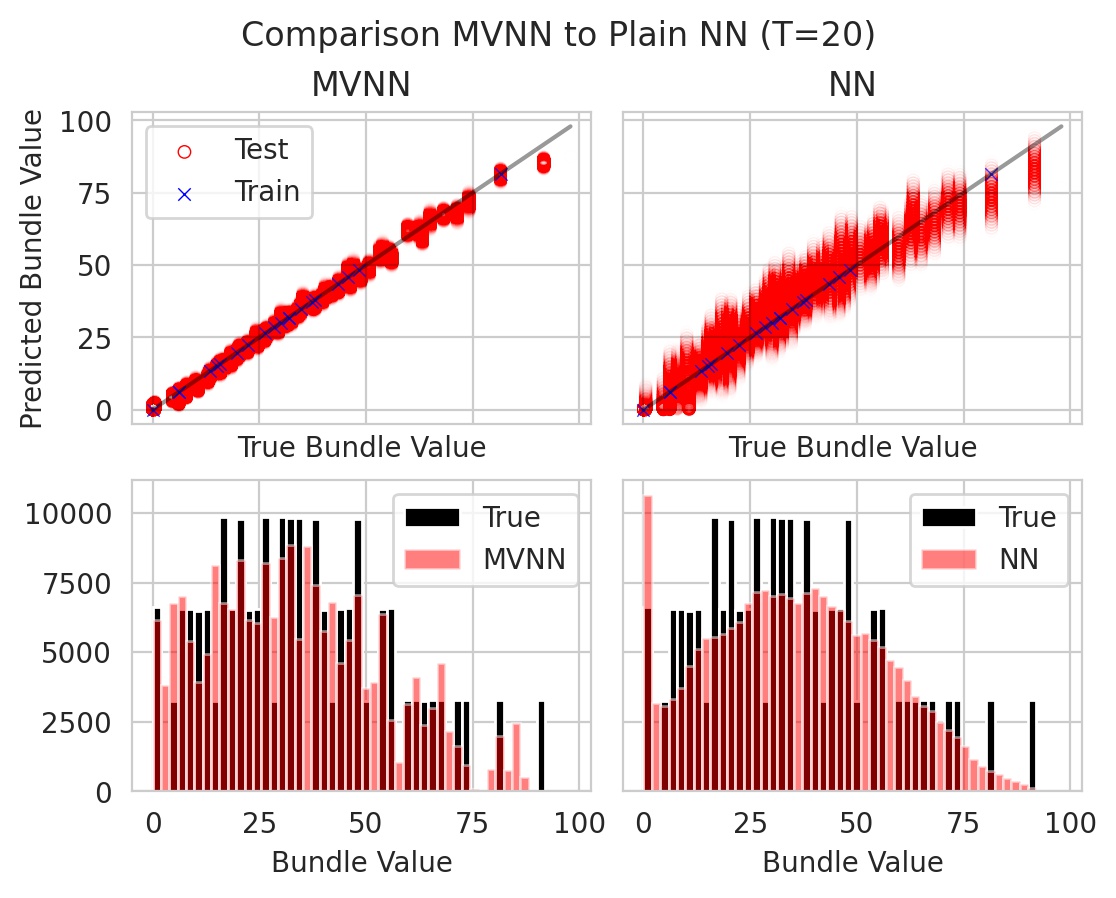} & \includegraphics[width=0.2\textwidth]{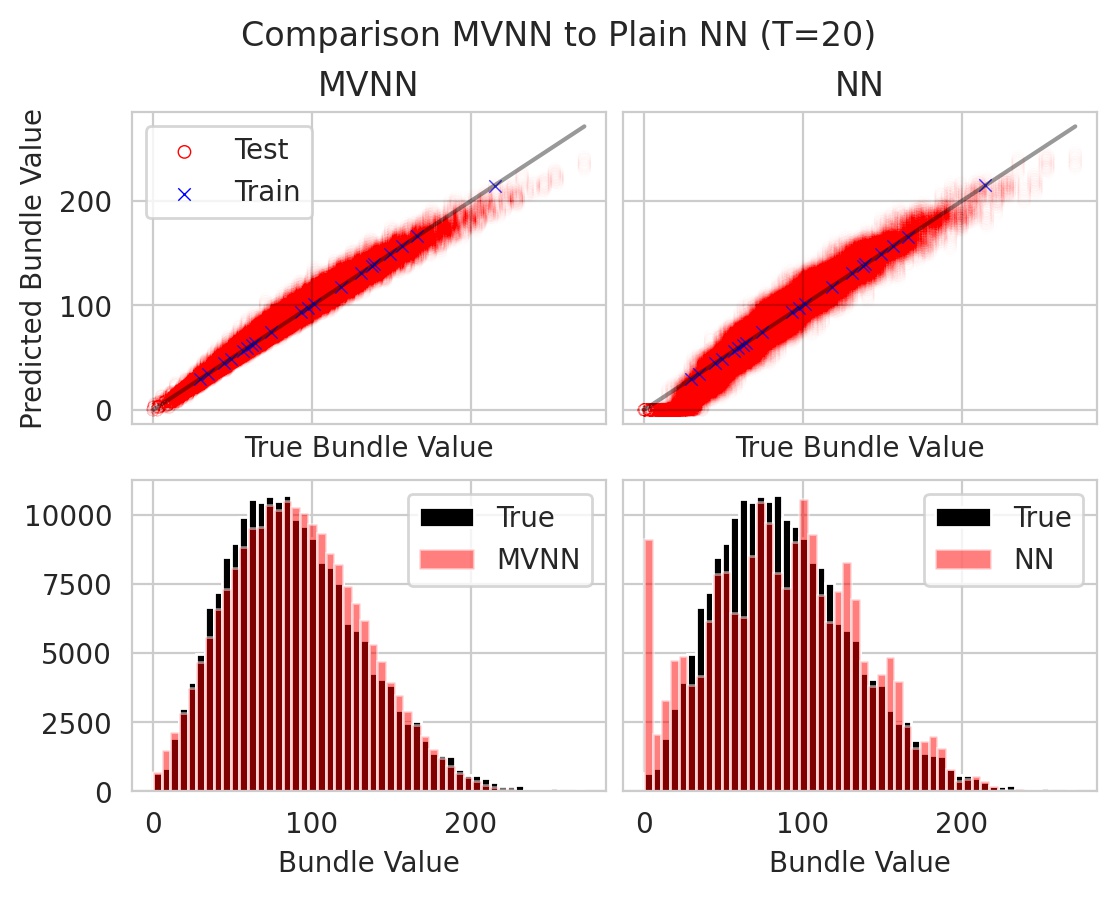} &       &                \\ \midrule
\multicolumn{4}{c}{LSVM} \\
\includegraphics[width=0.2\textwidth]{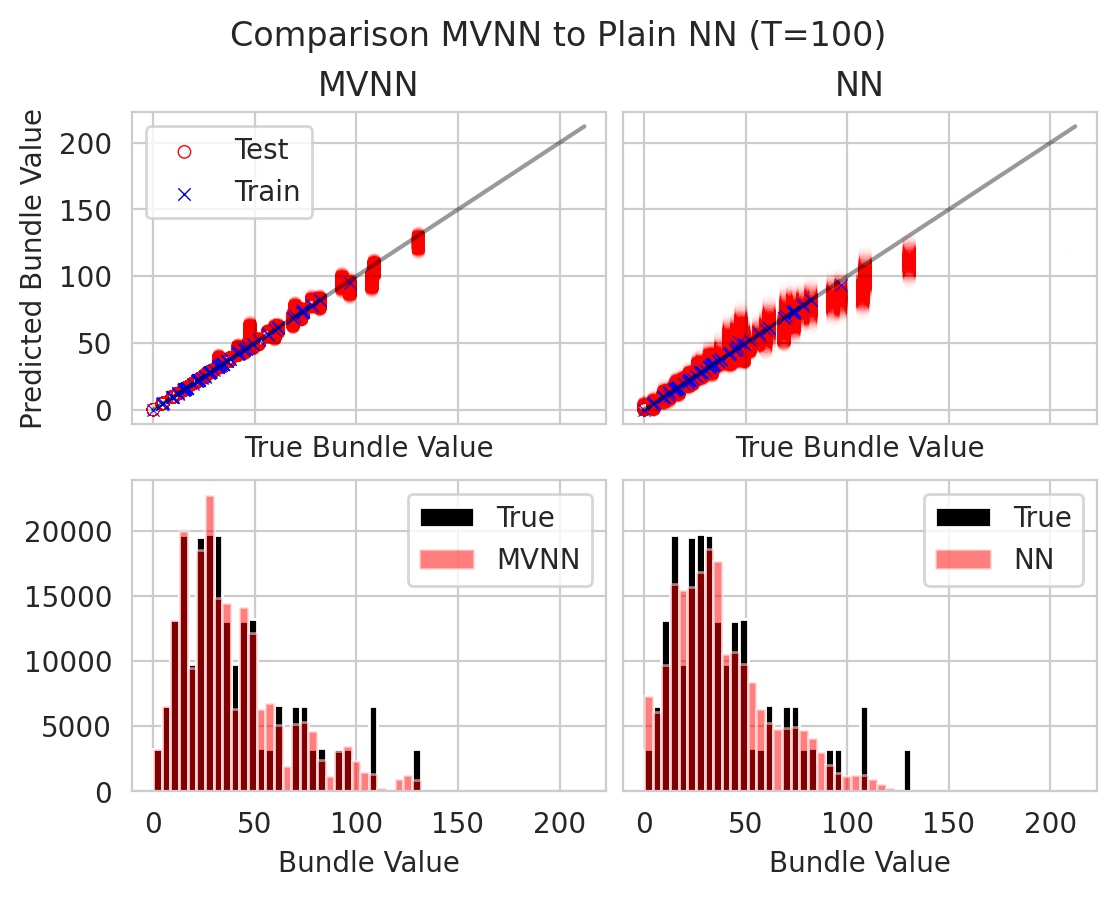} & \includegraphics[width=0.2\textwidth]{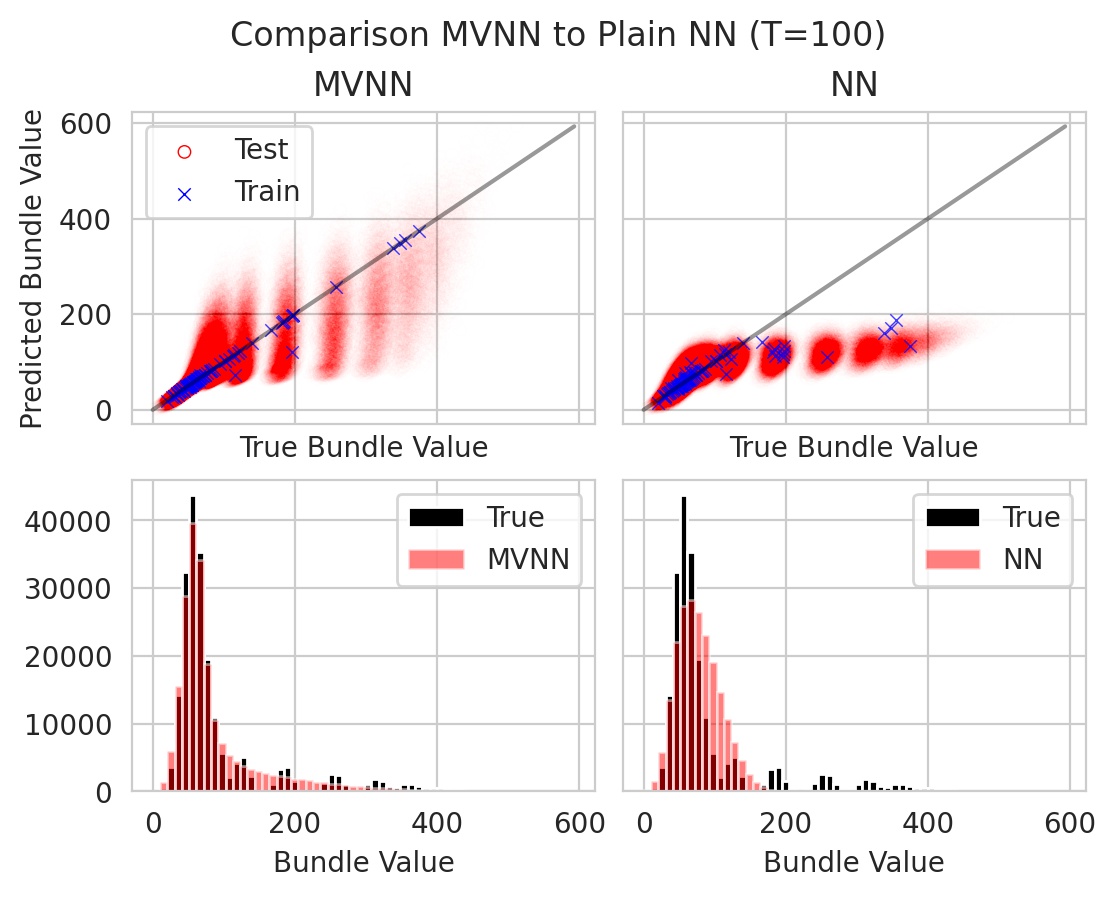}  &       &                \\\midrule
\multicolumn{4}{c}{SRVM} \\
   \includegraphics[width=0.2\textwidth]{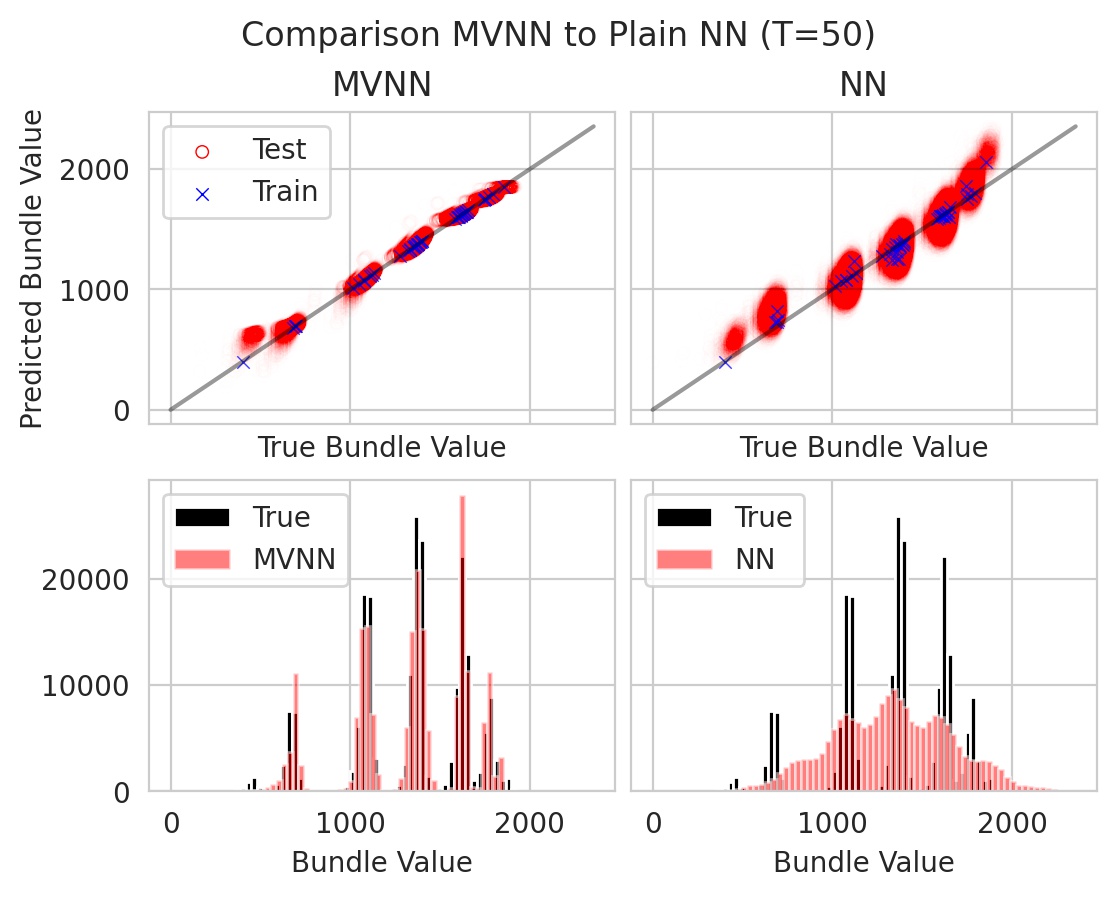} & \includegraphics[width=0.2\textwidth]{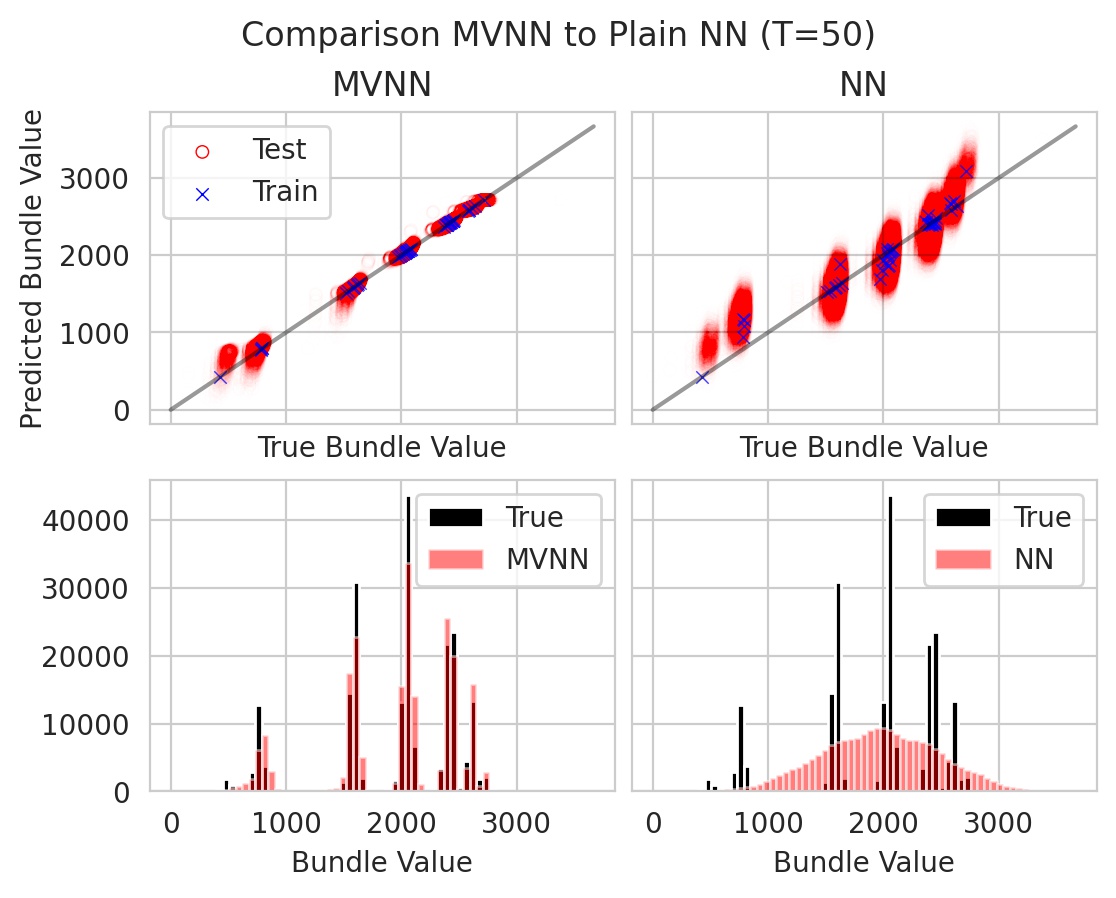}  & \includegraphics[width=0.2\textwidth]{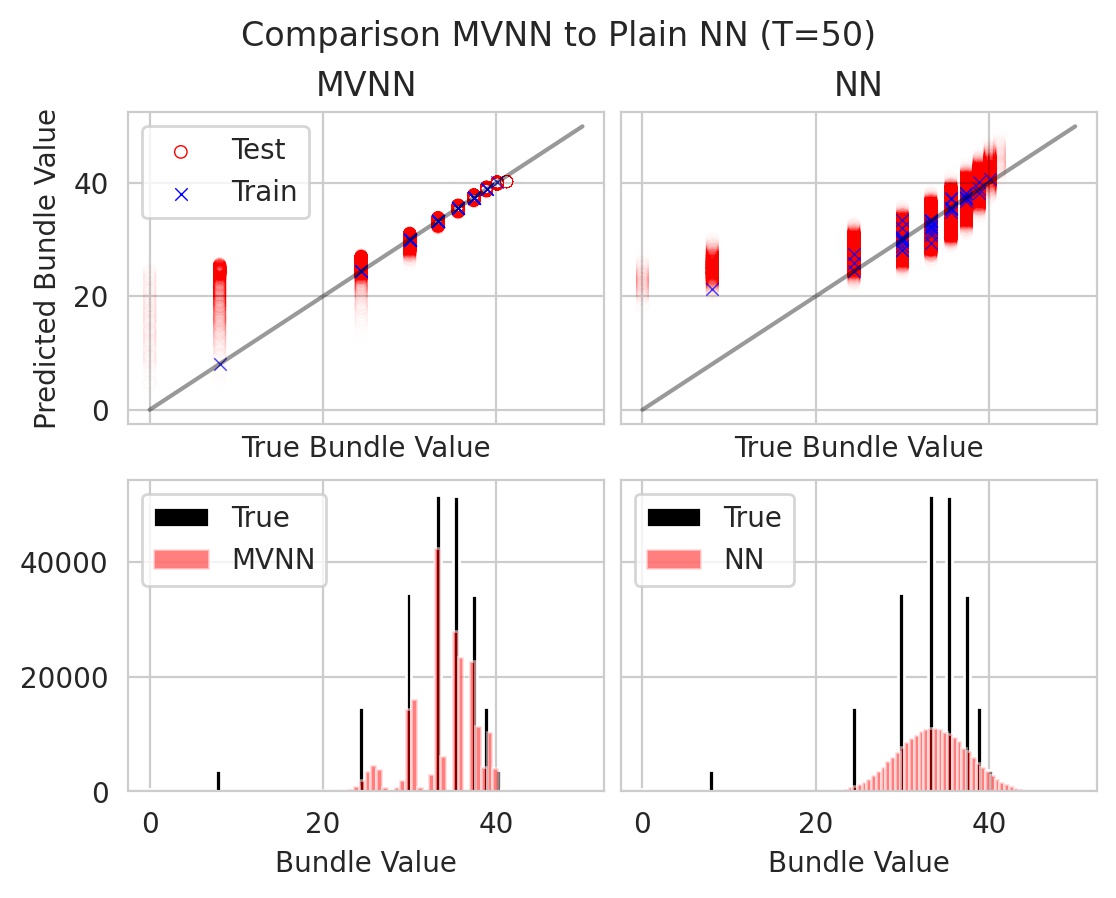}  & \includegraphics[width=0.2\textwidth]{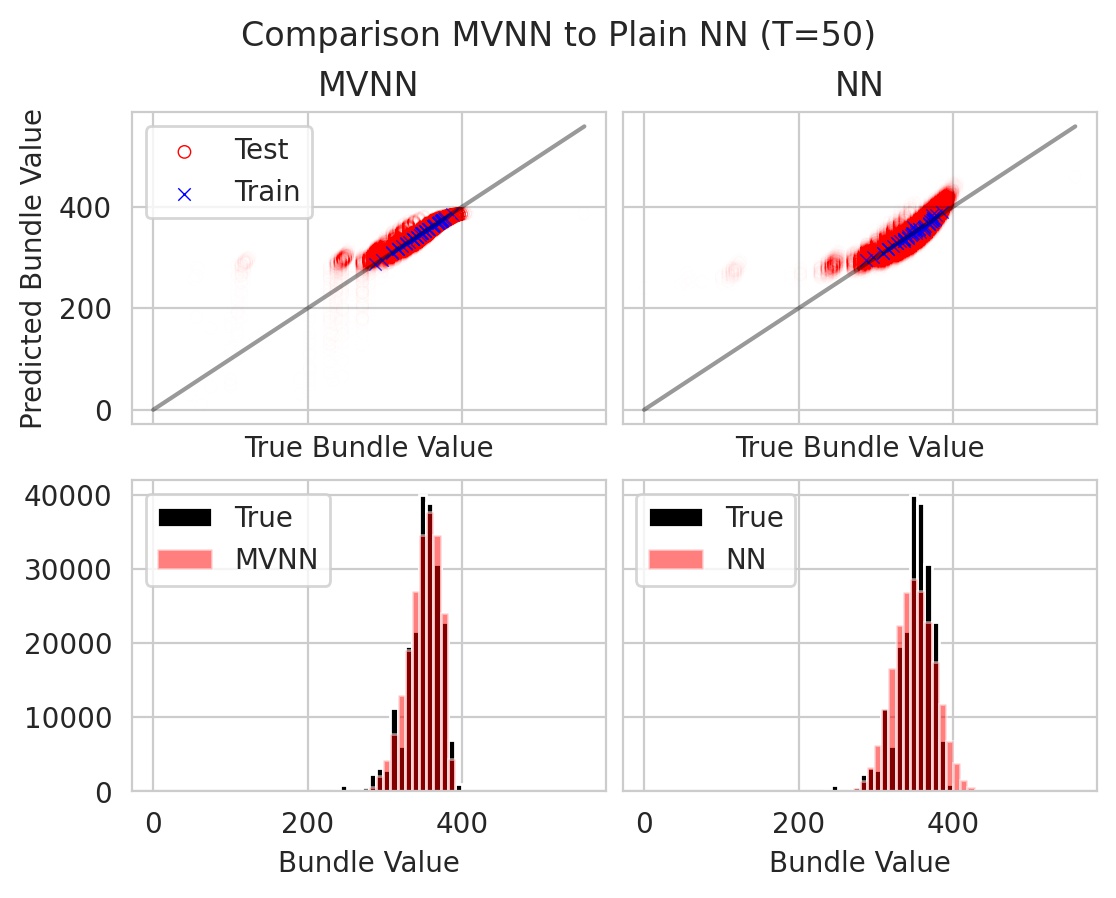}
   \\\midrule
\multicolumn{4}{c}{MRVM} \\
   \includegraphics[width=0.2\textwidth]{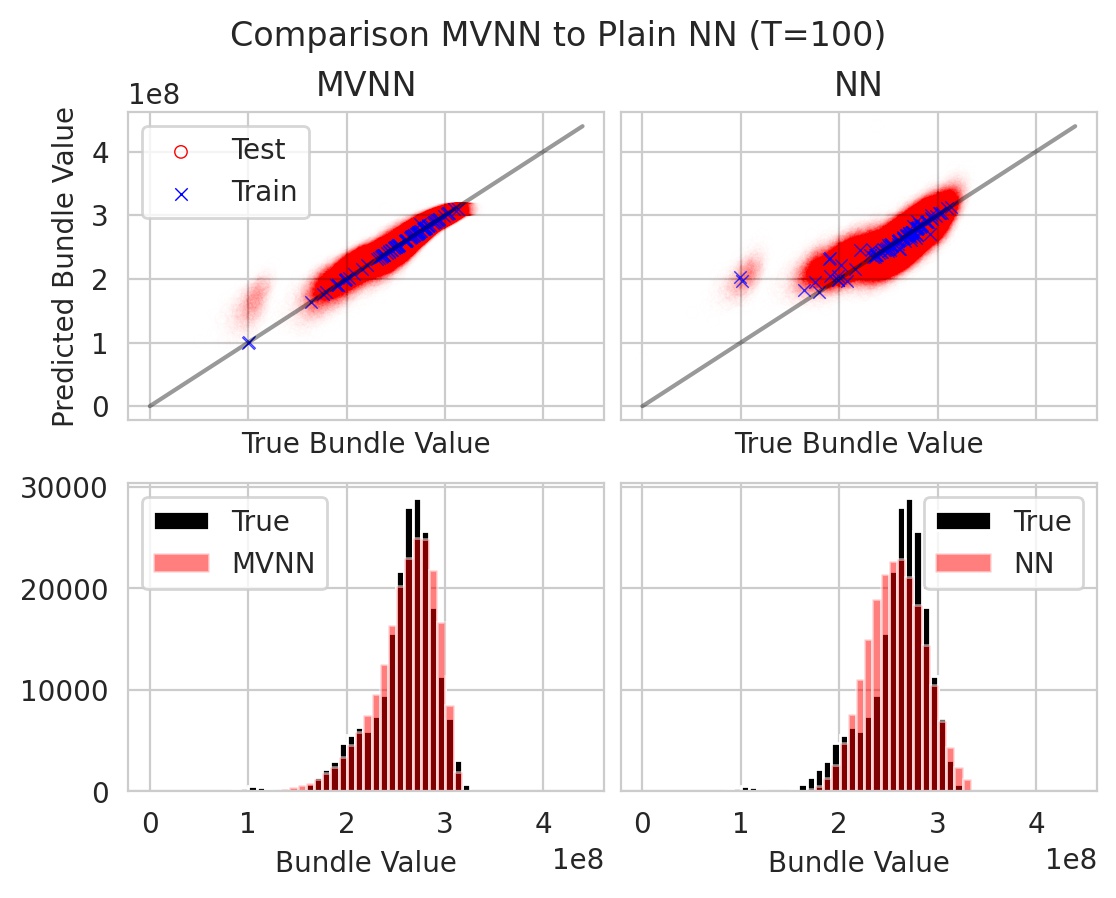} & \includegraphics[width=0.2\textwidth]{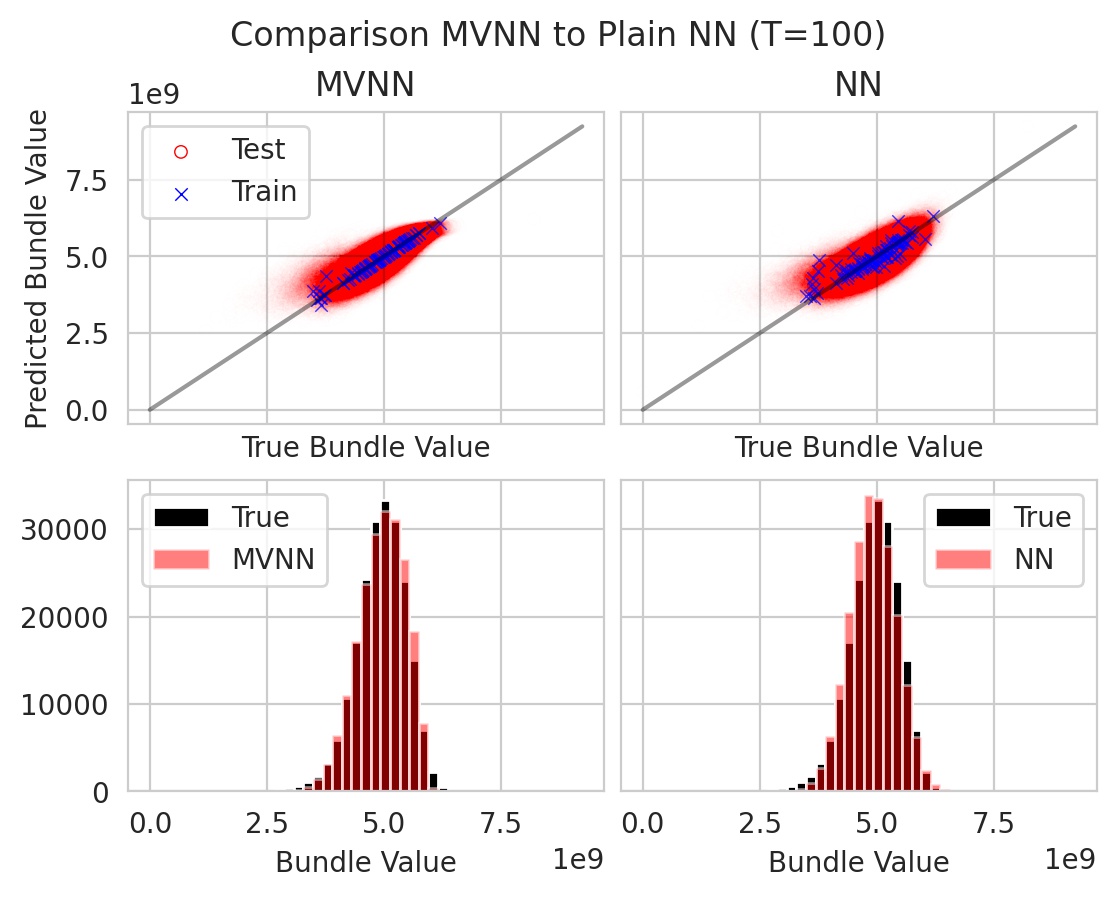}  & \includegraphics[width=0.2\textwidth]{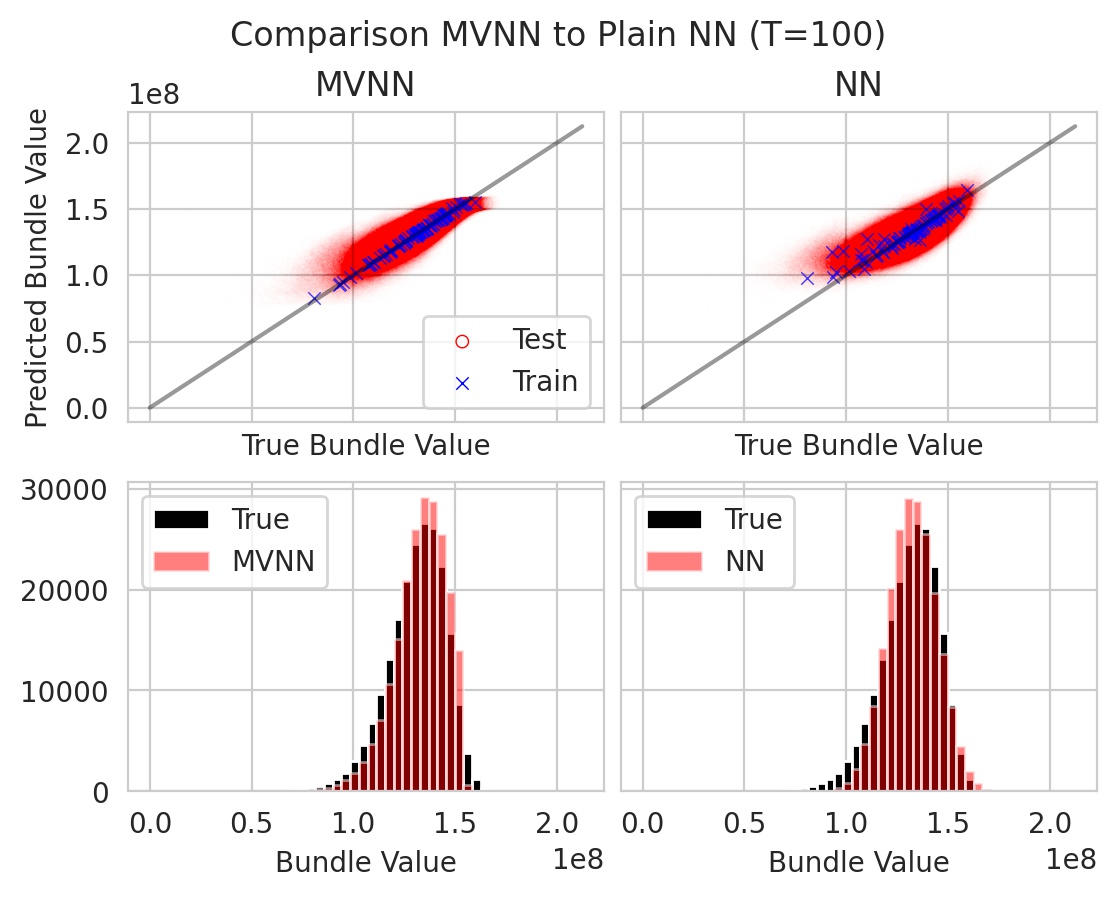} &   \\
 \bottomrule
\end{tabular}
}
\end{sc}
\vskip -0.2cm
\caption{Prediction performance comparison for selected MVNNs (\textsc{MVNN-ReLU-Projected}) and plain NNs from Table~\ref{tab:full_pred_performance_table} across all domains and bidder types. The top plots compare the quality of the prediction (training data as blue crosses), while the bottom plots focus on how well the models capture the overall value distribution.}
\label{tab:pred_perf_plots_appendix}
\end{table*}
\vspace{-10cm}
\restoregeometry

\end{document}